\documentclass[journal,onecolumn]{IEEEtran}
\usepackage{booktabs}

\usepackage{amsmath}                
\usepackage{amsthm}                 
\usepackage{amssymb}
\usepackage{thmtools}
\usepackage{kpfonts}
\usepackage[geometry]{ifsym}
\usepackage{dsfont}
\usepackage{multirow}
\usepackage[mathscr]{euscript}
\usepackage{graphicx}
\usepackage{color}
\usepackage[inline]{enumitem}
\usepackage{framed}
\usepackage{float}
\usepackage{longtable}
\usepackage{cite}
\usepackage{enumitem}
\usepackage{caption} 
\usepackage{subcaption} 
\usepackage{circuitikz}
\usepackage{tikz}

\usetikzlibrary{patterns,decorations.pathmorphing,positioning, decorations.markings}
\usepackage{xcolor}

\usepackage[colorlinks=true, citecolor=blue, linkcolor=blue]{hyperref}       
\usepackage[font=itshape]{quoting}
\usepackage{lscape}
\usepackage{makecell}
\usepackage[section]{placeins}

\usepackage{listings}
\usepackage{comment}
\usepackage{framed} 
\usepackage{nomencl} 
\makenomenclature

\declaretheoremstyle[%
  headfont=\bfseries,%
  headpunct={:},%
  notefont=\normalfont\bfseries,%
  notebraces={--~}{},
    qed=$\blacksquare$,
]{definitionstyle}
\theoremstyle{definition}
\declaretheorem[style=definitionstyle,name=Definition]{defn}
\declaretheorem[style=definitionstyle,name=Theorem]{thm}

\theoremstyle{definition}

\theoremstyle{plain}
\theoremstyle{remark}

\begin{document}
%
\title{Generalizing Linear Graphs and Bond Graph Models with Hetero-functional Graphs for System-of-Systems Engineering Applications}
\author{Ehsanoddin Ghorbanichemazkati, Amro M. Farid
\thanks{E. Ghorbanichemazkati is a doctoral research assistant with the Department of Systems and Enterprises at the Stevens Insititute of Technology, Hoboken NJ 07030}
\thanks{A.M. Farid is the Alexander Crombie Humphreys Chair Professor of Economics in Engineering in Department of Systems and Enterprises at the Stevens Insititute of Technology, Hoboken NJ 07030.  He is also the Principal Systems Scientist at CSIRO Smart Energy (Newcastle, Australia) and a Visiting Scientist at MIT Mechanical Engineering Cambridge, MA.}
}

\date{March 28 2024}
\maketitle

\begin{abstract}
In the 20th century, individual technology products like the generator, telephone, and automobile were connected to form many of the large-scale, complex, infrastructure networks we know today: the power grid, the communication infrastructure, and the transportation system. Progressively, these networked systems began interacting, forming what is now known as systems-of-systems. Because the component systems in the system-of-systems differ, modeling and analysis techniques with primitives applicable across multiple domains or disciplines are needed. For example, linear graphs and bond graphs have been used extensively in the electrical engineering, mechanical engineering, and mechatronic fields to design and analyze a wide variety of engineering systems. In contrast, hetero-functional graph theory (HFGT) has emerged to study many complex engineering systems and systems-of-systems (e.g. electric power, potable water, wastewater, natural gas, oil, coal, multi-modal transportation, mass-customized production, and personalized healthcare delivery systems).  This paper seeks to relate hetero-functional graphs to linear graphs and bond graphs and demonstrate that the former is a generalization of the latter two.  The contribution is relayed in three stages.  First, the three modeling techniques are compared conceptually.  Next, these techniques are contrasted on six example systems: (a) an electrical system, (b) a translational mechanical system, (c) a rotational mechanical system, (d) a fluidic system, (e) a thermal system, and (f) a multi-energy (electro-mechanical) system.  Finally, this paper proves mathematically that hetero-functional graphs are a formal generalization of both linear graphs and bond graphs.
\end{abstract}

\section{Introduction}
In the 20th century, individual technology products like the generator, telephone, and automobile were connected to form many of the large-scale, complex, infrastructure networks we know today: the power grid, the communication infrastructure, and the transportation system\cite{De-Weck:2011:00}. Over time, these networked systems began to develop interactions between themselves in what is now called systems-of-systems\cite{Little:2019:00,Blokdyk:2018:00, Delaurentis:2022:00, Jamshidi:2017:00, Mann:2015:00, Rainey:2022:00}. The ``smart grid”\cite{Annaswamy:2013:SPG-BK02}, the energy-water nexus\cite{Olsson:2015:01,Lubega:2014:EWN-J11}, the electrification of transport\cite{Anair:2012:00,
Pasaoglu:2012:00, Karabasoglu:2013:00, Raykin:2012:00, Yang:2012:00}, are all good examples where one network system has fused with another to form a new and much more capable system.  This trend is only set to continue. The energy-water-food nexus\cite{Albrecht:2018:00} fuses three such systems and the recent interest in smart cities\cite{Farid:2021:ISC-J41,Farid:2021:ISC-J42,Cocchia:2014:00} provides a platform upon which to integrate all of these efforts. 

Because the component systems in the system-of-systems are unlike each other, there is a need to use modeling and analysis techniques that have modeling primitives that can be applied across multiple domains or disciplines.  For example, linear graphs\cite{Rowell:1997:00} and bond graphs\cite{Karnopp:1990:00} have been used extensively in the electrical engineering, mechanical engineering, and mechatronic fields to design and analyze a wide variety of engineering systems.  They use modeling primitives such as generalized resistors, capacitors, inductors, transformers, and gyrators to address mechanical, electrical, fluidic, and thermal systems. Unfortunately, the application of linear graphs and bond graphs are limited to systems where system elements are connected via flows of power (of various types).

In contrast, hetero-functional graph theory (HFGT)\cite{Schoonenberg:2019:ISC-BK04,Farid:2022:ISC-J51,Farid:2016:ISC-BC06} has emerged to study many complex engineering systems and systems-of-systems.  More specifically, HFGT provides a means of algorithmically translating SysML models\cite{Delligatti:2014:00,Friedenthal:2014:00,Weilkiens:2007:00} into hetero-functional graphs and/or Petri Nets\cite{Girault:2013:00} where they can be structurally analyzed\cite{Thompson:2021:SPG-J46,Thompson:2024:ISC-J55}, dynamically simulated\cite{Khayal:2021:ISC-J48}, and ultimately optimized\cite{Schoonenberg:2022:ISC-J50}.  HFGT has been applied to numerous application domains including electric power, potable water, wastewater, natural gas, oil, coal, multi-modal transportation, mass-customized production, and personalized healthcare delivery systems.  In so doing, HFGT has demonstrated its ability to model the supply, demand, transportation, storage, transformation, assembly, and disassembly of multiple operands in distinct locations over time\cite{Schoonenberg:2022:ISC-J50}.  These multiple operands include matter, energy, information, money, and living organisms\cite{Schoonenberg:2019:ISC-BK04,Farid:2022:ISC-J51} and not just energy as in the case of linear graphs and bond graphs.  

\subsection{Original Contribution}
Consequently, this paper seeks to relate hetero-functional graphs to linear graphs and bond graphs and demonstrate that the former is a generalization of the latter two.  The contribution is relayed in three stages.  First, the three modeling techniques are compared conceptually.  Next, the three modeling techniques are contrasted on six example systems: (a) an electrical system, (b) a translational mechanical system, (c) a rotational mechanical system, (d) a fluidic system, (e) a thermal system, and (f) a multi-energy (electro-mechanical) system.  Finally, this paper proves mathematically that hetero-functional graphs are a formal generalization of both linear graphs and bond graphs.  

To facilitate the discussion, several assumptions and limitations are made.  
\begin{itemize}
\item As the majority of the literature for all three modeling approaches concerns lumped parameter models, this paper restricts its scope to such systems.  
\item So as to facilitate the discussion, all physical systems will be modeled with ``power-variables"\cite{Rowell:1997:00,Karnopp:1990:00} -- pairs of physical variables whose product equals a quantity of power. 
\item Linear graphs and bond graphs are assumed to describe the flows of power of various types:  translational mechanical, rotational mechanical, electrical, fluidic, and thermal.  This has been the typical application of linear graphs and bond graphs in the literature.  
\item Without loss of generality, and for simplicity of discussion, this paper restricts its discussion to linear constitutive laws.  Non-linear constitutive laws can be readily integrated into bond graphs and hetero-functional graphs.  
\item Without loss of generality, and also for simplicity of discussion, this paper restricts itself to the following types of elements (as understood in bond graphs):  
\begin{enumerate*}
\item Effort sources/sinks, 
\item Flow sources/sinks, 
\item Generalized resistors, 
\item Generalized capacitors, 
\item Generalized inductors, 
\item Generalized transformers, and
\item Generalized gyrators.
\end{enumerate*}
While other types of elements have been introduced in bond graphs \cite{Brown:2006:00}, their exclusion does not detract from the original contribution presented here and can be re-introduced straightforwardly.  
\end{itemize}
\subsection{Paper Outline}
The remainder of the paper proceeds as follows.  Sec. \ref{sec:Overview_on_methods} provides an overview of linear graphs, bond graphs, and hetero-functional graph techniques.  Sec. \ref{sec:examples} then introduces six illustrative examples that serve as the basis for comparison.  Sections \ref{subsec:Linear_graph_by_example}, \ref{subsec:Bond_graph_by_example}, and \ref{subsec:HFGT_by_example} then demonstrate the linear graph, bond graph, and hetero-functional graph methodologies on each of these six systems.  Sec. \ref{sec:On_generality_of_HFGT} then proves mathematically that hetero-functional graphs are a formal generalization of both linear graphs and bond graphs.  Sec. \ref{sec:Conclusion} brings the paper to a conclusion.  
\section{Overview of Linear Graph, Bond Graph, and Hetero-functional Graph Techniques}
\label{sec:Overview_on_methods}
In order to begin to relate hetero-functional graphs to linear graphs and bond graphs, all three approaches must be briefly described at a conceptual level.    Interestingly, linear graphs share more common features with bond graphs and hetero-functional graphs than the other two with each other.  Consequently, Sections \ref{subsec:linear_graphs}, \ref{subsec:bond_graphs} and \ref{subsec:HFGT} describe linear graphs, bond graphs, and hetero-functional in that order to facilitate discussion. Figure \ref{fig:different_graphs} serves to guide the discussion for the remainder of the section.

\begin{figure}
\centering
\includegraphics[width=0.85\textwidth]{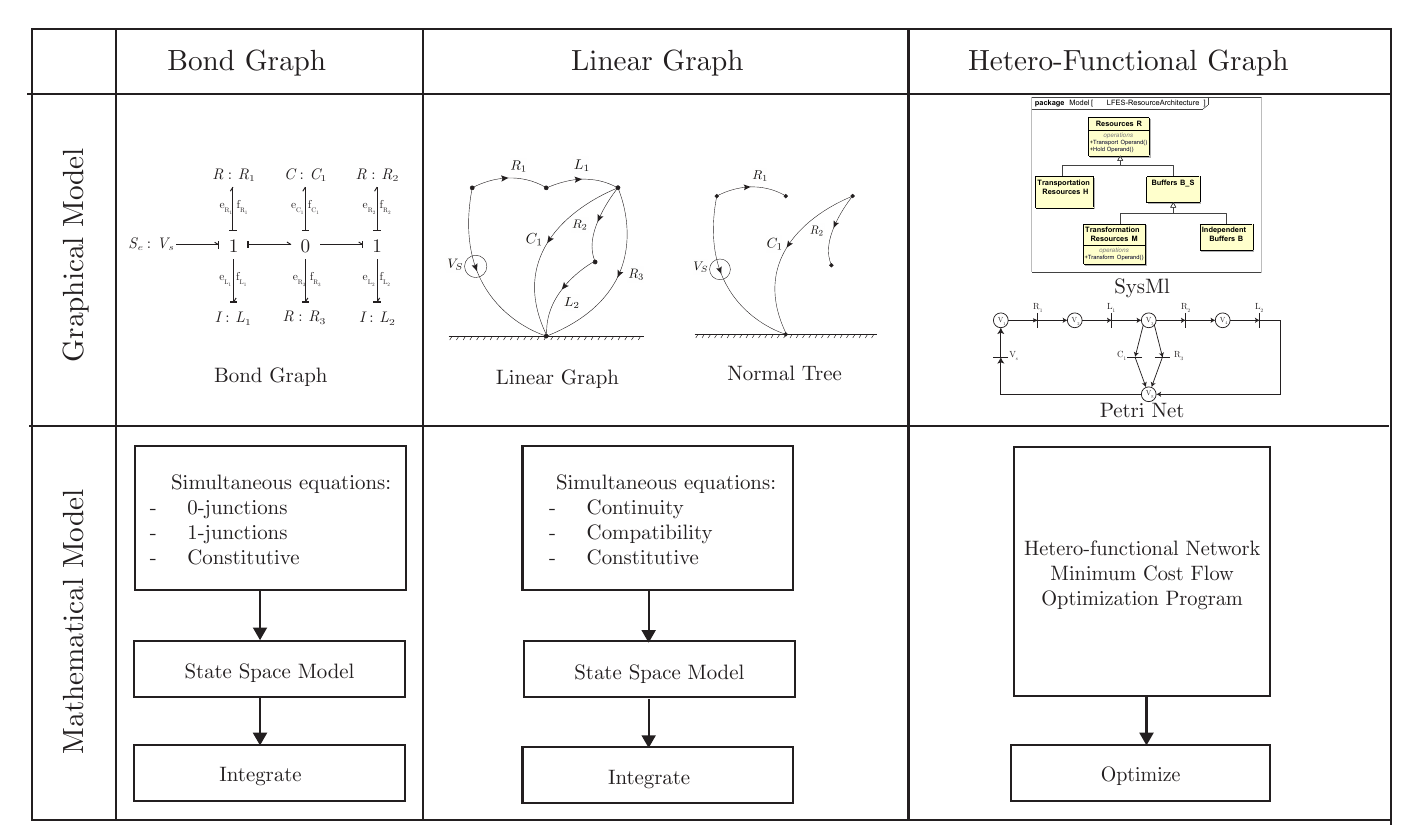}
\caption{Graphical and mathematical models of Bond Graph, Linear Graph, and Hetero-functional Graph.}
\label{fig:different_graphs}
\end{figure}

\subsection{Linear Graphs}
\label{subsec:linear_graphs}
As shown in Fig. \ref{fig:different_graphs}, linear graphs are a type of graphical model that shows the interconnectedness of lumped parameter elements depicted as arcs\cite{Rowell:1997:00} that transform and transport the flow of power.  It is important to recognize that when defining the linear graph's lumped parameter elements, the system modeler is implicitly choosing either a Lagrangian or an Eulerian view\cite{White:1994:00} of the system.  Each of these power flows in the linear graph's arcs is associated with a pair of variables; one ``across" variable and another ``through" variable.  Fig. \ref{fig:across_through_variables} shows the across and through variables for each of the five energy domains.  Furthermore, as shown in Fig. \ref{fig:graph_elements}, each of the arcs in a linear graph can be classified into one of several different types. These arcs connect with each other at nodes that are associated with points in space that have distinct across-variable values measured relative to a well-chosen absolute reference frame.  Once a linear graph has been constructed it is then translated into a ``normal tree" so as to eliminate redundant variables\cite{Rowell:1997:00}. From there, three sets of simultaneous differential and algebraic equations are derived to form a mathematical model of the system.  
\begin{itemize}
\item \textbf{Continuity laws:} For each node in the normal tree, a continuity law (e.g. Newton's 1st Law, Kirchhoff's Current law) is derived that describes the conservation of the relevant through variables.  
\item \textbf{Constitutive laws:} For each arc in the normal tree, a constitutive law (e.g. Newton's 2nd Law, Ohm's Law) is derived that describes the relationship between across and through variables in the lumped parameter element.  
\item \textbf{Compatibility law:} For each arc in the normal tree, a compatibility law is derived to relate the across variable of the lumped parameter element to the across variables of the linear graph's nodes. 
 In most cases, the compatibility laws can be entirely omitted if the simultaneous equations are expressed exclusively in terms of the (absolute) across variables at each of the nodes.  
\end{itemize}
The three sets of simultaneous equations that constitute the mathematical model can then be algebraically simplified into a state space model of the form:  
\begin{align}
\dot{X}= AX + BU + E\dot{U} \label{eq:general_state_eq1}\\
y = CX + DU + F\dot{U} \label{eq:general_state_eq2}
\end{align}
where $X$ is the system state vector, $U$ are the system inputs, $y$ are the system outputs, and $A,B,C,D,E$ and $F$ are constant parameter coefficient matrices of appropriate size.  The size of the system state vector is determined by the number of independent energy storage elements (as determined from the normal tree)\cite{Rowell:1997:00}.  The size of the system input vector is determined from the number of source/sink elements, and the size of the output vector is left to the modeler's discretion.  Finally, the state space model is simulated in the time domain by numerical integration.

\begin{figure}[htbp]
\centering
\includegraphics[width=0.9\linewidth]{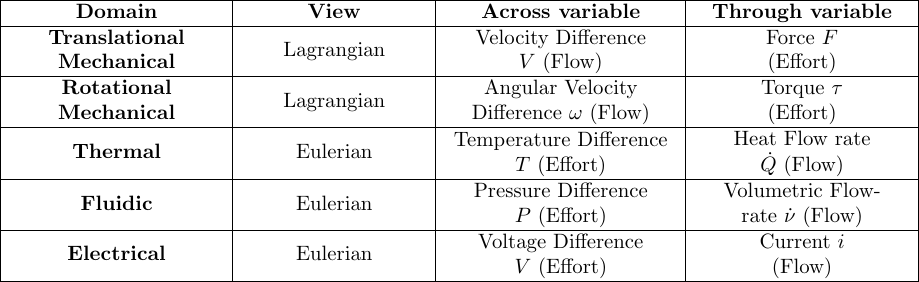}
\caption{Comparison of physical domains with respective Lagrangian and Eulerian views, illustrating the across and through, flow and effort variables associated with each domain.}
\label{fig:across_through_variables} 
\end{figure}

\begin{figure}[htbp]
\centering
\includegraphics[width=0.9\linewidth]{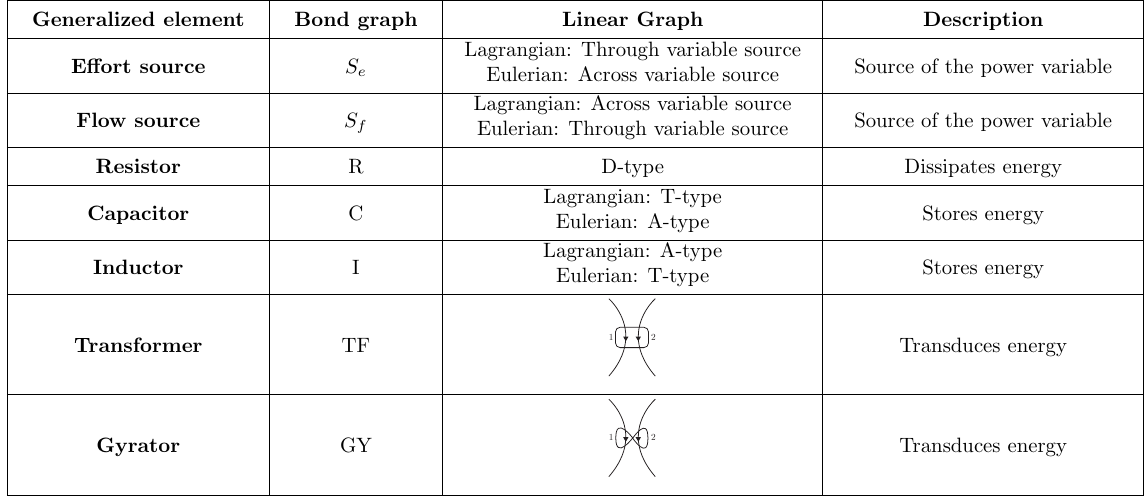}
\caption{Physical elements used in bond graphs and linear graphs}
\label{fig:graph_elements} 
\end{figure}

\subsection{Bond Graphs}
\label{subsec:bond_graphs}
Bond graph models are derived in a similar manner; although with several key differences.  As shown in Fig. \ref{fig:different_graphs}, bond graphs are a type of graphical model that shows the interconnectedness of lumped parameter elements depicted as labeled nodes \cite{Karnopp:1990:00} transform and transport power.  Each of these power flows is associated with a pair of variables; one ``effort" variable and another ``flow" variable.  Fig. \ref{fig:across_through_variables} shows the effort and flow variables for each of the five energy domains.  Importantly, when the system model adopts a Lagrangian view, the effort and flow variables map to through and across variables respectively.  In contrast, when the system model adopts an Eulerian view, the effort and flow variables map to across and through variables respectively.  Furthermore, as shown in Fig. \ref{fig:graph_elements}, each of the elements in the bond graph can be classified into one of several different types.  In addition to these bond graph elements, the bond graph also includes ``0-Junctions" and ``1-Junctions".  0-Junctions, or flow junctions, conserve the sum of all flow variables and have a common associated effort variable.  Meanwhile, 1-Junctions, or effort junctions, conserve the sum of all effort variables and have a common associated flow variable.  Finally, the 0-Junctions and the 1-Junctions are connected to the bond graph elements with arcs (i.e. power bonds) that describe a flow of power between elements and junctions.  Once the bond graph has been constructed, three sets of simultaneous differential-algebraic equations are derived to form a mathematical model of the system.  
\begin{itemize}
\item \textbf{0-Junction Laws:} For each 0-Junction, the flow conversation law is derived.  For systems with a Lagrangian view, these 0-Junction laws are compatibility laws.  For systems with an Eulerian view, these 0-junction laws are continuity laws (e.g. Kirchhoff's current law).  
\item \textbf{1-Junction Laws:} For each 1-Junction, the effort conservation law is derived.  For systems with a Lagrangian view, these 1-Junction laws are continuity laws (e.g. Newton's 1st law).  For systems with an Eulerian view, these 1-junction laws are compatibility laws (e.g. Kirchhoff's voltage law).  
\item \textbf{Constitutive laws:} For each element in the bond graph, a constitutive law (e.g. Newton's 2nd Law, Ohm's Law) is derived that describes the relationship between effort and flow variables in the lumped parameter element.  
\end{itemize}
As with linear graphs, the three sets of simultaneous equations that constitute the mathematical model can then be algebraically simplified into a state space model shown in Eq. \ref{eq:general_state_eq1} and \ref{eq:general_state_eq2}.  

\subsection{Hetero-functional Graphs}\label{subsec:HFGT}
As shown in Fig. \ref{fig:different_graphs}, hetero-functional graphs are also a type of graphical models that show the interconnectedness of lumped parameter elements.  Unlike the relatively specific graphical ontologies used in linear graphs and bond graphs, Hetero-functional graph theory stems from the universal structure of human language with subjects and predicates and the latter made up of verbs and objects\cite{Schoonenberg:2019:ISC-BK04,Farid:2022:ISC-J51}  It includes set of system resources $R$ as subjects, a set of system processes $P$ as predicates, and a set of operands $L$ as their constituent objects.  
\begin{defn}[System Operand \cite{SE-Handbook-Working-Group:2015:00}]\label{Defn:D1}
An asset or object $l_i \in L$ that is operated on or consumed during the execution of a process.
\end{defn}
\begin{defn}[System Process\cite{Hoyle:1998:00,SE-Handbook-Working-Group:2015:00}]\label{def:CH7:process}
An activity $p \in P$ that transforms or transports a predefined set of input operands into a predefined set of outputs. 
\end{defn}
\begin{defn}[System Resource \cite{SE-Handbook-Working-Group:2015:00}]
An asset or object $r_v \in R$ that facilitates the execution of a process.  
\end{defn}
\noindent As shown in Fig. \ref{fig:LFESMetaArchitecture}, these operands, processes, and resources are organized in an engineering system meta-architecture stated in the Systems Modeling Language (SysML)\cite{Delligatti:2014:00,Friedenthal:2014:00,Weilkiens:2007:00}.  
\begin{figure}[htbp]
\centering
\includegraphics[width=\textwidth]{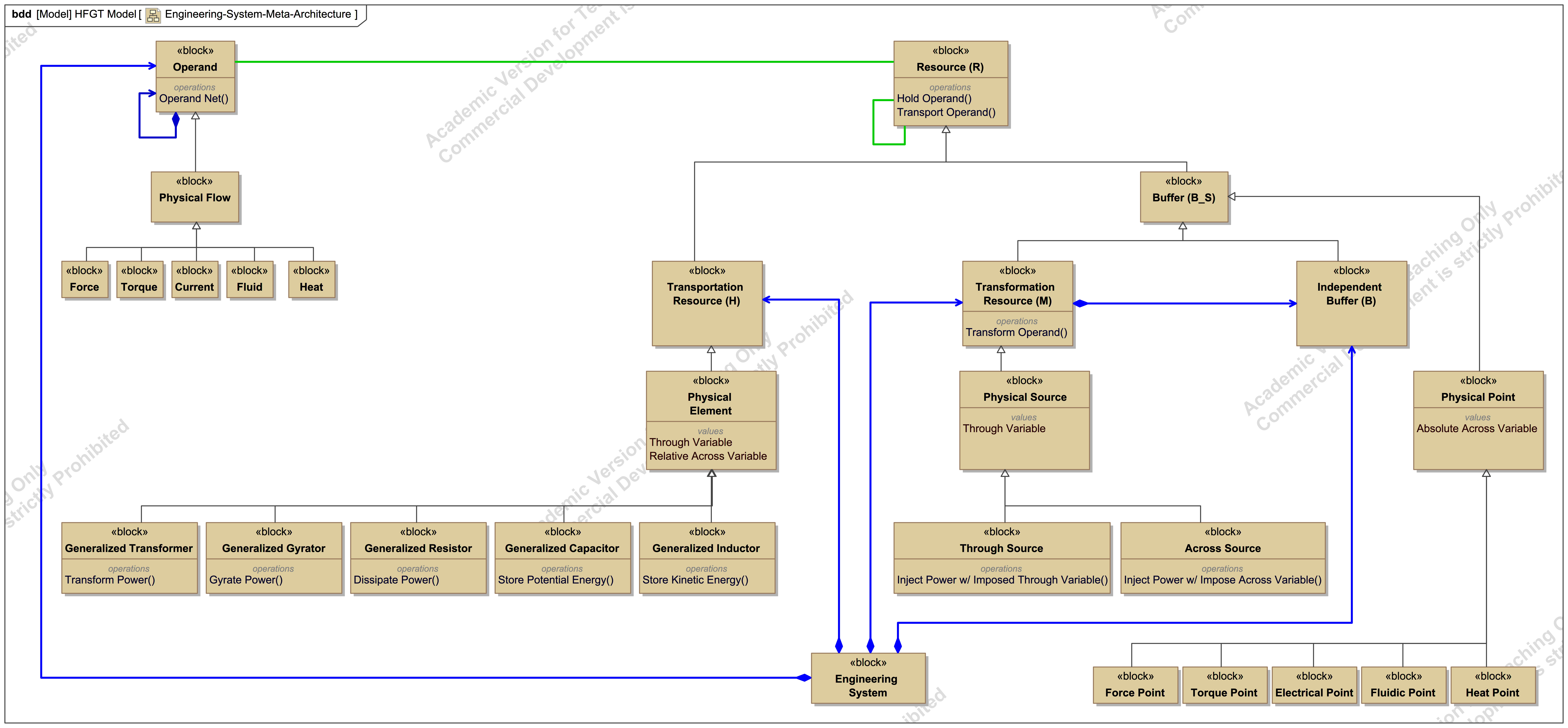}
\caption{A SysML Block Definition Diagram of the System Form of the Engineering System Meta-Architecture}
\label{fig:LFESMetaArchitecture}
\end{figure}

Importantly, the system resources $R=M \cup B \cup H$ are classified into transformation resources $M$, independent buffers $B$, and transportation resources $H$.  Additionally, the set of ``buffers" $B_S=M \cup B$ is introduced to support the discussion.  Equally important, the system processes $P = P_\mu \cup P_{\bar{\eta}}$ are classified into transformation processes $P_\mu$ and refined transportation processes $P_\eta$.  The latter arises from the simultaneous execution of one transportation process and one holding process.  Finally, hetero-functional graph theory emphasizes that resources are capable of one or more system processes to produce a set of ``capabilities"\cite{Schoonenberg:2019:ISC-BK04}.
\begin{defn}[Buffer\cite{Schoonenberg:2019:ISC-BK04,Farid:2022:ISC-J51}]\label{defn:BSCh7}
A resource $r_v \in R$ is a buffer $b_s \in B_S$ iff it is capable of storing or transforming one or more operands at a unique location in space.  
\end{defn}
\begin{defn}[Capability\cite{Schoonenberg:2019:ISC-BK04,Farid:2022:ISC-J51,Farid:2016:ISC-BC06}]\label{defn:capabilityCh7}
An action $e_{wv} \in {\cal E}_S$ (in the SysML sense) defined by a system process $p_w \in P$ being executed by a resource $r_v \in R$.  It constitutes a subject + verb + operand sentence of the form: ``Resource $r_v$ does process $p_w$".  
\end{defn}
\noindent The highly generic and abstract nature of these definitions has allowed HFGT to be applied to numerous application domains including electric power, potable water, wastewater, natural gas, oil, coal, multi-modal transportation, mass-customized production, and personalized healthcare delivery systems.  For a more in-depth description of HFGT, readers are directed to past works\cite{Schoonenberg:2019:ISC-BK04,Farid:2022:ISC-J51,Farid:2016:ISC-BC06}. 

Fig. \ref{fig:LFESMetaArchitecture} also serves to relate the hetero-functional graph theory definitions to the ontological elements found in linear graphs and bond graphs.  
\begin{enumerate}
\item First, linear graphs and bond graphs are concerned with physical flows, and more specifically the flows of force, torque, electrical current, fluids, and heat.  Consequently, Fig. \ref{fig:LFESMetaArchitecture} shows these objects as types of operands.
\item Second, linear graphs and bond graphs implicitly define a set of physical points at which there are distinct values of across-value attributes measured in absolute terms relative to a well-chosen reference frame.  Again, Fig. \ref{fig:across_through_variables} summarizes the different types of physical points by the type of across variable.  Consequently, Fig. \ref{fig:LFESMetaArchitecture} shows these as types of independent buffers.  
\item Third, linear graphs and bond graphs are composed of physical sources.  Because the across variables at physical points are types of independent buffers (in HFGT), across-variable and through-variable sources are adopted into the taxonomy (rather than effort and flow sources).  Because these physical sources inject power across the system boundary as system processes, and all system processes that inject operands across the system boundary are transformation processes\cite{Schoonenberg:2019:ISC-BK04}, then Fig. \ref{fig:LFESMetaArchitecture} shows these physical sources as transformation resources.  Furthermore, because they are transformation resources, they are also buffers and inherently have an absolute across variable associated with physical points.  
\item Fourth, linear graphs and bond graphs are composed of the physical elements shown in Fig. \ref{fig:graph_elements}.  Each of these is associated with a through-variable attribute and an across-variable attribute defined relative to the across variables of the physical points (as independent buffers).  They are further classified as generalized resistors, capacitors, inductors, transformers, and gyrators (as defined in bond graphs).  They dissipate power, store potential energy in an effort variable, store energy in a flow variable, transform power, and gyrate power respectively.  (For brevity, the storage of energy in a flow variable will be referred to as storing kinetic energy; although this is a strained analogy in the electrical domain).  As each of these system processes is a flow of power between two physical points, Fig. \ref{fig:LFESMetaArchitecture} shows all of these physical elements as transportation resources.  
\end{enumerate}
Finally, it is important to recognize that linear graphs and bond graphs are limited to physical sources and physical elements with only a single associated capability.  In contrast, HFGT is able to address engineering systems with resources that have an arbitrary number of processes.  In this regard, and as shown in Fig. \ref{fig:LFESMetaArchitecture}, HFGT is more ontologically rich than both linear graphs and bond graphs.

Returning to Fig. \ref{fig:different_graphs}, the engineering system meta-architecture stated in SysML must be instantiated and ultimately transformed into the associated Petri net model(s). To that end, the positive and negative hetero-functional incidence tensors (HFIT) are introduced to describe the flow of operands through buffers and capabilities.  
\begin{defn}[The Negative 3$^{rd}$ Order Hetero-functional Incidence Tensor (HFIT) $\widetilde{\cal M}_\rho^-$\cite{Farid:2022:ISC-J51}]\label{Defn:D6}
The negative hetero-functional incidence tensor $\widetilde{\cal M_\rho}^- \in \{0,1\}^{|L|\times |B_S| \times |{\cal E}_S|}$  is a third-order tensor whose element $\widetilde{\cal M}_\rho^{-}(i,y,\psi)=1$ when the system capability ${\epsilon}_\psi \in {\cal E}_S$ pulls operand $l_i \in L$ from buffer $b_{s_y} \in B_S$.
\end{defn} 

\begin{defn}[The Positive  3$^{rd}$ Order Hetero-functional Incidence Tensor (HFIT)$\widetilde{\cal M}_\rho^+$\cite{Farid:2022:ISC-J51}]
The positive hetero-functional incidence tensor $\widetilde{\cal M}_\rho^+ \in \{0,1\}^{|L|\times |B_S| \times |{\cal E}_S|}$  is a third-order tensor whose element $\widetilde{\cal M}_\rho^{+}(i,y,\psi)=1$ when the system capability ${\epsilon}_\psi \in {\cal E}_S$ injects operand $l_i \in L$ into buffer $b_{s_y} \in B_S$.
\end{defn}
\noindent These incidence tensors are straightforwardly ``matricized" to form 2$^{nd}$ Order Hetero-functional Incidence Matrices $M = M^+ - M^-$ with dimensions $|L||B_S|\times |{\cal E}|$. Consequently, the supply, demand, transportation, storage, transformation, assembly, and disassembly of multiple operands in distinct locations over time can be described by an Engineering System Net and its associated State Transition Function\cite{Schoonenberg:2022:ISC-J50}.

\begin{defn}[Engineering System Net\cite{Schoonenberg:2022:ISC-J50}]\label{Defn:ESN}
An elementary Petri net ${\cal N} = \{S, {\cal E}_S, \textbf{M}, W, Q\}$, where
\begin{itemize}
    \item $S$ is the set of places with size: $|L||B_S|$,
    \item ${\cal E}_S$ is the set of transitions with size: $|{\cal E}|$,
    \item $\textbf{M}$ is the set of arcs, with the associated incidence matrices: $M = M^+ - M^-$,
    \item $W$ is the set of weights on the arcs, as captured in the incidence matrices,
    \item $Q=[Q_B; Q_E]$ is the marking vector for both the set of places and the set of transitions. 
\end{itemize}
\end{defn}

\begin{defn}[Engineering System Net State Transition Function\cite{Schoonenberg:2022:ISC-J50}]\label{Defn:ESN-STF}
The  state transition function of the engineering system net $\Phi()$ is:
\begin{equation}\label{CH6:eq:PhiCPN}
Q[k+1]=\Phi(Q[k],U^-[k], U^+[k]) \quad \forall k \in \{1, \dots, K\}
\end{equation}
where $k$ is the discrete time index, $K$ is the simulation horizon, $Q=[Q_{B}; Q_{\cal E}]$, $Q_B$ has size $|L||B_S| \times 1$, $Q_{\cal E}$ has size $|{\cal E}_S|\times 1$, the input firing vector $U^-[k]$ has size $|{\cal E}_S|\times 1$, and the output firing vector $U^+[k]$ has size $|{\cal E}_S|\times 1$.  
\begin{align}\label{CH6:CH6:eq:Q_B:HFNMCFprogram}
Q_{B}[k+1]&=Q_{B}[k]+{M}^+U^+[k]\Delta T-{M}^-U^-[k]\Delta T \\ \label{CH6:CH6:eq:Q_E:HFNMCFprogram}
Q_{\cal E}[k+1]&=Q_{\cal E}[k]-U^+[k]\Delta T +U^-[k]\Delta T
\end{align}
where $\Delta T$ is the duration of the simulation time step.  
\end{defn}

Here, it is important to recognize that the engineering system net state transition function is an explicit restatement of the continuity laws in linear graphs.  Similarly, the engineering system net describes the 0-Junction laws in bond graphs that use an Eulerian view and describes the 1-Junction laws in bond graphs that use a Lagrangian view.  For this reason, the relationship between hetero-functional graphs and linear graphs is more straightforward than between hetero-functional graphs and bond graphs.  

In addition to the engineering system net, in HFGT, each operand can have its own state and evolution.  This behavior is described by an Operand Net and its associated State Transition Function for each operand.  
\begin{defn}[Operand Net\cite{Farid:2008:IEM-J04,Schoonenberg:2019:ISC-BK04,Khayal:2017:ISC-J35,Schoonenberg:2017:IEM-J34}]\label{Defn:OperandNet} Given operand $l_i$, an elementary Petri net ${\cal N}_{l_i}= \{S_{l_i}, {\cal E}_{l_i}, \textbf{M}_{l_i}, W_{l_i}, Q_{l_i}\}$ where 
\begin{itemize}
\item $S_{l_i}$ is the set of places describing the operand's state.  
\item ${\cal E}_{l_i}$ is the set of transitions describing the evolution of the operand's state.
\item $\textbf{M}_{l_i} \subseteq (S_{l_i} \times {\cal E}_{l_i}) \cup ({\cal E}_{l_i} \times S_{l_i})$ is the set of arcs, with the associated incidence matrices: $M_{l_i} = M^+_{l_i} - M^-_{l_i} \quad \forall l_i \in L$.  
\item $W_{l_i} : \textbf{M}_{l_i}$ is the set of weights on the arcs, as captured in the incidence matrices $M^+_{l_i},M^-_{l_i} \quad \forall l_i \in L$.  
\item $Q_{l_i}= [Q_{Sl_i}; Q_{{\cal E}l_i}]$ is the marking vector for both the set of places and the set of transitions. 
\end{itemize}
\end{defn}

\begin{defn}[Operand Net State Transition Function\cite{Farid:2008:IEM-J04,Schoonenberg:2019:ISC-BK04,Khayal:2017:ISC-J35,Schoonenberg:2017:IEM-J34}]\label{Defn:OperandNet-STF}
The  state transition function of each operand net $\Phi_{l_i}()$ is:
\begin{equation}\label{CH6:eq:PhiSPN}
Q_{l_i}[k+1]=\Phi_{l_i}(Q_{l_i}[k],U_{l_i}^-[k], U_{l_i}^+[k]) \quad \forall k \in \{1, \dots, K\} \quad i \in \{1, \dots, L\}
\end{equation}
where $Q_{l_i}=[Q_{Sl_i}; Q_{{\cal E} l_i}]$, $Q_{Sl_i}$ has size $|S_{l_i}| \times 1$, $Q_{{\cal E} l_i}$ has size $|{\cal E}_{l_i}| \times 1$, the input firing vector $U_{l_i}^-[k]$ has size $|{\cal E}_{l_i}|\times 1$, and the output firing vector $U^+[k]$ has size $|{\cal E}_{l_i}|\times 1$.  

\begin{align}\label{X}
Q_{Sl_i}[k+1]&=Q_{Sl_i}[k]+{M_{l_i}}^+U_{l_i}^+[k]\Delta T - {M_{l_i}}^-U_{l_i}^-[k]\Delta T \\ \label{CH6:CH eq:Q_E:HFNMCFprogram}
Q_{{\cal E} l_i}[k+1]&=Q_{{\cal E} l_i}[k]-U_{l_i}^+[k]\Delta T +U_{l_i}^-[k]\Delta T
\end{align}
\end{defn}
Here, it is important to recognize that although HFGT introduces operand nets and their respective state transition functions, linear graphs and bond graphs do not have a counterpart modeling concept.  This is because when power flows as an operand, regardless of whether it is mechanical, electrical, fluidic, or thermal power, it does not change state and therefore does not require an operand net.  In contrast, other application domains, most notably production systems\cite{Schoonenberg:2017:IEM-J34,Farid:2008:IEM-J04,Farid:2008:IEM-J05} and healthcare systems\cite{Khayal:2021:ISC-J48,Khayal:2017:ISC-J35,Khayal:2015:ISC-J20} respectively have products and patients as operands with often very complex operand state evolution.  

Returning to Fig. \ref{fig:different_graphs}, HFGT describes the behavior of an engineering system using the Hetero-Functional Network Minimum Cost Flow (HFNMCF) problem\cite{Schoonenberg:2022:ISC-J50}.  Whereas linear graphs and bond graphs models insist on the state space form in Eq. \ref{eq:general_state_eq1}-\ref{eq:general_state_eq2}, HFGT similarly insists on the HFNMCF mathematical program.  It optimizes the time-dependent flow and storage of multiple operands (or commodities) between buffers, allows for their transformation from one operand to another, and tracks the state of these operands.  In this regard, it is a very flexible optimization problem that applies to a wide variety of complex engineering systems.  For the purposes of this paper, the HFNMCF is a type of discrete-time-dependent, time-invariant, convex optimization program\cite{Schoonenberg:2022:ISC-J50}.

\vspace{0.2in}
\begin{align}\label{Eq:ObjFunc1}
\text{minimize } Z &= \sum_{k=1}^{K-1} f_k(x[k],y[k]) \\ \label{Eq:EqualityConstraints}
\text{s.t. } A_{CP}X &= B_{CP} \\ \label{ch6:eq:QPcanonicalform:3}
\underline{E}_{CP} \leq D(X) &\leq \overline{E}_{CP} \\ \label{Eq:DeviceModels}
g(X,Y) &= 0 \\ \label{Eq:DevicModels2}
h(Y) &\leq 0
\end{align}

where 
\begin{itemize}
\item $Z$ is a convex objective function separable in $k$.
\item $k$ is the discrete time index. 
\item $K$ is the simulation horizon.
\item $f_k()$ is a set of discrete-time-dependent convex functions.
\item $X=\left[x[1]; \ldots; x[K]\right]$  is the vector of primary decision variables at time $k$.
\begin{equation}\label{eq:primary_decision_var}
x[k] = \begin{bmatrix} Q_B ; Q_{\cal E} ; Q_{SL} ; Q_{{\cal E}L} ; U^- ; U^+ ; U^-_L ; U^+_L \end{bmatrix}[k] \quad \forall k \in \{1, \dots, K\}
\end{equation}
\item $Y=\left[y[1]; \ldots; y[K]\right]$  is the vector of auxiliary decision variables at time $k$.  The need for auxiliary decision variables depends on the presence and nature of the device models $g(X,Y)=0$ in Eq. \ref{Eq:DeviceModels} and $h(Y)\leq 0$ in Eq. \ref{Eq:DevicModels2}.  
\item $A_{CP}$ is the linear equality constraint coefficient matrix.
\item $B_{CP}$ is the linear equality constraint intercept vector.
\item $D_{CP}$ is the linear inequality constraint coefficient matrix.
\item $E_{CP}$ is the linear inequality constraint intercept vector.  
\item g(X,Y) and h(Y) is a set of device model functions whose presence and nature depend on the specific problem application.  
\end{itemize}

Despite the terse description of the HFNMCF problem presented above, it has immediate relationships to linear graphs and bond graphs.  
\begin{itemize}
\item The through variables in a linear graph appear amongst the primary decision variables X.  
\item The across variables in a linear graph appear amongst the auxiliary variables Y.
\item Because HFGT assumes that the across variables are stated in absolute terms, the compatibility laws in a linear graph are not required in the HFNMCF.  
\item The continuity relations in a linear graph appear amongst the linear equality constraints in Eq. \ref{Eq:EqualityConstraints}
\item The constitutive relations in a linear graph appear amongst the device model constraints in Eq. \ref{Eq:DeviceModels}.  
\end{itemize}
The relationships between the HFNMCF problem and bond graphs can be similarly deduced via the HFGT-to-linear graph relationships stated above.  With the above understanding, Equations \ref{Eq:ObjFunc1}-\ref{Eq:DevicModels2} are elaborated below.  

\vspace{0.1in}
\subsubsection{Objective Function}
With respect to the objective function in Eq.  \ref{Eq:ObjFunc1}, $Z$ is a convex objective function separable in discrete time steps $k$.  For the remainder of this work, the discrete-time-dependent functions $f_k$ are assumed to be time-invariant quadratic functions.  Matrix $F_{QP}$ and vector $f_{QP}$ in Equation \ref{Eq:ObjFunc} allow quadratic and linear costs to be incurred from the place and transition markings in both the engineering system net and operand nets. 
\begin{align}\label{Eq:ObjFunc}
Z &= \sum_{k=1}^{K-1} x^T[k] F_{QP} x[k] + f_{QP}^T x[k]
\end{align}
\begin{itemize}
\item $F_{QP}$ is a positive semi-definite, diagonal, quadratic coefficient matrix.
\item $f_{QP}$ is a linear coefficient matrix.
\end{itemize}

\vspace{0.1in}
\subsubsection{Equality Constraints}

Matrix $A_{QP}$ and vector $B_{QP}$ in Equation \ref{Eq:EqualityConstraints} are constructed by concatenating constraints Equations  \ref{Eq:ESN-STF1}-\ref{CH6:eq:HFGTprog:comp:Fini}.  

\begin{align}\label{Eq:ESN-STF1}
-Q_{B}[k+1]+Q_{B}[k]+{M}^+U^+[k]\Delta T - {M}^-U^-[k]\Delta T=&0 && \!\!\!\!\!\!\!\!\!\!\!\!\!\!\!\!\!\!\!\!\!\!\!\!\!\!\!\!\!\!\!\!\!\!\!\!\!\!\!\!\!\forall k \in \{1, \dots, K\}\\  \label{Eq:ESN-STF2}
-Q_{\cal E}[k+1]+Q_{\cal E}[k]-U^+[k]\Delta T + U^-[k]\Delta T=&0 && \!\!\!\!\!\!\!\!\!\!\!\!\!\!\!\!\!\!\!\!\!\!\!\!\!\!\!\!\!\!\!\!\!\!\!\!\!\!\!\!\!\forall k \in \{1, \dots, K\}\\ \label{Eq:DurationConstraint}
 - U_\psi^+[k+k_{d\psi}]+ U_{\psi}^-[k] = &0 && \!\!\!\!\!\!\!\!\!\!\!\!\!\!\!\!\!\!\!\!\!\!\!\!\!\!\!\!\!\!\!\!\!\!\!\!\!\!\!\!\!\forall k\in \{1, \dots, K\} \quad \psi \in \{1, \dots, {\cal E}_S\}\\ \label{Eq:OperandNet-STF1}-Q_{Sl_i}[k+1]+Q_{Sl_i}[k]+{M}_{l_i}^+U_{l_i}^+[k]\Delta T - {M}_{l_i}^-U_{l_i}^-[k]\Delta T=&0 && \!\!\!\!\!\!\!\!\!\!\!\!\!\!\!\!\!\!\!\!\!\!\!\!\!\!\!\!\!\!\!\!\!\!\!\!\!\!\!\!\!\forall k \in \{1, \dots, K\} \quad i \in \{1, \dots, |L|\}\\ \label{Eq:OperandNet-STF2}
-Q_{{\cal E}l_i}[k+1]+Q_{{\cal E}l_i}[k]-U_{l_i}^+[k]\Delta T + U_{l_i}^-[k]\Delta T=&0 && \!\!\!\!\!\!\!\!\!\!\!\!\!\!\!\!\!\!\!\!\!\!\!\!\!\!\!\!\!\!\!\!\!\!\!\!\!\!\!\!\!\forall k \in \{1, \dots, K\} \quad i \in \{1, \dots, |L|\}\\ \label{Eq:OperandNetDurationConstraint}
- U_{xl_i}^+[k+k_{dxl_i}]+ U_{xl_i}^-[k] = &0 &&  \!\!\!\!\!\!\!\!\!\!\!\!\!\!\!\!\!\!\!\!\!\!\!\!\!\!\!\!\!\!\!\!\!\!\!\!\!\!\!\!\!
\forall k\in \{1, \dots, K\}, \: \forall x\in \{1, \dots, |{\cal E}_{l_i}\}|, \: l_i \in \{1, \dots, |L|\}\\ \label{Eq:SyncPlus}
U^+_L[k] - \widehat{\Lambda}^+ U^+[k] =&0 && \!\!\!\!\!\!\!\!\!\!\!\!\!\!\!\!\!\!\!\!\!\!\!\!\!\!\!\!\!\!\!\!\!\!\!\!\!\!\!\!\!\forall k \in \{1, \dots, K\}\\ \label{Eq:SyncMinus}
U^-_L[k] - \widehat{\Lambda}^- U^-[k] =&0 && \!\!\!\!\!\!\!\!\!\!\!\!\!\!\!\!\!\!\!\!\!\!\!\!\!\!\!\!\!\!\!\!\!\!\!\!\!\!\!\!\!\forall k \in \{1, \dots, K\}\\ \label{CH6:eq:HFGTprog:comp:Bound}
\begin{bmatrix}
D_{Up} & \mathbf{0} \\ \mathbf{0} & D_{Un}
\end{bmatrix} \begin{bmatrix}
U^+ \\ U^-
\end{bmatrix}[k] =& \begin{bmatrix}
C_{Up} \\ C_{Un}
\end{bmatrix}[k] && \!\!\!\!\!\!\!\!\!\!\!\!\!\!\!\!\!\!\!\!\!\!\!\!\!\!\!\!\!\!\!\!\!\!\!\!\!\!\!\!\!\forall k \in \{1, \dots, K\} \\\label{Eq:OperandRequirements}
\begin{bmatrix}
E_{Lp} & \mathbf{0} \\ \mathbf{0} & E_{Ln}
\end{bmatrix} \begin{bmatrix}
U^+_{l_i} \\ U^-_{l_i}
\end{bmatrix}[k] =& \begin{bmatrix}
F_{Lpi} \\ F_{Lni}
\end{bmatrix}[k] && \!\!\!\!\!\!\!\!\!\!\!\!\!\!\!\!\!\!\!\!\!\!\!\!\!\!\!\!\!\!\!\!\!\!\!\!\!\!\!\!\!\forall k \in \{1, \dots, K\}\quad i \in \{1, \dots, |L|\} \\\label{CH6:eq:HFGTprog:comp:Init} 
\begin{bmatrix} Q_B ; Q_{\cal E} ; Q_{SL} \end{bmatrix}[1] =& \begin{bmatrix} C_{B1} ; C_{{\cal E}1} ; C_{{SL}1} \end{bmatrix} \\ \label{CH6:eq:HFGTprog:comp:Fini}
\begin{bmatrix} Q_B ; Q_{\cal E} ; Q_{SL} ; U^- ; U_L^- \end{bmatrix}[K+1] =   &\begin{bmatrix} C_{BK} ; C_{{\cal E}K} ; C_{{SL}K} ; \mathbf{0} ; \mathbf{0} \end{bmatrix}
\end{align}

\begin{itemize}
\item Equations \ref{Eq:ESN-STF1} and \ref{Eq:ESN-STF2} describe the state transition function of an engineering system net (Defn \ref{Defn:ESN} \& \ref{Defn:ESN-STF}).
\item Equation \ref{Eq:DurationConstraint} is the engineering system net transition duration constraint where the end of the $\psi^{th}$ transition occurs $k_{d\psi}$ time steps after its beginning. 
\item Equations \ref{Eq:OperandNet-STF1} and \ref{Eq:OperandNet-STF2} describe the state transition function of each operand net ${\cal N}_{l_i}$ (Defn. \ref{Defn:OperandNet} \& \ref{Defn:OperandNet-STF}) associated with each operand $l_i \in L$.  
\item Equation \ref{Eq:OperandNetDurationConstraint} is the operand net transition duration constraint where the end of the $x^{th}$ transition occurs $k_{dx_{l_i}}$ time steps after its beginning. 
\item Equations \ref{Eq:SyncPlus} and \ref{Eq:SyncMinus} are synchronization constraints that couple the input and output firing vectors of the engineering system net to the input and output firing vectors of the operand nets respectively. $U_L^-$ and $U_L^+$ are the vertical concatenations of the input and output firing vectors $U_{l_i}^-$ and $U_{l_i}^+$ respectively.
\begin{align}
U_L^-[k]&=\left[U^-_{l_1}; \ldots; U^-_{l_{|L|}}\right][k] \\
U_L^+[k]&=\left[U^+_{l_1}; \ldots; U^+_{l_{|L|}}\right][k]
\end{align}
\item Equations \ref{CH6:eq:HFGTprog:comp:Bound} and \ref{Eq:OperandRequirements} are boundary conditions.  Eq. \ref{CH6:eq:HFGTprog:comp:Bound} is a boundary condition constraint that allows some of the engineering system net firing vectors decision variables to be set to an exogenous constant.  Eq. \ref{Eq:OperandRequirements} is a boundary condition constraint that allows some of the operand net firing vector decision variables to be set to an exogenous constant.  
\item Equations \ref{CH6:eq:HFGTprog:comp:Init} and \ref{CH6:eq:HFGTprog:comp:Fini} are the initial and final conditions of the engineering system net and the operand nets where $Q_{SL}$ is the vertical concatenation of the place marking vectors of the operand nets $Q_{Sl_i}$.
\begin{align}
Q_{SL}^-[k]&=\left[Q^-_{Sl_1}; \ldots; U^-_{Sl_{|L|}}\right][k] \\
U_{SL}^+[k]&=\left[U^+_{Sl_1}; \ldots; U^+_{Sl_{|L|}}\right][k]
\end{align}
\end{itemize}

\vspace{0.1in}
\subsubsection{Inequality Constraints}
 $D_{QP}()$ and vector $E_{QP}$ in Equation \ref{ch6:eq:QPcanonicalform:3} place capacity constraints on the vector of decision variables at each time step $x[k] = \begin{bmatrix} Q_B ; Q_{\cal E} ; Q_{SL} ; Q_{{\cal E}L} ; U^- ; U^+ ; U^-_L ; U^+_L \end{bmatrix}[k] \quad \forall k \in \{1, \dots, K\}$. This flexible formulation allows capacity constraints on the place and transition markings in both the engineering system net and operand nets.  

\vspace{0.1in}
\subsubsection{Device Model Constraints}
As mentioned above, g(X,Y) and h(Y) are a set of device model functions whose presence and nature depend on the specific problem application.  They can not be further elaborated until the application domain and its associated capabilities are identified.  

\begin{figure}[H]
\centering
\includegraphics[width=1\textwidth]{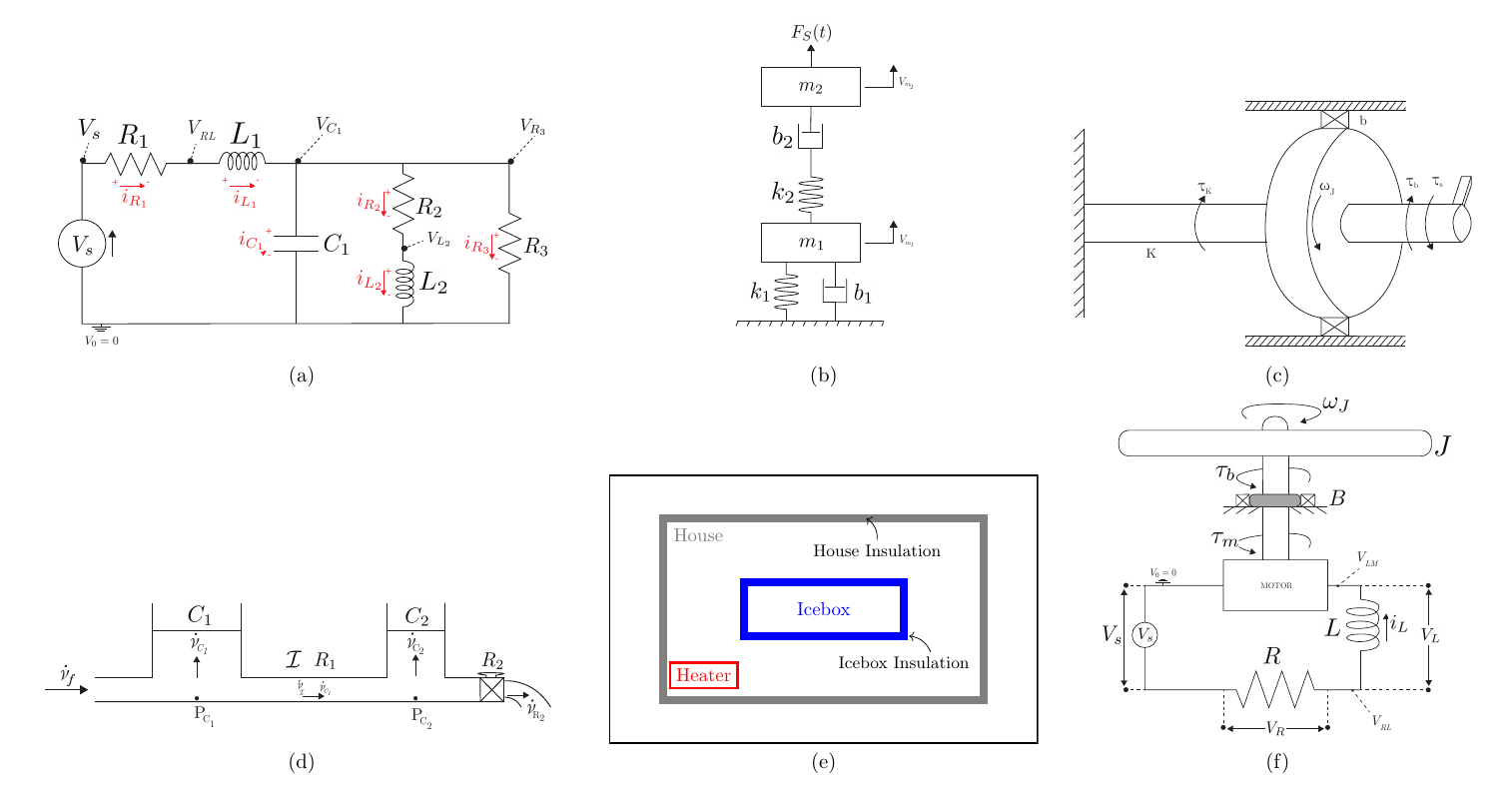}
\caption{Different system domains studied in this paper: (a) Electrical system; (b) Translational mechanical system; (c) Rotational mechanical system; (d) Fluidic system; (e) Thermal system; (f) Electronic motor generator}
\label{fig:illustrative_examples}
\end{figure}

\section{Illustrative Examples}\label{sec:examples}
The overview of linear graphs, bond graphs, and hetero-functional graphs in the previous section, provides a strong foundation for comparison.  In order to concretely describe the relationship between linear graphs, bond graphs, and hetero-functional graphs, this section introduces several illustrative examples for further discussion.  As these graph-based methodologies have been applied to electrical, translational mechanical, rotational mechanical, fluidic, thermal, and multi-energy domains, one example system for each of these domains is provided. Figure \ref{fig:illustrative_examples} summarizes each of these systems graphically.

\begin{itemize}
\item Figure \ref{fig:illustrative_examples}a is an electrical system composed of an ideal voltage source $V_s$, three resistors $R_1$, $R_2$, and $R_3$, two inductors $L_1$ and $L_2$, and one capacitor $C_1$.
\item Figure \ref{fig:illustrative_examples}b is a translational mechanical system composed of an ideal force source $F_s(t)$, two translational dampers $B_1$ and $B_2$, two translational springs $K_1$ and $K_2$, two masses $m_1$ and $m_2$. $V_{m_1}$ and $V_{m_2}$ indicate the sign convention of positive velocity.
\item Figure \ref{fig:illustrative_examples}c is a rotational mechanical system composed of a rotating disk $J$, a torsional spring $K$, a torsional damper $b$, and a torque source $\tau_s(t)$. The angular velocity $\omega_J$ indicates the sign convention of positive angular velocity.
\item Figure \ref{fig:illustrative_examples}d is a fluidic system composed of two tanks $C_1$ and $C_2$, a pipe with inductance ${\cal I}$ and resistance $R_{1}$ and a valve with resistance $R_{2}$, and a volumetric flow rate source $\dot{\cal V}_f$. The pressure measurement points are $P_{C_1}$ and $P_{C_2}$.
\item Figure \ref{fig:illustrative_examples}e is a thermal system composed of a ``House" and an ``ice box" which are considering as thermal capacitors $C_h$ and $C_i$ respectively. Also the ``House Insulation" and ``Ice box insulation " are considered as thermal resistors, $R_h$ and $R_i$ respectively. The system also includes a ``Heater" as a heat flow source $\dot{Q}_s$.  
\item Figure \ref{fig:illustrative_examples}f is an electro-mechanical system composed of an ideal voltage source $V_s$, a resistor $R$, an inductor $L$, a torsional damper $B$, and a rotating disk $J$. The angular velocity $\omega_J$ indicates the sign convention of positive angular velocity of the disk.
\end{itemize}

\section{Linear Graphs by Example}
\label{subsec:Linear_graph_by_example}
Building upon the illustrative examples described in the previous section, this section demonstrates the linear graph methodology by example.  To recall from Fig. \ref{fig:different_graphs} and the overview provided in Sec. \ref{subsec:linear_graphs}, the linear graph methodology consists of four essential steps:
\begin{enumerate}
\item Construct the linear graph from the identified system elements.
\item Translate the linear graph into its corresponding normal tree.  
\item State the continuity, constitutive, and compatibility laws of the system.
\item Simplify these laws into a single state-space model.  
\end{enumerate}
This section follows each of these four steps for the six illustrative examples identified in Fig. \ref{fig:illustrative_examples}.  

\subsection{Electrical System}\label{Subsec:lin_graph_electrical_system}
First, the electrical circuit diagram in Fig. \ref{fig:illustrative_examples}a is transformed into the linear graph shown in Fig. \ref{fig:linear_graph_electrical_system}. Based on Fig. \ref{fig:across_through_variables}, the through-variable in electrical systems is current $i$, and the across-variable is voltage $V$. Furthermore Fig. \ref{fig:graph_elements} shows that resistors, capacitors, and inductors are categorized as D-type, A-type, and T-type elements respectively. Consequently, the voltage source in Fig. \ref{fig:illustrative_examples}a becomes the across-variable source in Fig. \ref{fig:linear_graph_electrical_system}.

Second, the linear graph shown in Fig. \ref{fig:linear_graph_electrical_system} is translated into its associated normal tree in Fig. \ref{fig:normal_tree_electrical_system}. To recall, in the linear graph methodology, a normal tree is a ``spanning tree" \cite{Hillier:2010:00,Steen:2010:00} that includes the following elements in order of priority:  
\begin{enumerate*}
\item all the system graph nodes,
\item all across variable sources,
\item as many A-Type elements as possible,
\item in the case of existing transformers or gyrators, include one branch of each transformer and both or neither branch of each gyrator,
\item as many D-Type elements as possible,
\item as many T-Type elements as possible, and
\item as many through variable sources as possible,
\end{enumerate*}
without creating any loops in the graph.
\cite{Rowell:1997:00}. For the electrical system in Fig. \ref{fig:linear_graph_electrical_system}, the $V_s$, $C_1$, $R_1$, and $R_2$ elements are included following the above prioritization, and consequently the $R_3$, $L_1$, and $L_2$ elements are removed. Importantly, the A-Type elements in the normal tree, and T-Type elements hidden from the normal tree, define the state variables in the system:  $V_{C_1}$ for the capacitor, and $i_{L_1}$ and $i_{L_2}$ for the two inductors.  

\begin{figure}[H] 
\centering
\begin{subfigure}[b]{0.5\textwidth} 
\centering
\includegraphics[width=0.6\linewidth]{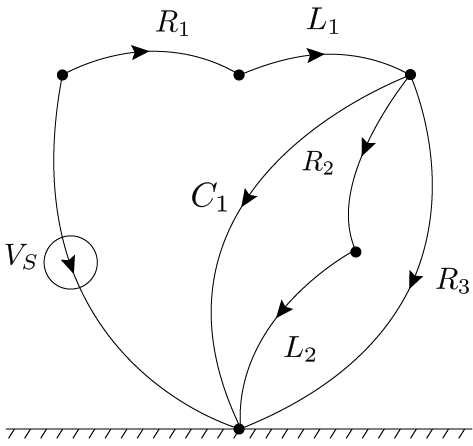} 
\caption{Linear graph}
\label{fig:linear_graph_electrical_system}
\end{subfigure}%
\hfill 
\begin{subfigure}[b]{0.5\textwidth}
\centering
\includegraphics[width=0.7\linewidth]{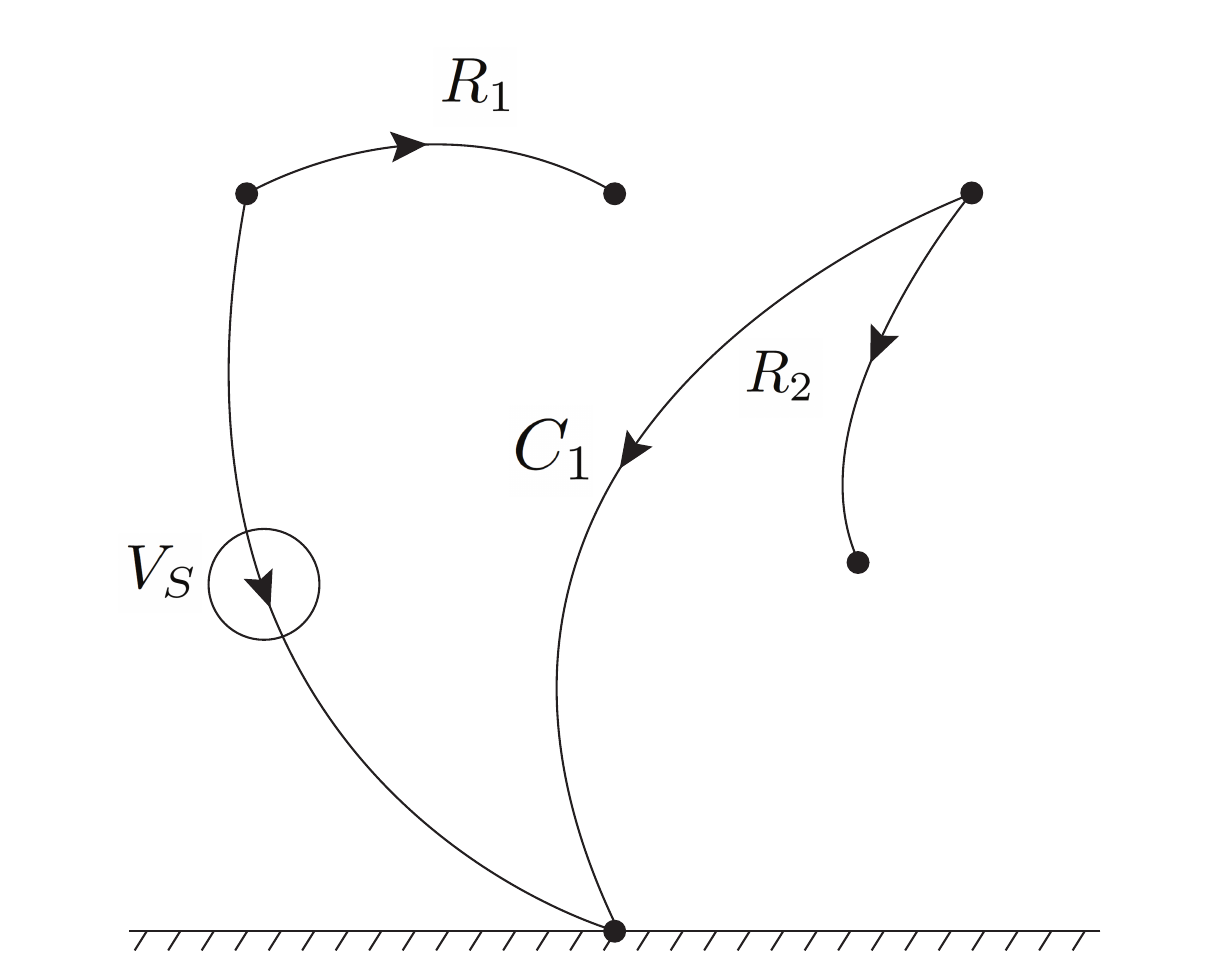} 
\caption{Normal tree}
\label{fig:normal_tree_electrical_system}
\end{subfigure}
\caption{Linear graph and normal tree of the electrical system, as referred to in Fig. 5a}
\end{figure}

Third, the normal tree facilitates the statement of the electrical system's constitutive, continuity, and compatibility laws. The constitutive laws of the system are: 
\begin{align}
\frac{dV_{C_1}}{dt} &= \frac{1}{C_1} i_{C_1}, \label{eq:electrical_const_1} \\
V_{R_1} &= i_{R_1} R_1, \label{eq:electrical_const_2} \\
\frac{di_{L_1}}{dt} &= \frac{1}{L_1} V_{L_1}, \label{eq:electrical_const_3} \\
V_{R_2} &= i_{R_2} R_2, \label{eq:electrical_const_4} \\
\frac{di_{L_2}}{dt} &= \frac{1}{L_2} V_{L_2}, \label{eq:electrical_const_5} \\
i_{R_3} &= \frac{1}{R_3} V_{R_3}. \label{eq:electrical_const_6}
\end{align}
Additionally, the continuity laws are:
\begin{align}
i_{C_1} &= i_{L_1} - i_{R_2} - i_{R_3}, \label{eq:electrical_cont_1} \\
i_{R_1} &= i_{L_1}, \label{eq:electrical_cont_2} \\
i_{R_2} &= i_{L_2}. \label{eq:electrical_cont_3}
\end{align}
Also, the compatibility laws are:
\begin{align}
V_{L_1} &= V_S - V_{C_1} - V_{R_1}, \label{eq:electrical_comp_1} \\
V_{L_2} &= V_{C_1} - V_{R_2}, \label{eq:electrical_comp_2} \\
V_{R_3} &= V_{C_1}. \label{eq:electrical_comp_3}
\end{align}

Finally, these laws are simplified algebraically to produce a state space model in Eq. \ref{eq:electrical_state}.

\begin{equation}
\frac{d}{dt}
\begin{bmatrix}
V_{C_1} \\
i_{L_1} \\
i_{L_2}
\end{bmatrix}
=
\begin{bmatrix}
-\frac{1}{R_3C_1} & \frac{1}{C_1} & -\frac{1}{C_1} \\
-\frac{1}{L_1} & -\frac{R_1}{L_1} & 0 \\
\frac{1}{L_2} & 0 & -\frac{R_2}{L_2}
\end{bmatrix}
\begin{bmatrix}
V_{C_1} \\
i_{L_1} \\
i_{L_2}
\end{bmatrix}
+
\begin{bmatrix}
0 \\
\frac{1}{L_1} \\
0
\end{bmatrix}
V_S \label{eq:electrical_state}
\end{equation}

\subsection{Translational Mechanical System}
The state space model of the translational mechanical system in Fig. \ref{fig:illustrative_examples}b is developed similarly.  First, it is transformed into the linear graph shown in Fig. \ref{fig:linear_graph_translational_mechanical_system}. According to Fig. \ref{fig:across_through_variables}, the through-variable in translation mechanical systems is force $F$, and the across-variable is velocity $V$. Furthermore, Fig. \ref{fig:graph_elements} shows that translational dampers, masses, and translational springs are categorized as D-type, A-type, and T-type elements respectively. Consequently, a force source is a through variable source.

Second, the linear graph shown in Fig. \ref{fig:linear_graph_translational_mechanical_system} is translated into the normal tree in Fig. \ref{fig:normal_tree_translational_mechanical_system}.
Given the prioritization exposited in Sec. \ref{Subsec:lin_graph_electrical_system}, the $m_1$, $m_2$, and $b_2$ elements in this translational mechanical system have been included, and consequently the $F_s$, $b_1$, $k_1$, and $k_2$ elements are removed. The across variables of the A-Type elements included in the normal tree $V_{m_1}$ and $V_{m_2}$, and the through variables of T-Type elements not included in the normal tree $F_{k_1}$ and $F_{k_2}$, are identified as system state variables.  
\begin{figure}[H] 
\centering
\begin{subfigure}[b]{0.5\textwidth} 
\centering
\includegraphics[width=0.6\linewidth]{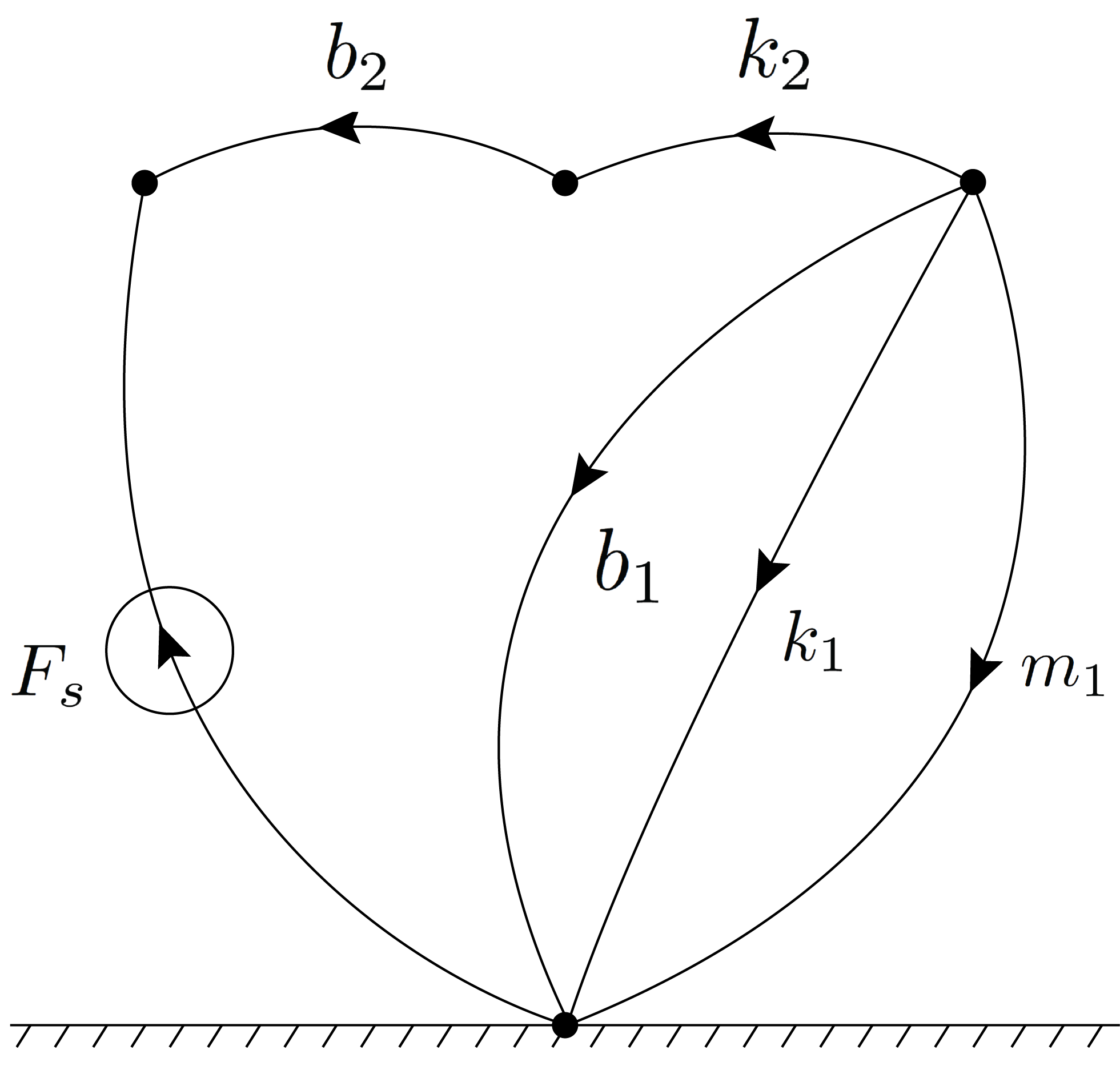} 
\caption{Linear graph}
\label{fig:linear_graph_translational_mechanical_system}
\end{subfigure}%
\hfill 
\begin{subfigure}[b]{0.5\textwidth}
\centering
\includegraphics[width=0.7\linewidth]{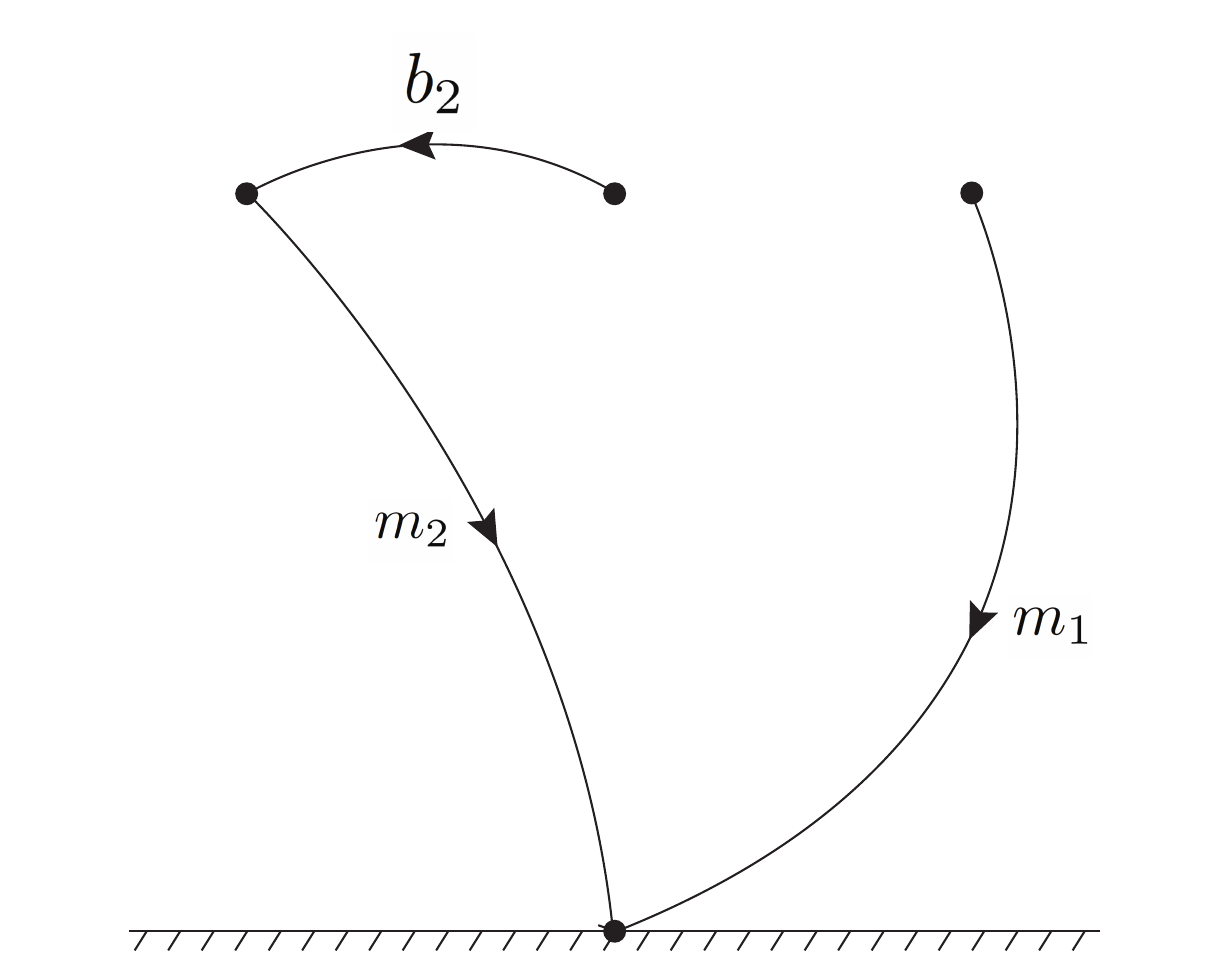} 
\caption{Normal tree}
\label{fig:normal_tree_translational_mechanical_system}
\end{subfigure}
\caption{Linear graph and normal tree of the translational mechanical system, as referred to in Fig. 5b}
\end{figure}

Thirdly, the normal tree facilitates the statement of the translational mechanical system's constitutive, continuity, and compatibility laws. The constitutive laws are:  
\begin{align}
\frac{dV_{m_1}}{dt} &= \frac{1}{m_1} F_{m_1}, \label{eq:trans_const_1} \\
\frac{dV_{m_2}}{dt} &= \frac{1}{m_2} F_{m_2}, \label{eq:trans_const_2} \\
\frac{dF_{k_1}}{dt} &= k_1 V_{k_1}, \label{eq:trans_const_3} \\
\frac{dF_{k_2}}{dt} &= k_2 V_{k_2}, \label{eq:trans_const_4} \\
V_{b_2} &= \frac{1}{b_2} F_{b_2}, \label{eq:trans_const_5} \\
F_{b_1} &= b_1 V_{b_1} \label{eq:trans_const_6}
\end{align}
The continuity laws are:
\begin{align}
F_{m_1} &= - F_{b_1} - F_{k1} - F_{k2}, \label{eq:trans_cont_1} \\
F_{m_2} &= F_{b_2} + F_S, \label{eq:trans_cont_2} \\
F_{b_2} &= F_{k_2} \label{eq:trans_cont_3}
\end{align}
The compatibility laws are:
\begin{align}
V_{k_1} &= V_{m_1}, \label{eq:trans_comp_1} \\
V_{k_2} &= V_{m_1} - V_{m_2} - V_{b_2}, \label{eq:trans_comp_2} \\
V_{b_1} &= V_{m_1} \label{eq:trans_comp_3}
\end{align}

Finally, these equations are algebraically simplified to produce the system's state space model in Equation \ref{eq:trans_state}.  
\begin{equation}
\frac{d}{dt}
\begin{bmatrix}
V_{m_1} \\
V_{m_2} \\
F_{k_1} \\
F_{k_2}
\end{bmatrix}
=
\begin{bmatrix}
-\frac{b_1}{m_1} & 0 & -\frac{1}{m_1} & -\frac{1}{m_1} \\
0 & 0 & 0 & \frac{1}{m_2} \\
k_1 & 0 & 0 & 0 \\
k_2 & -k_2 & 0 & -\frac{k_2}{b_2}
\end{bmatrix}
\begin{bmatrix}
V_{m_1} \\
V_{m_2} \\
F_{k_1} \\
F_{k_2}
\end{bmatrix}
+
\begin{bmatrix}
0 \\
\frac{1}{m_2} \\
0 \\
0
\end{bmatrix}
V_S \label{eq:trans_state}
\end{equation}

\subsection{Rotational Mechanical System}
The linear graph methodology is applied similarly to the rotational mechanical system in Fig. \ref{fig:illustrative_examples}c.   It is transformed into the linear graph shown in Fig. \ref{fig:linear_graph_rotational_mechanical_system}. Referring back to Fig. \ref{fig:across_through_variables}, the through-variable in a rotational mechanical system is torque $\tau$, and the across-variable is angular velocity $\omega$. Also referring back to Fig. \ref{fig:graph_elements}, rotational dampers, rotational disks, and rotational springs are categorized as D-type, A-type, and T-type elements respectively. Consequently, a torque source is a through-variable source.

Second, the linear graph shown in Fig. \ref{fig:linear_graph_rotational_mechanical_system} is translated into the normal tree in Fig. \ref{fig:normal_tree_rotational_mechanical_system}.  Following the prioritization mentioned in Sec \ref{Subsec:lin_graph_electrical_system}, the $J$ element in this rotational mechanical system is included accordingly, and consequently, the $b$, $K$, and $\tau_s$ elements are necessarily removed. Furthermore, the across variables of A-Type elements included in the normal tree $\omega_{J}$, and the through variables of T-Type elements not included in the normal tree $\tau_{K}$ are the state variables of this system.
\begin{figure}[H] 
\centering
\begin{subfigure}[b]{0.5\textwidth} 
\centering
\includegraphics[width=0.63\linewidth]{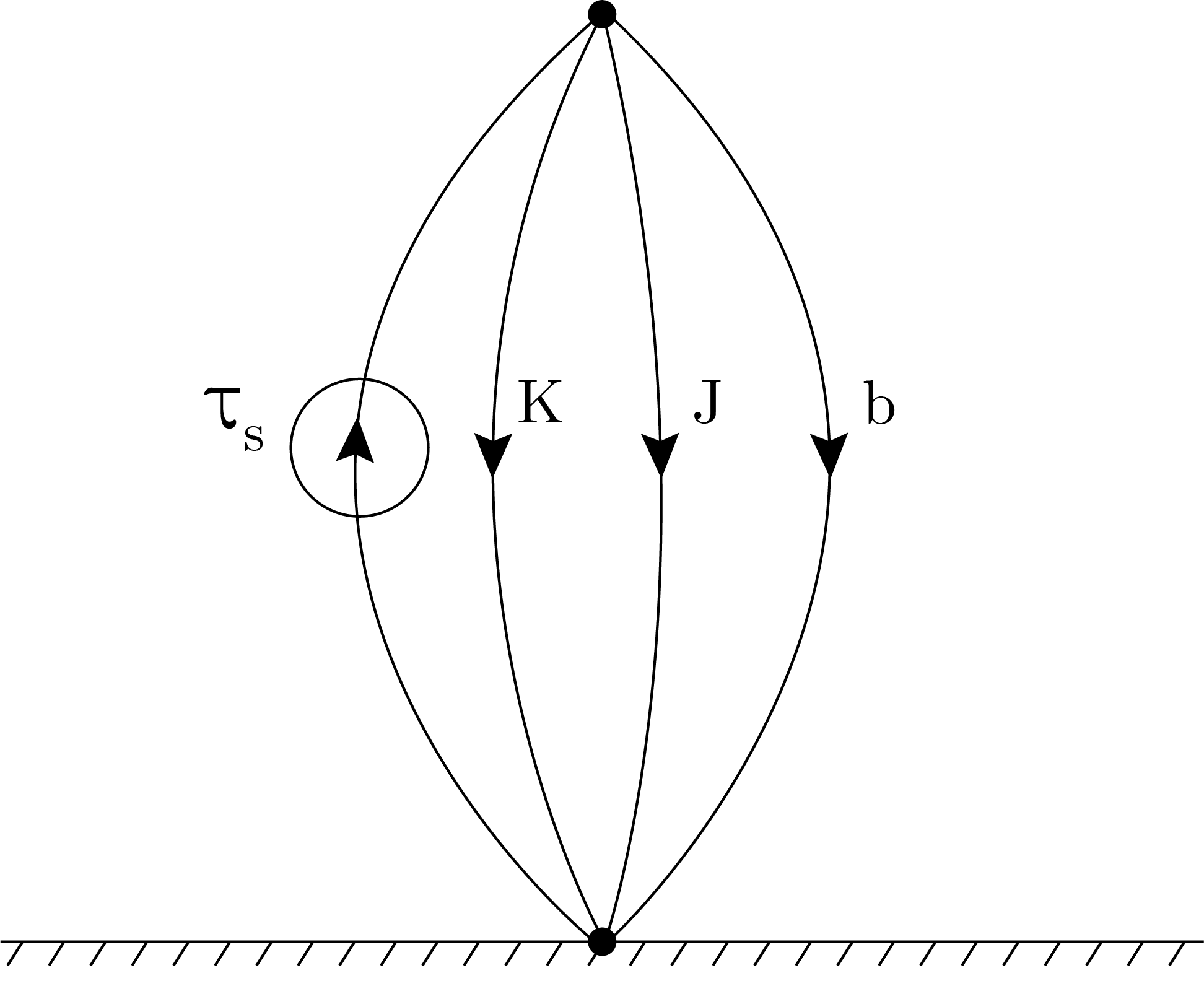} 
\caption{Linear graph}
\label{fig:linear_graph_rotational_mechanical_system}
\end{subfigure}%
\hfill 
\begin{subfigure}[b]{0.5\textwidth}
\centering
\includegraphics[width=0.7\linewidth]{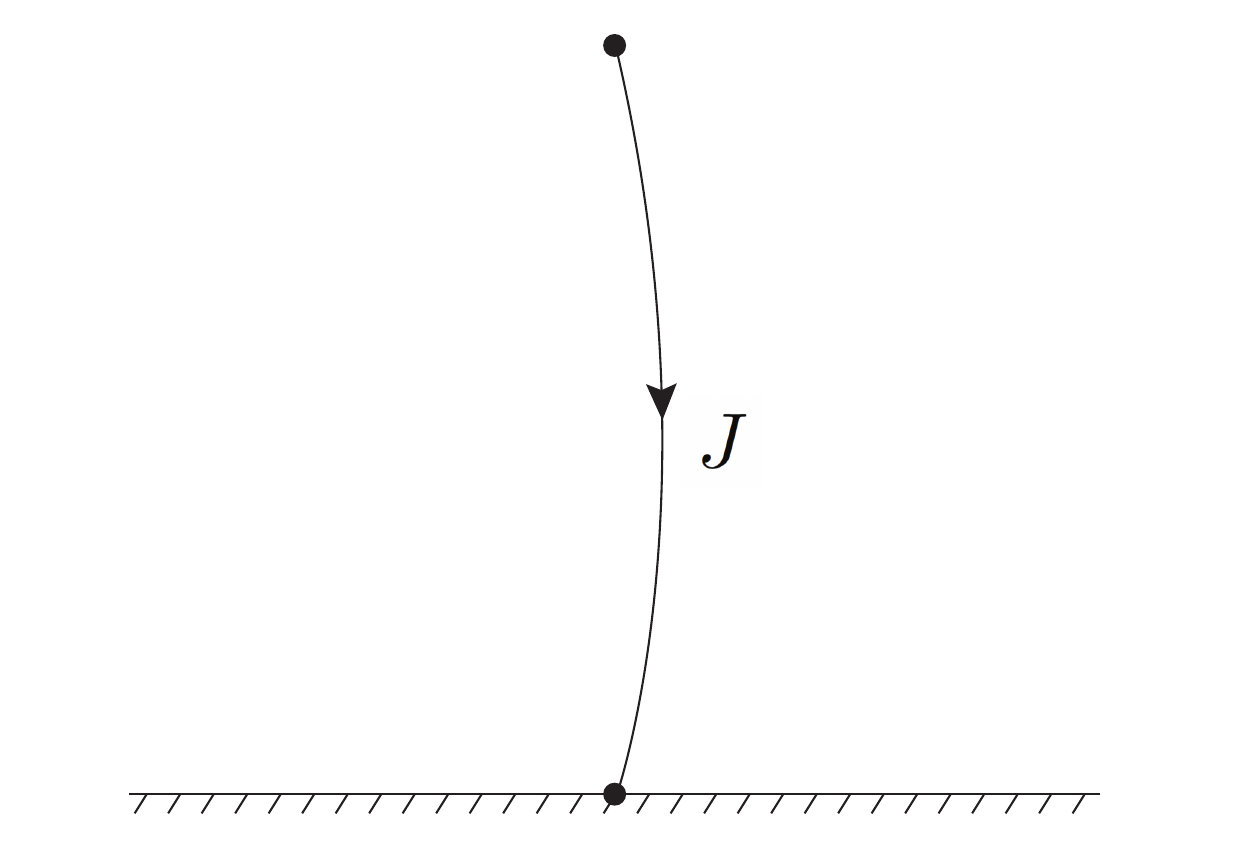} 
\caption{Normal tree}
\label{fig:normal_tree_rotational_mechanical_system}
\end{subfigure}
\caption{Linear graph and normal tree of the rotational mechanical system, as referred to in Fig. 5c}
\end{figure}

Third, the normal tree facilitates the statement of the rotational mechanical system's constitutive, continuity, and compatibility laws. The constitutive laws are:
\begin{align}
\frac{d\omega_J}{dt} &= \frac{1}{J} \tau_J \label{eq:rot_const_1} \\
\frac{d\tau_K}{dt} &= K \omega_K \label{eq:rot_const_2} \\
\tau_b &= b \omega_b \label{eq:rot_const_3}
\end{align}
The continuity law is: 
\begin{align}
\tau_J &= \tau_s - \tau_K - \tau_b \label{eq:rot_cont_1}
\end{align}
The compatibility laws are:
\begin{align}
\omega_K &= \omega_J \label{eq:rot_comp_1} \\
\omega_b &= \omega_J \label{eq:rot_comp_2}
\end{align}

Finally, these laws are simplified algebraically to produce a state space model in Eq. \ref{eq:rot_state}.
\begin{equation}
\frac{d}{dt} \begin{bmatrix} \omega_J \\ \tau_K \end{bmatrix} = \begin{bmatrix} -\frac{b}{J} & -\frac{1}{J} \\ k & 0 \end{bmatrix} \begin{bmatrix} \omega_J \\ \tau_K \end{bmatrix} + \begin{bmatrix} \frac{1}{J} \\ 0 \end{bmatrix} \tau_s \label{eq:rot_state}
\end{equation}

\subsection{Fluidic System}
The linear graph methodology is also applied to the fluidic system in Fig. \ref{fig:illustrative_examples}d. It is transformed into the linear graph shown in Fig. \ref{fig:linear_graph_fluidic_system}. In the fluidic system domain, Fig. \ref{fig:across_through_variables} shows that the through-variable is volumetric flow rate $\dot{\cal V}$, and the across-variable is pressure difference $P$.  Similarly, Fig. \ref{fig:graph_elements} shows that a fluidic resistance in a pipe or valve, a fluidic capacitance in a tank, and a fluidic inertance are categorized as D-type, A-type, and T-type elements respectively. Consequently, a fluid flow source is a through-variable source.

Second, the linear graph shown in Fig. \ref{fig:linear_graph_fluidic_system} is translated into its associated normal tree in Fig. \ref{fig:normal_tree_fluidic_system}.
Given the prioritization defined in Sec. \ref{Subsec:lin_graph_electrical_system}, the $C_1$, $C_2$, and $R_1$ elements are included, and consequently $\dot{\cal V}_f$, $R_2$, and $\cal I$ elements are necessarily removed. The across variables of A-Type elements included in the normal tree $P_{C_1}$ and $P_{C_2}$, and the through variables of T-Type elements not included in the normal tree $\dot{\cal V_{\cal I}}$, are the system state variables.

Third, the constitutive, continuity, and compatibility laws are determined from the normal tree. The constitutive laws are:
\begin{align}
\frac{dp_{C_1}}{dt} &= \frac{1}{C_1} \dot{\cal V}_{C_1} \label{eq:fluid_const_1} \\
\frac{dp_{C_2}}{dt} &= \frac{1}{C_2} \dot{\cal V}_{C_2} \label{eq:fluid_const_2} \\
\frac{d\dot{\cal V}_{\cal I}}{dt} &= \frac{1}{I} \dot{\cal V}_{R_1} \label{eq:fluid_const_3} \\
P_{R_1} &= R_1 P_{R_2} \label{eq:fluid_const_4} \\
\dot{\cal V}_{R_2} &= \frac{1}{R_2} P_{R_2} \label{eq:fluid_const_5}
\end{align}
The continuity laws are:
\begin{align}
\dot{\cal V}_{C_1} &= \dot{\cal V}_f - \dot{\cal V}_{\cal I} \label{eq:fluid_cont_1} \\
\dot{\cal V}_{C_2} &= \dot{\cal V}_{R_1} - \dot{\cal V}_{R_2} \label{eq:fluid_cont_2} \\
\dot{\cal V}_{R_1} &= \dot{\cal V}_{\cal I} \label{eq:fluid_cont_3}
\end{align}
The compatibility laws are:
\begin{align}
P_{\cal I} &= P_{C_1} - P_{C_2} - P_{R_1} \label{eq:fluid_comp_1} \\
P_{R_2} &= P_{C_2} \label{eq:fluid_comp_2}
\end{align}

Finally, these equations are simplified into the state space model in Eq \ref{eq:fluid_state}.
\begin{equation}
\frac{d}{dt} \begin{bmatrix} P_{C_1} \\ P_{C_2} \\ \dot{\cal V}_I \end{bmatrix} = \begin{bmatrix} 0 & 0 & -\frac{1}{C_1} \\ 0 & -\frac{1}{C_2 R_2} & \frac{1}{C_2} \\ \frac{1}{\cal I} & -\frac{1}{\cal I} & -\frac{R_1}{\cal I} \end{bmatrix} \begin{bmatrix} P_{C_1} \\ P_{C_2} \\ \dot{\cal V}_{\cal I} \end{bmatrix} + \begin{bmatrix} \frac{1}{C_1} \\ 0 \\ 0 \end{bmatrix} \dot{\cal V}_f \label{eq:fluid_state}
\end{equation}
\begin{figure} 
\centering
\begin{subfigure}[b]{0.4\textwidth} 
\centering
\includegraphics[width=0.65\linewidth]{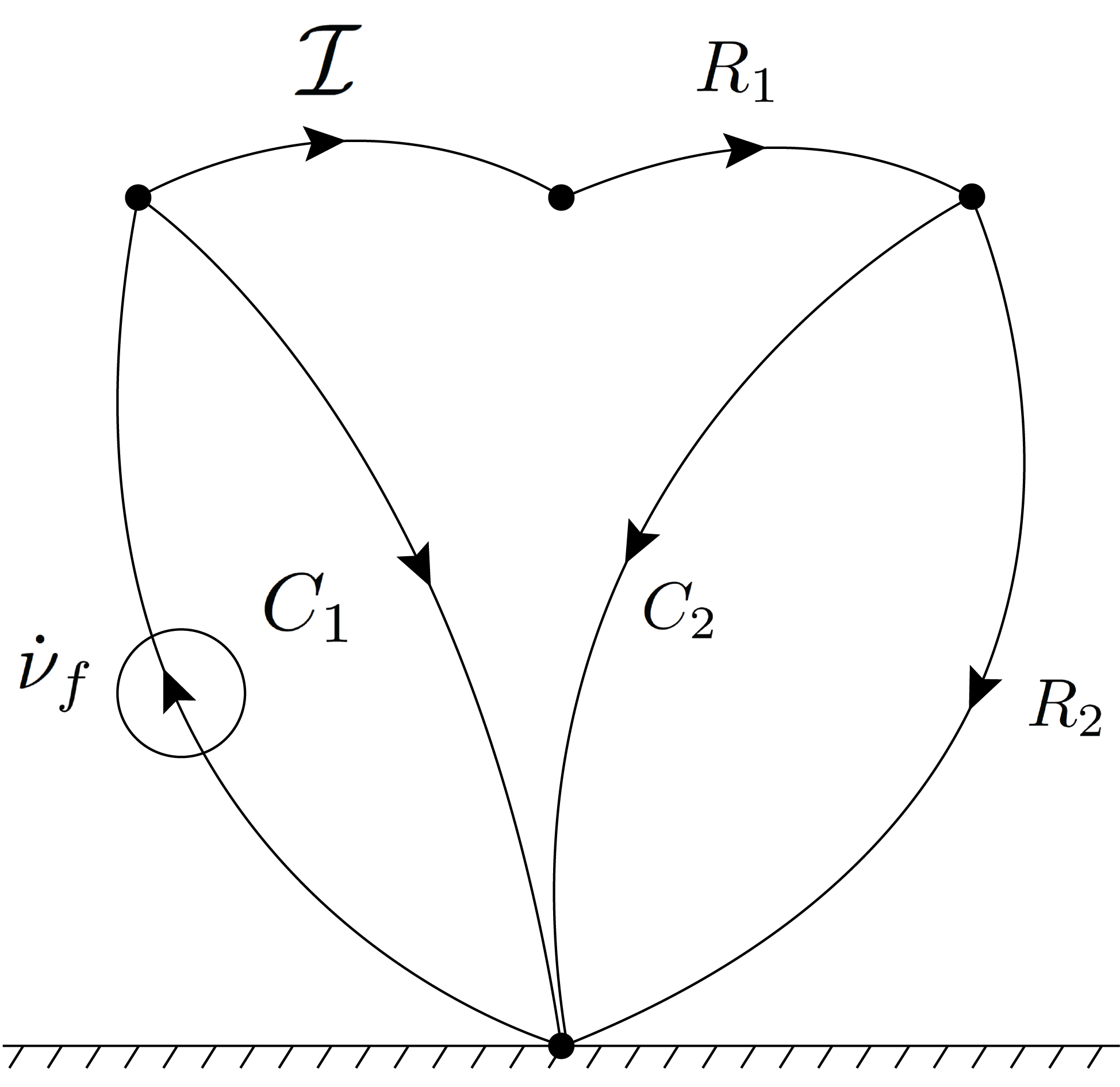} 
\caption{Linear graph}
\label{fig:linear_graph_fluidic_system}
\end{subfigure}%
\hfill 
\begin{subfigure}[b]{0.4\textwidth}
\centering
\includegraphics[width=0.7\linewidth]{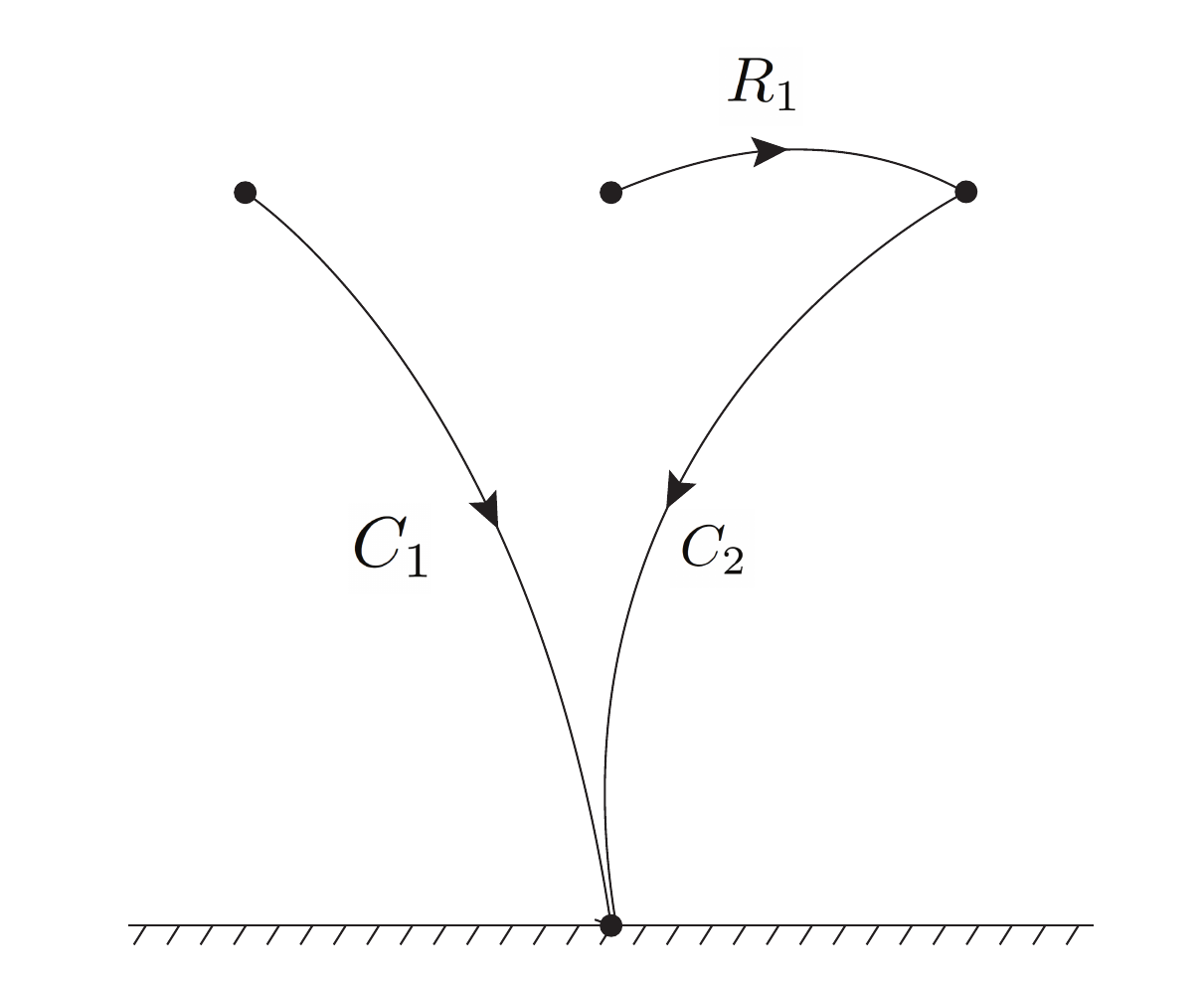} 
\caption{Normal tree}
\label{fig:normal_tree_fluidic_system}
\end{subfigure}
\caption{Linear graph and normal tree of the fluidic system, as referred to in Fig. 5d}
\end{figure}
\subsection{Thermal System}
The linear graph methodology is also applied to the thermal system depicted in Fig. \ref{fig:illustrative_examples}e.  To facilitate the process, this thermal system is first converted into an analogous electrical circuit as illustrated in Fig. \ref{fig:thermal_system_electrical_circuit}.  Furthermore, the following notation is adopted.  
\begin{itemize}
\item $T_{C_i}$: The temperature within the icebox, representing the temperature difference across the $C_i$ element, which is the thermal capacitance of the icebox.
\item $T_{R_i}$: The temperature drop across icebox wall resistor.
\item $T_{C_h}$:  The temperature within the house, representing the temperature difference across the $C_h$ element, which is the thermal capacitance of the house.
\item $T_{R_h}$: The temperature drop across the house wall resistor.
\item $\dot {Q}_{C_i}$: The heat flow into the icebox.
\item $\dot {Q}_{R_i}$: The heat flow through the icebox wall.
\item $\dot {Q}_{C_h}$: The heat flow into the house.
\item $\dot {Q}_{R_h}$: The heat flow through the house wall.
\end{itemize}
\begin{figure}[H]
\centering
\includegraphics[width=0.6\textwidth]{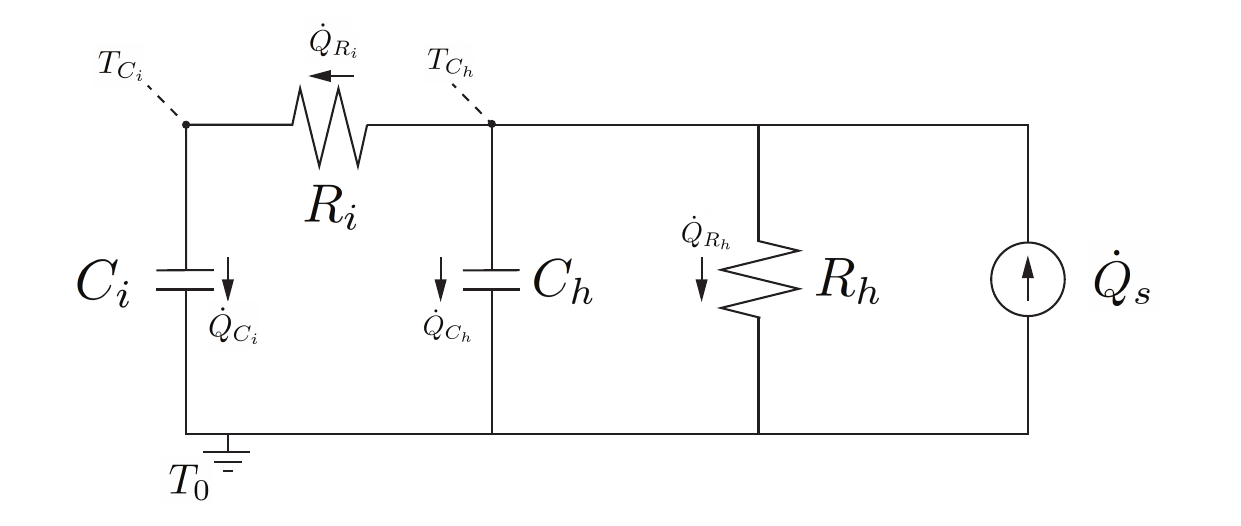}
\caption{Equivalent electrical system to the thermal system shown in Fig. \ref{fig:illustrative_examples}e}
\label{fig:thermal_system_electrical_circuit}
\end{figure}
From this point, the linear graph methodology is straightforwardly applied.  The fluidic system in Fig. \ref{fig:thermal_system_electrical_circuit} is transformed into the linear graph shown in Fig. \ref{fig:linear_graph_thermal_system}.  Next, Fig. \ref{fig:across_through_variables} shows that, in thermal systems, the through-variable is heat flow rate $\dot{Q}$, and the across-variable is temperature $T$. Furthermore, Fig. \ref{fig:across_through_variables} states that thermal resistances and capacitances are D-type and A-type elements respectively. Consequently, a fluid flow source is a through-variable source. 

Second, the associated normal tree is derived in Fig. \ref{fig:normal_tree_thermal_system}.  Given the prioritization defined in Sec. \ref{Subsec:lin_graph_electrical_system}, the $C_h$, and $C_i$ elements are included, and consequently, $\dot{Q_s}$, $R_h$, and $R_i$ elements are necessarily removed. The across variables of A-Type elements included in the normal tree $T_{C_i}$ and $T_{C_h}$ are the state variables of this system.

\begin{figure}[H] 
\centering
\begin{subfigure}[b]{0.5\textwidth} 
\centering
\includegraphics[width=0.65\linewidth]{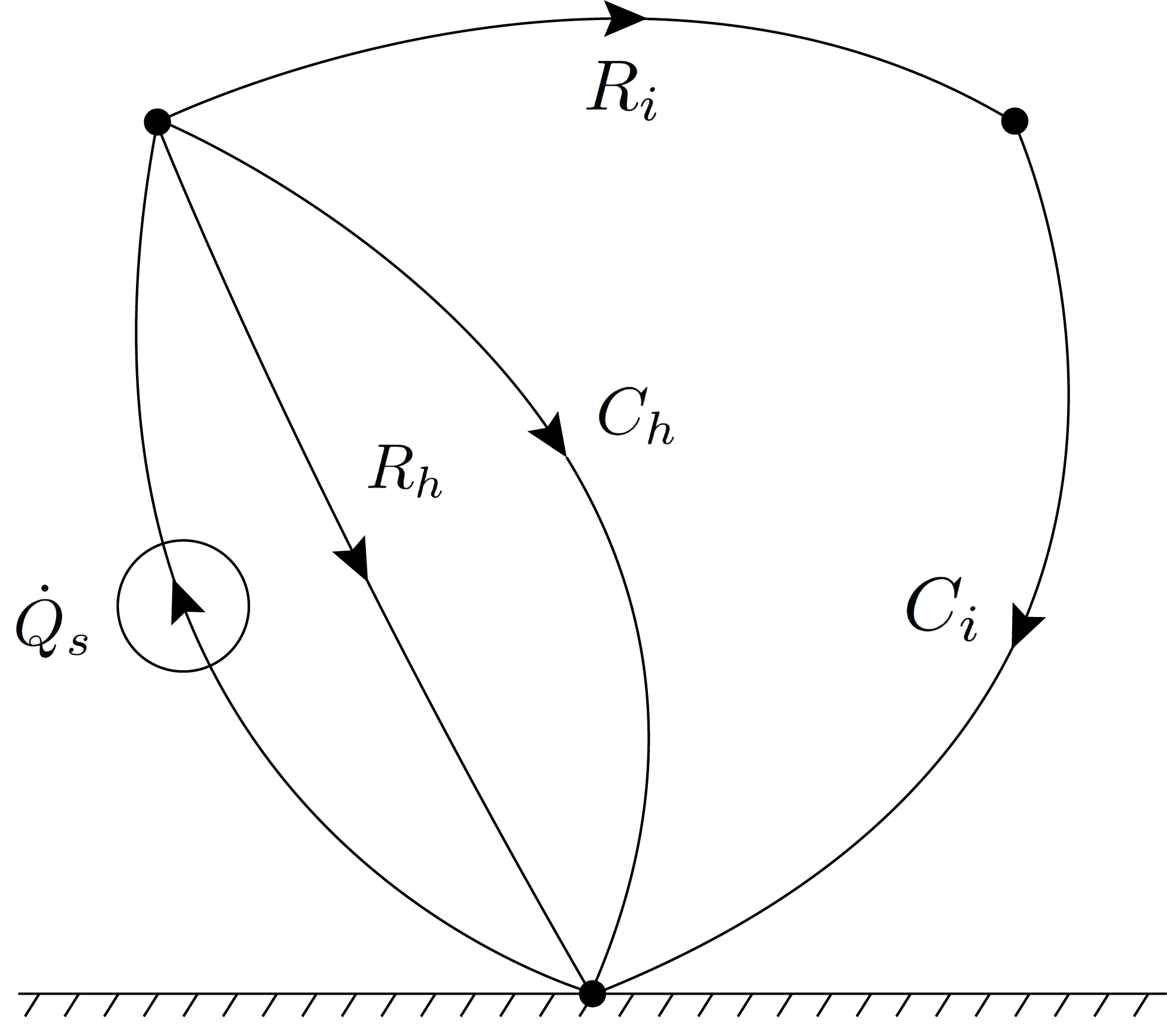} 
\caption{Linear graph}
\label{fig:linear_graph_thermal_system}
\end{subfigure}%
\hfill 
\begin{subfigure}[b]{0.5\textwidth}
\centering
\includegraphics[width=0.7\linewidth]{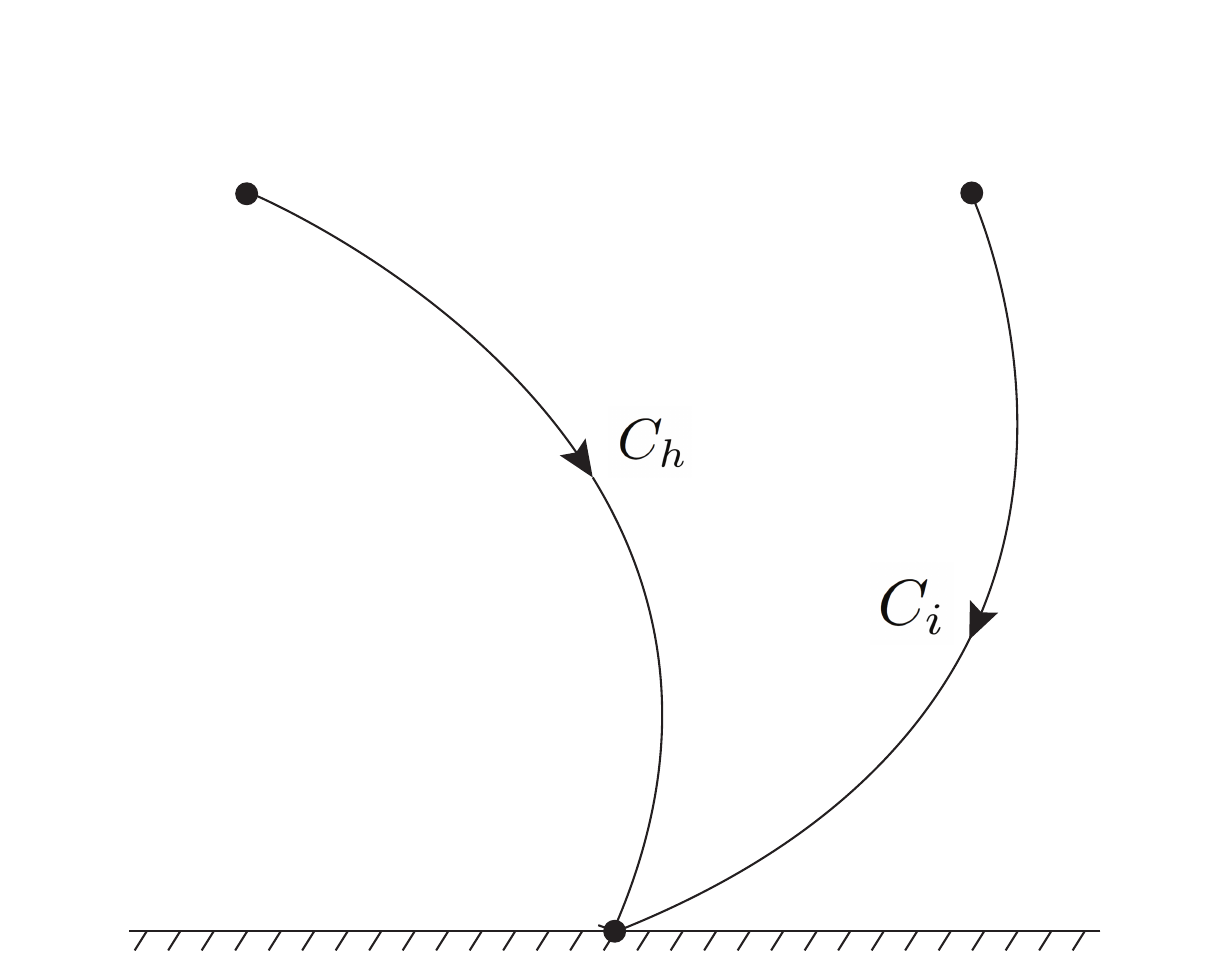} 
\caption{Normal tree}
\label{fig:normal_tree_thermal_system}
\end{subfigure}
\caption{Linear graph and normal tree of the thermal system, as referred to in Fig. 5e}
\end{figure}

Third, the normal tree facilitates the statement of the thermal system's constitutive, continuity, and compatibility laws. The constitutive laws are:
\begin{align}
\frac{dT_{C_i}}{dt} &= \frac{1}{C_i} \dot{Q}_{C_i} \label{eq:thermal_const_1} \\
\frac{dT_{C_h}}{dt} &= \frac{1}{C_h} \dot{Q}_{C_h} \label{eq:thermal_const_2} \\
\dot{Q}_{R_h} &= \frac{1}{R_h} T_{R_h} \label{eq:thermal_const_3} \\
\dot{Q}_{R_i} &= \frac{1}{R_i} T_{R_i} \label{eq:thermal_const_4}
\end{align}
The continuity laws are:
\begin{align}
\dot{Q}_{C_i} &= \dot{Q}_{R_i} \label{eq:thermal_cont_1} \\
\dot{Q}_{C_h} &= \dot{Q}_s - \dot{Q}_{R_i} - \dot{Q}_{R_h} \label{eq:thermal_cont_2}
\end{align}
The compatibility laws are:
\begin{align}
T_{R_h} &= T_{C_h} \label{eq:thermal_comp_1} \\
T_{R_i} &= T_{C_h} - T_{C_i} \label{eq:thermal_comp_2}
\end{align}

Finally, these equations are algebraically simplified to the state space model in Eq. \ref{eq:thermal_state}
\begin{equation}
\frac{d}{dt} \begin{bmatrix} T_{C_i} \\ T_{C_h} \end{bmatrix} = \begin{bmatrix}  -\frac{1}{C_i R_i} & \frac{1}{C_i R_i} \\ \frac{1}{R_i C_h} & -\frac{1}{R_i C_h}-\frac{1}{C_h R_h} \end{bmatrix} \begin{bmatrix} T_{C_i} \\ T_{C_h} \end{bmatrix} + \begin{bmatrix} 0 \\ \frac{1}{C_h} \end{bmatrix} \dot{Q}_s \label{eq:thermal_state}
\end{equation}

\subsection{Multi-Energy System}
First, the linear graph methodology is also applied to the electro-mechanical system shown in Fig. \ref{fig:illustrative_examples}f.  It is transformed into the linear graph shown in Fig. \ref{fig:linear_graph_multi_system}. As electro-mechanical systems are a combination of electrical and (rotational) mechanical systems; current $i$ and torque $\tau$ are the through variables, while voltage $V$ and angular velocity $\omega$ are the across-variables.  Additionally, referring to Fig. \ref{fig:graph_elements}, the electrical inductor and the rotational spring are the T-type energy storage elements.  The electrical capacitor and the rotating disk are A-type energy storage elements.  The electrical resistor and the mechanical damper are the D-type elements. The voltage $V_{S}$ represents an ideal across-variable source.  Lastly, a transformer element connects the electrical subsystem to the rotational subsystem and transforms power between the two different domains.

Second, Fig. \ref{fig:normal_tree_multi_system} shows the normal tree associated with the electro-mechanical system.  Given the prioritization defined in Sec. \ref{Subsec:lin_graph_electrical_system}, the $V_S$, the $J$, the electrical branch of the transformer, and the $R$ elements are included.  Consequently, the $L$, the $B$, and the mechanical branch of the transformer are necessarily removed. As mentioned in the electrical and rotational mechanical system examples, $\omega_{J}$, and $i_{L}$ are the state variables of this system.
\begin{figure}[H] 
\centering
\begin{subfigure}[b]{0.5\textwidth} 
\centering
\includegraphics[width=0.67\linewidth]{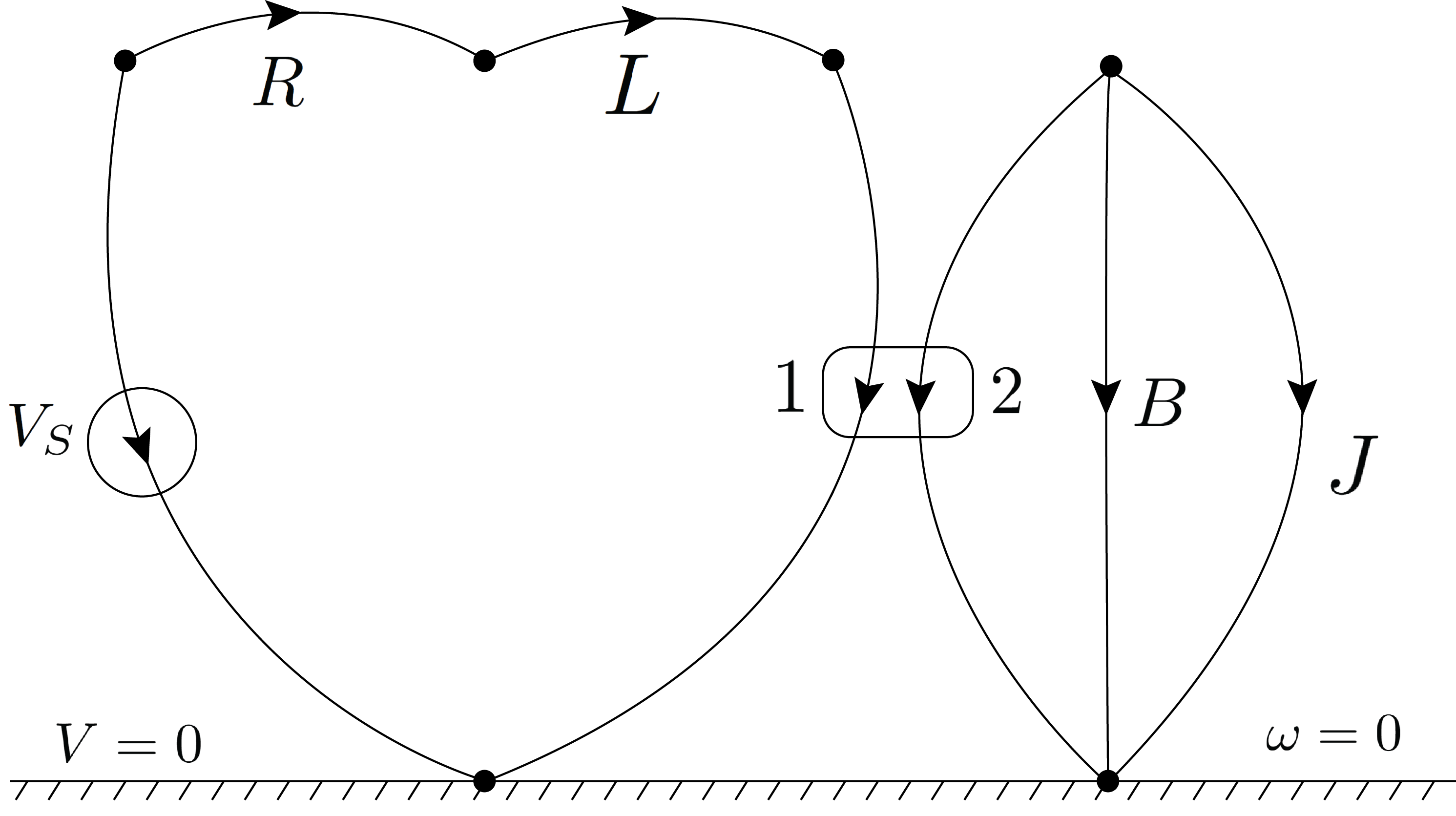} 
\caption{Linear graph}
\label{fig:linear_graph_multi_system}
\end{subfigure}%
\hfill 
\begin{subfigure}[b]{0.5\textwidth}
\centering
\includegraphics[width=0.7\linewidth]{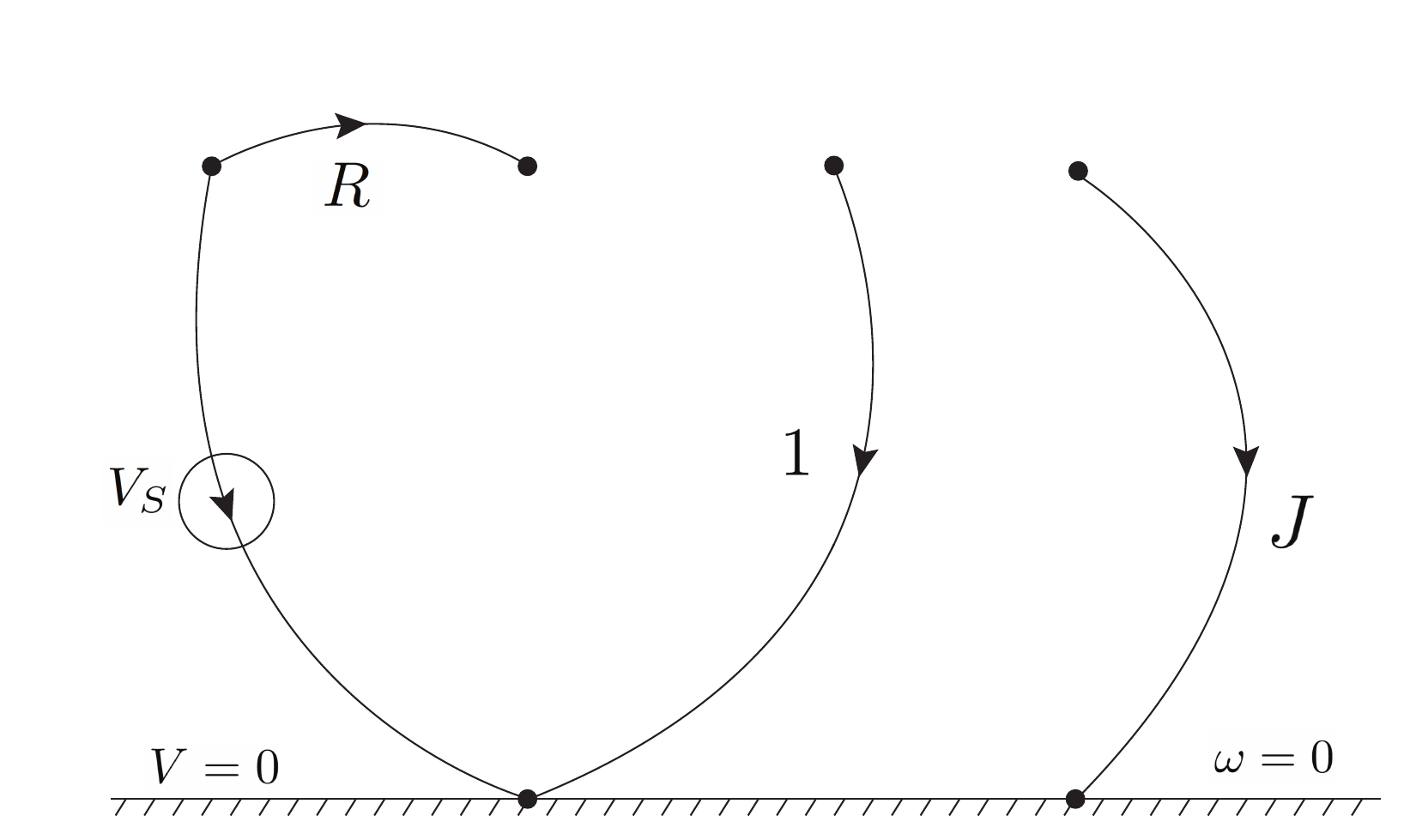} 
\caption{Normal tree}
\label{fig:normal_tree_multi_system}
\end{subfigure}
\caption{Linear graph and the normal tree of the electro-mechanical system, as referred to in Fig. 5f}
\end{figure}

Third, the normal tree facilitates the statement of the electro-mechanical system's constitutive, continuity, and compatibility laws. The constitutive laws are:  
\begin{align}
\frac{d\omega_J}{dt} &= \frac{1}{J} \tau_J \label{eq:elecmech_const_1} \\
\frac{di_L}{dt} &= \frac{1}{L} V_L \label{eq:elecmech_const_2} \\
V_R &= R i_R \label{eq:elecmech_const_3} \\
\tau_B &= B \omega_B \label{eq:elecmech_const_4} \\
V_1 &= \frac{1}{K_a} \omega_2 \label{eq:elecmech_const_5} \\
\tau_2 &= -\frac{1}{K_a} i_1 \label{eq:elecmech_const_6}
\end{align}
The continuity laws are:
\begin{align}
\tau_J &= -\tau_2 - \tau_B \label{eq:elecmech_cont_1} \\
i_R &= i_L \label{eq:elecmech_cont_2} \\
i_1 &= i_L \label{eq:elecmech_cont_3}
\end{align}
The compatibility laws are:
\begin{align}
V_L &= V_s - V_1 - V_R \label{eq:elecmech_comp_1} \\
\omega_2 &= \omega_J \label{eq:elecmech_comp_2} \\
\omega_B &= \omega_J \label{eq:elecmech_comp_3}
\end{align}

Finally, these equations are algebraically simplified to the state space model in Eq. \ref{eq:elecmech_state}.
\begin{equation}
\frac{d}{dt} \begin{bmatrix} \omega_J \\ i_L \end{bmatrix} = \begin{bmatrix} -\frac{B}{J} & \frac{1}{J K_a} \\ -\frac{1}{K_a L} & -\frac{R}{L} \end{bmatrix} \begin{bmatrix} \omega_J \\ i_L \end{bmatrix} + \begin{bmatrix} 0 \\ \frac{1}{L} \end{bmatrix} V_s(t) \label{eq:elecmech_state}
\end{equation}

\section{Bond Graphs by Example}
\label{subsec:Bond_graph_by_example}
In order to concretely describe the relationship between linear graphs, bond graphs, and hetero-functional graphs, the same illustrative examples are now modeled using the bond graph methodology.  According to Fig. \ref{fig:different_graphs}, and the overview provided in Sec \ref{subsec:bond_graphs}, the bond graph methodology follows these three main steps:
\begin{enumerate}
\item Construct the bond graph from the identified system elements.
\item State the 0-junction, 1-junction, and constitutive laws of the system using bond graph junctions.
\item Simplify the above-mentioned laws into a single state space model.
\end{enumerate}
This section follows each of these three steps for the six illustrative examples identified in Fig. \ref{fig:illustrative_examples}.  

\subsection{Electrical System}
First, the bond graph associated with the electrical system illustrated in Fig. \ref{fig:illustrative_examples}a is shown in Fig. \ref{fig:bond_graph_electrical_system}. According to Fig. \ref{fig:across_through_variables}, in the electrical system domain, voltage $V$ is the effort variable, and current $i$ is the flow variable.  Additionally, as shown in Fig. \ref{fig:graph_elements}, electrical resistors, capacitors, and inductors are categorized as generalized resistors, capacitors, and inductors respectively. Furthermore, in bond graphs, the state variables of this system are the effort variables of generalized C elements ($V_{C_1}$), and the flow variables of generalized inductors ($i_{L_1}$ and $i_{L_2}$).  
\begin{figure}[H]
\centering
\includegraphics[width=0.45\textwidth]{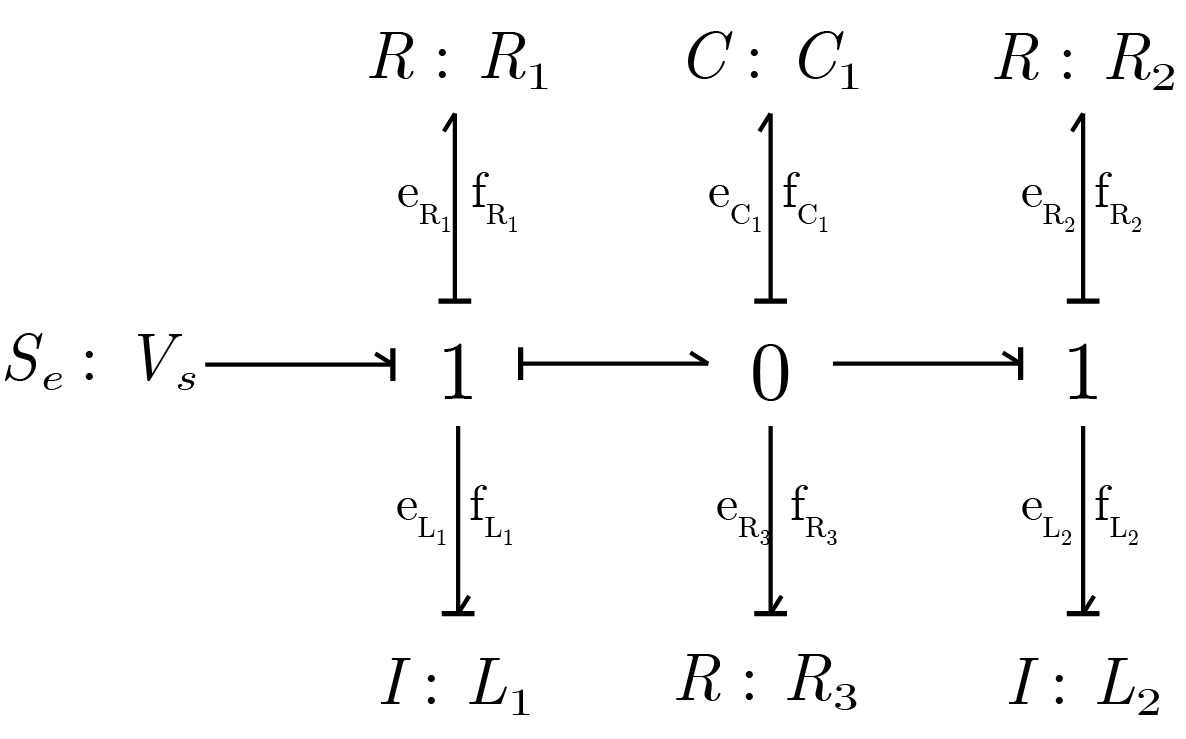} 
\caption{Bond graph of the electrical system, as referred to in Fig. \ref{fig:illustrative_examples}a}
\label{fig:bond_graph_electrical_system}
\end{figure}
\noindent Second, the constitutive laws in Eq. \ref{eq:electrical_const_1}-\ref{eq:electrical_const_6} are retained, Eq. \ref{eq:electrical_cont_1}-\ref{eq:electrical_cont_3} are adopted as 0-junction laws, and Eq. \ref{eq:electrical_comp_1}-\ref{eq:electrical_comp_3} are adopted as 1-junction laws. Finally, these laws are simplified algebraically to produce the state space model in Eq. \ref{eq:electrical_state}.  

\subsection{Translational Mechanical System}

The bond graph methodology is also applied to the translational mechanical system shown in Fig. \ref{fig:illustrative_examples}b. The bond graph associated with this system is illustrated in \ref{fig:bond_graph_translational_mechanical_system}. According to Fig. \ref{fig:across_through_variables}, in the translational mechanical system domain, force $F$ is the effort variable, and velocity $V$ is the flow variable.  Referring back to Fig. \ref{fig:graph_elements}, translational dampers, springs, and masses are categorized as generalized resistors, capacitors, and inductors respectively. The state variables of this system are the effort variables of generalized C elements ($F_{k_1}$ and $F_{k_2}$), and the flow variable of generalized inductors ($V_{m _1}$ and $V_{m_2}$).
\begin{figure}[H]
\centering
\includegraphics[width=0.55\textwidth]{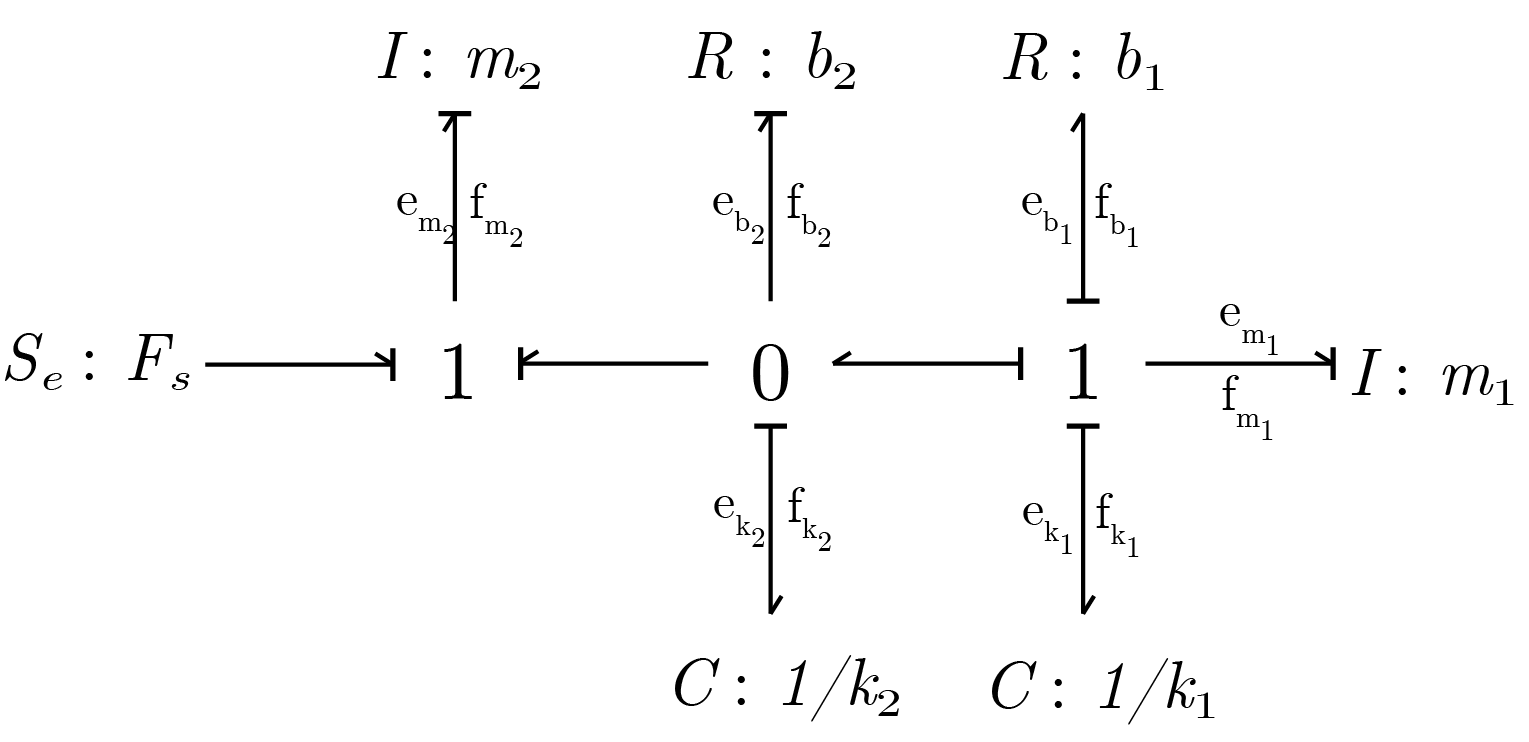} 
\caption{Bond graph of the translational mechanical system, as referred to in Fig. \ref{fig:illustrative_examples}b}
\label{fig:bond_graph_translational_mechanical_system}
\end{figure}
\noindent After the bond graph is derived, the constitutive laws in Eq. \ref{eq:trans_const_1}-\ref{eq:trans_const_6} are retained, Eq. \ref{eq:trans_comp_1}-\ref{eq:trans_comp_3} are adopted as 0-junction laws, and Eq. \ref{eq:trans_cont_1}-\ref{eq:trans_comp_3} are adopted as 1-junction laws.  Finally, these laws are simplified algebraically to produce the state space model in Eq. \ref{eq:trans_state}.  

\subsection{Rotational Mechanical System}

The bond graph associated with the rotational mechanical system illustrated in Fig. \ref{fig:illustrative_examples}c is shown in Fig. \ref{fig:bond_graph_rotational_mechanical_system}. Fig. \ref{fig:across_through_variables} shows that in the rotational mechanical system domain, torque $\tau$ is the effort variable, and angular velocity $\omega$ is the flow variable.  Additionally, as shown in Fig. \ref{fig:graph_elements}, rotational dampers, springs, and disks are categorized as generalized resistors, capacitors, and inductors respectively. Furthermore, the state variables of this system are the effort variable of generalized C elements ($\tau_{K}$), and the flow variable of generalized inductors ($\omega_J$).

\begin{figure}[H]
\centering
\includegraphics[width=0.4\textwidth]{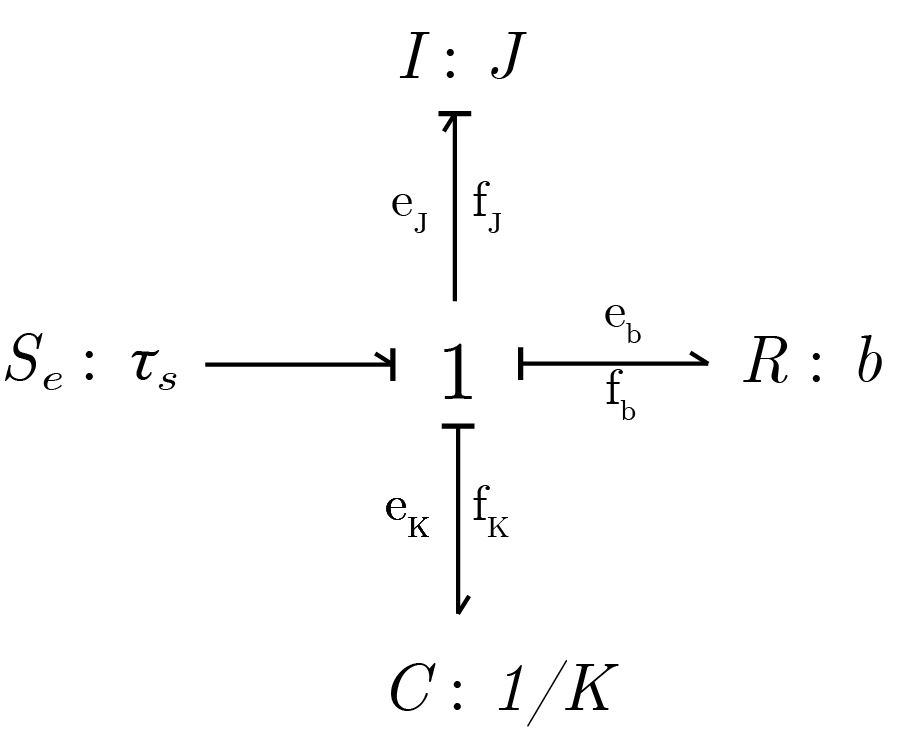} 
\caption{Bond graph of the rotational mechanical system, as referred to in Fig. \ref{fig:illustrative_examples}c}
\label{fig:bond_graph_rotational_mechanical_system}
\end{figure}

\noindent Consequently, the constitutive laws in Eq. \ref{eq:rot_const_1}-\ref{eq:rot_const_3} are retained, Eq. \ref{eq:rot_comp_1}-\ref{eq:rot_comp_2} are adopted as 0-junction laws, and Eq. \ref{eq:rot_cont_1} is adopted as 1-junction law.  Finally, these laws are simplified algebraically to produce the state space model in Eq. \ref{eq:rot_state}.  

\subsection{Fluidic System}
Fig. \ref{fig:bond_graph_fluidic_system} is the associated bond graph of the fluidic system illustrated in Fig. \ref{fig:illustrative_examples}d. According to Fig. \ref{fig:across_through_variables}, in the fluidic system domain, pressure $P$ is the effort variable, and volumetric flow rate $\dot{\cal V}$ is the flow variable.  Additionally, as shown in Fig. \ref{fig:graph_elements}, fluid resistance of pipes or valves, fluid tanks, and fluid inertances are categorized as generalized resistors, capacitors, and inductors respectively. Additionally, the state variables of this fluidic system are the effort variable of generalized C elements ($P_{C_1}$ and $P_{C_2}$), and the flow variable of generalized inductors ($\dot{\cal V_{\cal I}}$).

\begin{figure}[H]
\centering
\includegraphics[width=0.5\textwidth]{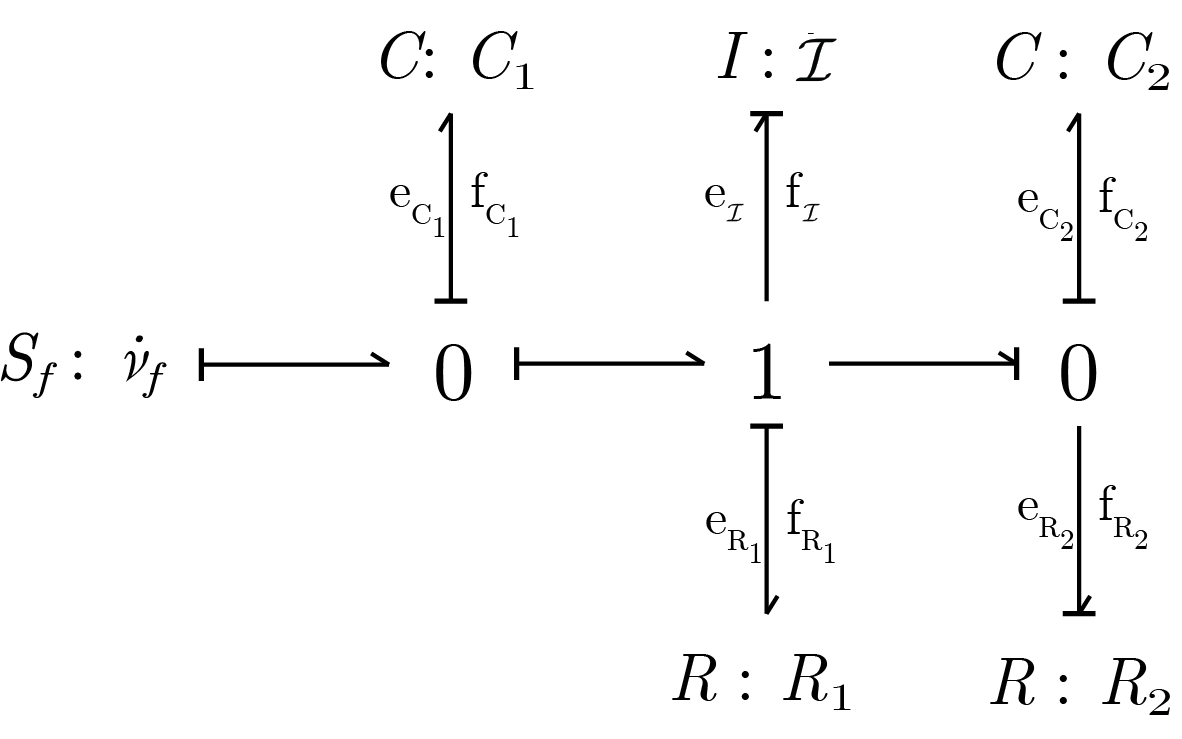} 
\caption{Bond graph of the fluidic system, as referred to in Fig. \ref{fig:illustrative_examples}d}
\label{fig:bond_graph_fluidic_system}
\end{figure}

\noindent After constructing the bond graph, the constitutive laws in Eq. \ref{eq:fluid_const_1}-\ref{eq:fluid_const_5} are retained, Eq. \ref{eq:fluid_cont_1}-\ref{eq:fluid_cont_3} are adopted as 0-junction laws, and Eq. \ref{eq:fluid_comp_1}-\ref{eq:fluid_comp_2} are adopted as 1-junction laws. Finally, these laws are simplified algebraically to produce the state space model in Eq. \ref{eq:fluid_state}.  

\subsection{Thermal System}
First, the bond graph associated with the thermal system and its analogous electrical circuit illustrated in Fig. \ref{fig:illustrative_examples}e and \ref{fig:thermal_system_electrical_circuit}, are shown in Fig. \ref{fig:bond_graph_thermal_system}. According to Fig. \ref{fig:across_through_variables}, in the thermal system domain, temperature $T$ is the effort variable, and heat flow rate $\dot{Q}$ is the flow variable.  Additionally, as shown in Fig. \ref{fig:graph_elements}, thermal resistances, and thermal capacitances are categorized as generalized resistors and capacitors respectively. Next, the state variables of this system are the effort variables of generalized C elements ($T_{C_i}$ and $T_{C_h}$).

\begin{figure}[H]
\centering
\includegraphics[width=0.45\textwidth]{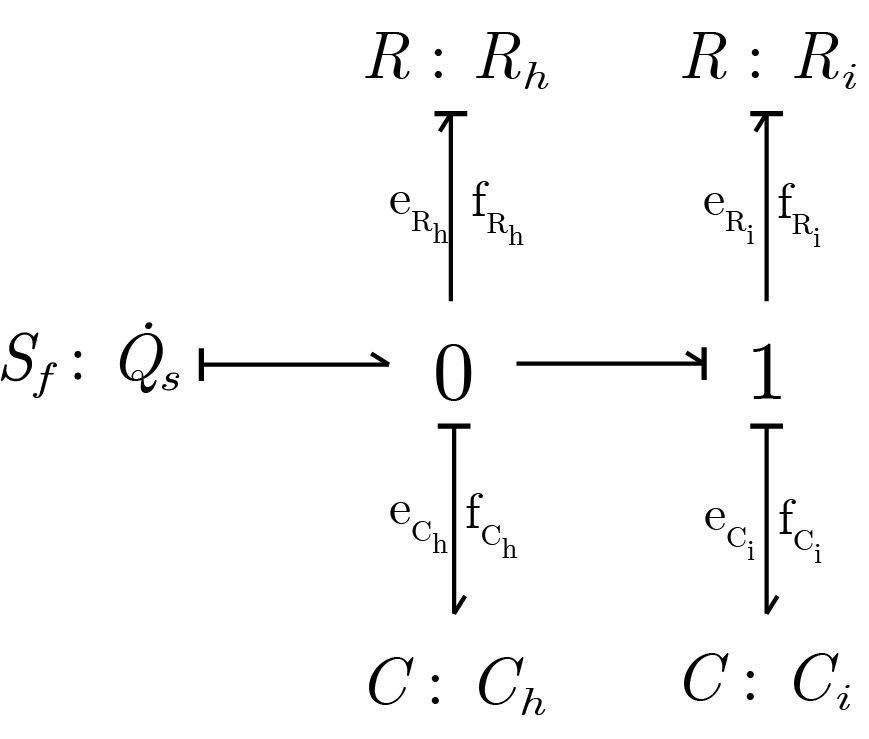} 
\caption{Bond graph of the thermal system, as referred to in Fig. \ref{fig:illustrative_examples}e}
\label{fig:bond_graph_thermal_system}
\end{figure}

\noindent Second, the constitutive laws in Eq. \ref{eq:thermal_const_1}-\ref{eq:thermal_const_4} are retained, Eq. \ref{eq:thermal_cont_1}-\ref{eq:thermal_cont_2} are adopted as 0-junction laws, and Eq. \ref{eq:thermal_comp_1}-\ref{eq:thermal_comp_2} are adopted as 1-junction laws. Finally, these laws are simplified algebraically to produce the state space model in Eq. \ref{eq:thermal_state}.  

\subsection{Multi-Energy System}

The bond graph methodology is applied similarly to the electro-mechanical system shown in Fig. \ref{fig:illustrative_examples}f. The bond graph associated with this system is shown in Fig. \ref{fig:bond_graph_electrical_system}. As electro-mechanical systems are a combination of electrical and (rotational) mechanical systems, voltage $V$ and torque $\tau$ are the effort variables, and current $i$ and angular velocity $\omega$ are the flow variables.  Additionally, as shown in Fig. \ref{fig:graph_elements}, the electrical resistors and the rotational dampers are generalized resistors. The electrical capacitors and rotational springs are generalized capacitors. The electrical inductors and rotational disks are categorized as generalized inductors. Furthermore, due to the different views in bond graphs, a generalized gyrator connects the electrical subsystem to the mechanical subsystem. Lastly, As mentioned in the electrical and rotational system examples, ($\omega_{J}$ and $i_{L}$ and $i_{L_2}$) are the state variables of this system.
\begin{figure}[H]
\centering
\includegraphics[width=0.55\textwidth]{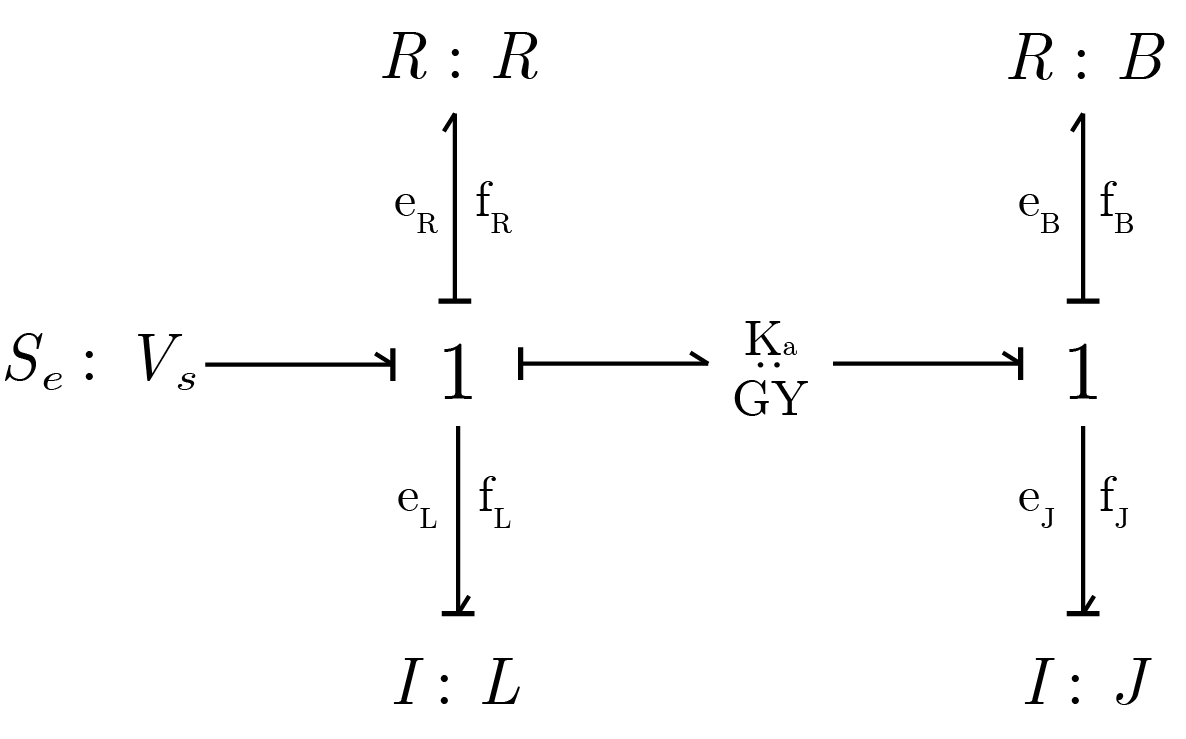} 
\caption{The Bond graph of electro-mechanical system showed in  Figure \ref{fig:illustrative_examples}(f)}
\label{fig:bond_graph_multi_system}
\end{figure}

\noindent Next, the constitutive laws in Eq. \ref{eq:elecmech_const_1}-\ref{eq:elecmech_const_6} are retained and Eq. \ref{eq:elecmech_cont_2}-\ref{eq:elecmech_cont_3} and Eq. \ref{eq:elecmech_comp_2}-\ref{eq:elecmech_comp_3} are adopted as 1-junction laws. Finally, these laws are simplified algebraically to produce the state space model in Eq. \ref{eq:elecmech_state}.  

\section{Hetero-functional Graphs by Example}
\label{subsec:HFGT_by_example}
In order to continue to concretely describe the relationship between linear graphs, bond graphs, and hetero-functional graphs, the same illustrative examples are now modeled using hetero-functional graph theory.  According to Fig. \ref{fig:different_graphs}, and the overview provided in Sec \ref{subsec:HFGT}, the hetero-functional graph methodology follows three main steps:  
\begin{enumerate}
\item Identify the system resources, the system processes, and their associated capabilities following Defn. \ref{Defn:D1}-\ref{defn:capabilityCh7}
\item Construct the engineering system net (and operand net if necessary) following Defn. \ref{Defn:D6}-\ref{Defn:ESN-STF}.
\item Setup and solve the hetero-functional network minimum cost flow problem as stated in Eq. \ref{eq:dummy_objective_function}-\ref{eq:device_model_initial} below.
\end{enumerate}

Before proceeding with derivation for each of the illustrative examples, it is important to recognize that linear graphs and bond graphs make several inherent, and limiting assumptions that are not made in hetero-functional graph theory by default.  
\begin{itemize}
\item $X[k] \in \mathbb{R} \quad \forall k \in \{1 \ldots K\}$. In physical systems, the primary decision variables are in the domain of real numbers.
\item $Y[k] \in \mathbb{R} \quad \forall k \in \{1 \ldots K\}$. 
Auxiliary decision variables are also in the domain of real numbers. 
\item $Z=0$. Linear graphs and bond graphs solve a set of simultaneous differential algebraic equations and do not require optimization.  Consequently, a dummy objective function is defined. 
\item $\Delta T \rightarrow 0$.  Linear graphs and bond graphs model differential algebraic equations where the simulation time step is infinitesimal.
\item $k_{d\psi}=0\  \quad \forall k, \forall \psi$.  The duration of each capability (Defn. \ref{defn:capabilityCh7}) is instantaneous.  Consequently, Eq. \ref{Eq:DurationConstraint} becomes:
\begin{align}\label{eq:simplicity_noDelay}
U_\psi^+[k] = U_\psi^-[k] \quad \forall k \in \{1, \dots, K\}
\end{align}
Additionally, Eq. \ref{Eq:ESN-STF2} collapses to triviality and is eliminated.
\item $Q_B[k] = Q_B[k+1] \quad \forall k \in \{1 \ldots K\}$. The engineering system does not accumulate operands at its buffers (Defn. \ref{defn:BSCh7}).  Furthermore, it is important to recognize that the above treatment of linear graphs and bond graphs only uses power variables (i.e. effort and flow pairs), and $Q_B[k]$ is a displacement variable in the Eulerian via and a momentum variable in the Lagrangian view.  Therefore, a 1-to-1 comparison of hetero-functional graphs to linear and bond graphs will not require the $Q_B$ variable.  Consequently, Eq. \ref{Eq:ESN-STF1} becomes:
\begin{align}
M U[k] \Delta T = 0 \quad \forall k \in \{1, \dots, K\}
\end{align}
where $M = M^+ - M^-$.  
\item $S_{l_i} = \emptyset$. ${\cal E}_{l_i} = \emptyset$. ${\cal N}_{l_i} = \emptyset$ (Defn. \ref{Defn:OperandNet}) All of the operands used in linear graphs and bond graphs have no state evolution and do not require their associated operand nets.  Consequently, Eqs. \ref{Eq:OperandNet-STF1}, \ref{Eq:OperandNet-STF2}, and \ref{Eq:OperandNetDurationConstraint} are eliminated. Similarly, without any operand net, there is no need for synchronization with the engineering system net firing vectors.  Consequently, Eqs.  \ref{Eq:SyncPlus}, \ref{Eq:SyncMinus} are eliminated as well.

\item The engineering system net boundary condition constraint in Equation \ref{CH6:eq:HFGTprog:comp:Bound} applies only when the engineering system has through-variable sources.  In such cases, and in light of the above, the boundary condition constraint becomes:
\begin{align}\label{eq:boundary_condition}
D_{U}.U[k] = C_{U}[k] \quad \forall k \in \{1, \dots, K\} 
\end{align}
Furthermore, the engineering system net boundary condition constraint is used to capture any initial conditions on the engineering system net firing vector. 
\begin{align}\label{eq:init_condition}
D_{Ui}.U[1] = C_{Ui}[1]  
\end{align}
\item Without any operand net, its boundary condition constraint in Equation \ref{Eq:OperandRequirements} is eliminated. 
\item The initial condition constraint in Eq \ref{CH6:eq:HFGTprog:comp:Init} is also eliminated as $Q_B$ is not retained as a decision variable. 
\item The final condition constraint in Eq. \ref{CH6:eq:HFGTprog:comp:Fini} is also eliminated as $Q_B$ is not retained as a decision variable.  Furthermore, all of the linear graph and bond graph models described above are initial value (rather than final value) problems. 
\item $\underbar{E}_{CP} \rightarrow -\infty$, $\overline{E}_{CP} \rightarrow \infty$.  The linear graph and bond graph methodologies do not place lower or upper bounds on the primary decision variables. Consequently, the inequality constraints on primary decision variables in Eq. \ref{ch6:eq:QPcanonicalform:3} are eliminated. 
\item The device model functions $g(X,Y)$ in Eq. \ref{Eq:DeviceModels} become the engineering system's constitutive laws. As elaborated below, they relate the engineering system net's primary variables (i.e. through variables) to its auxiliary variables (i.e. across variables).  
\item The device model function $h(Y)$ in Eq. \ref{Eq:DevicModels2} places bounds on the engineering system net's auxiliary variables.  Such constraints are used in linear graphs and bond graphs to impose times series from across variable sources.
\begin{align}
h(y[k]) = C_Y[k] \qquad \forall k \in \{1, \dots, K\}
\end{align}
They are also used to impose initial conditions on the auxiliary variables. 
\begin{align}\label{eq:aux_vars_init_condition}
h_i(y[1]) = C_{yi}[1]  
\end{align}
\end{itemize}
In summary, the HFNMCF problem stated in Eq. \ref{Eq:ObjFunc1}-\ref{Eq:DevicModels2} collapses to the following optimization problem in the context of linear and bond graph models.  

\begin{alignat}{3}
\text{minimize } Z &= 0 \label{eq:dummy_objective_function}\\
\text{s.t. }MU[k]\Delta T &= 0 &&  \quad \forall k \in \{1, \dots, K\} \label{eq:HFGT_continuity}\\ 
D_{U}.U[k] &= C_{U}[k] && \quad \forall k \in \{1, \dots, K\} 
\label{eq:HFGT_priVar_source}\\
D_{Ui}.U[1] &= C_{U}[1]
\label{eq:HFGT_priVar_initial}\\
g(X,Y)&=0\label{eq:device_model}\\
h(Y) &= C_Y[k] && \quad \forall k \in \{1, \dots, K\}\label{eq:device_model_source} \\
h_i(y[1]) & = C_{yi}[1] \label{eq:device_model_initial}
\end{alignat}

As mentioned above, the device model constraint in the HFNMCF problem shown first in Eq. \ref{Eq:DeviceModels} and now in Eq. \ref{eq:device_model} represents the constitutive laws in linear graphs and bond graphs.  As there are only a small number of generalized elements (e.g. resistor, inductor, capacitor, transformer, and gyrator), it is worthwhile recognizing that these generalized elements take on generic forms in the context of hetero-functional graph theory.  For generalized resistors: 
\begin{align} \label{eq:R_element_aux_eq}
S_R \cdot U[k] &= {\cal Z}_R\cdot S_R \cdot (-M)^T \cdot y[k] \quad  \forall k \in \{1, \dots, K\}
\end{align}
where $S_R$ is a projection operator that serves to select out the relevant primary (i.e. through) variables from the engineering system net firing vector.  Furthermore, ${\cal Z}_R$ is a diagonal matrix of resistance values. The engineering system net incidence matrix adopts a negative sign so that the across variables $y$ lose magnitude in the direction of flow.  Note that Eq. \ref{eq:R_element_aux_eq}, quite appropriately, is an algebraic relation between across and through variables at the same discrete time step $k$.  Next, for generalized inductors (in the Eulerian view, and generalized capacitors in the Lagrangian view):
\begin{align}\label{eq:L_element_aux_eq}
S_L \cdot (U[k+1]-U[k])\Delta T &= {\cal Z}_L \cdot S_L \cdot (-M)^T \cdot y[k] \quad \forall k \in \{1, \dots, K-1\}
\end{align}
where $S_L$  is a projection operator that serves to select out the relevant primary (i.e. through) variables from the engineering system net firing vector. Furthermore, ${\cal Z}_L$ is a diagonal matrix of inductance values. Note that Eq. \ref{eq:L_element_aux_eq}, quite appropriately, is a difference equation that results from discretizing the generalized inductor law via an Euler transformation.  Next, for generalized capacitors (in the Eulerian view and generalized inductors in the Lagrangian view): 
\begin{align}\label{eq:C_element_aux_eq}
S_C \cdot U[k] &= {\cal Z}_C \cdot S_C \cdot (-M)^T \cdot (y[k+1]-y[k]) \cdot \Delta t \quad \forall k \in \{1, \dots, K-1\}
\end{align}
where $S_C$ is a projection operator that serves to select out the relevant primary (i.e. through) variables from the engineering system net firing vector.  Furthermore, ${\cal Z}_C$ is a diagonal matrix of capacitance values. Once again, Eq. \ref{eq:C_element_aux_eq} results from the discretization of the generalized capacitance law via n Euler transformation.  Next, for generalized transformers:
\begin{align}
S_{T} y[k] &= 0 \quad \forall k \in \{1, \dots, K\} \label{eq:transformer_element_aux_eq}
\end{align}
where $S_{T_1}$ is the projection operator that serves to select out the transformer's relevant auxiliary (i.e. across) variables from the engineering system net. Finally, for generalized gyrators:
\begin{align}
S_{G_1}U[k] = S_{G_2} y[k] \quad \forall k \in \{1, \dots, K\}
\label{eq:gyrator_element_aux_eq}
\end{align}
where $S_{G_1}$ and $S_{G_2}$ are the projection operators that serve to select out the gyrator's relevant primary (i.e. through) and auxiliary (i.e. across) variables and from the engineering system net.

Given this specialized application of hetero-functional graphs to linear graphs and bond graphs, each of the six illustrative examples shown in Fig. \ref{fig:illustrative_examples} can be solved as the HFNMCF problem using the three steps identified at the top of the section.   

\subsection{Electrical System}
The first step is to recognize that the electrical system shown in Fig. \ref{fig:illustrative_examples}a is first, a specialization into the electrical domain, followed by an instantiation of the engineering system meta-architecture in Fig. \ref{fig:LFESMetaArchitecture}.  Consequently, Defn. \ref{Defn:D1}-\ref{defn:capabilityCh7} are understood as follows.  There are five electrical points that have distinct absolute values of across-variables that serve as independent buffers: $V_S$, $V_{R_L}$, $V_{C_1}$, $V_{L_2}$ and $V_0$.  Additionally, there is one across-variable source $V_S$ that serves as a transformation resource.  Additionally, the transportation resources include generalized resistors $R_1$, $R_2$, $R_3$, generalized inductors $L_1$, $L_2$, and a generalized capacitor $C_1$.  Fig. \ref{fig:LFESMetaArchitecture} shows that each of these transformation and transportation resources has exactly one system process; inject power with imposed through variable, dissipate power, store potential energy, and store kinetic energy.  The result is that each of these transformation and transportation resources introduces exactly one system capability with a primary through variable and an auxiliary across variable as attributes.  

In the next step, the electrical system shown in Fig. \ref{fig:illustrative_examples}a is transformed into the engineering system net shown in Fig. \ref{fig:PN_Electrical_system}. 
\begin{figure}[H]
\centering
\includegraphics[width=0.7\textwidth]{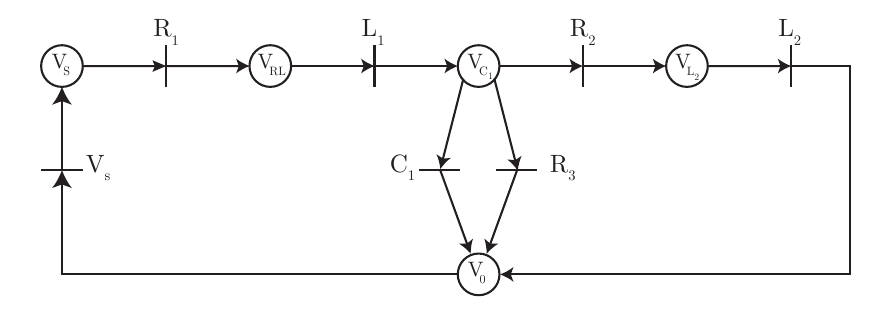} 
\caption{Engineering system net for the electrical system shown in Fig. \ref{fig:illustrative_examples}a}
\label{fig:PN_Electrical_system}
\end{figure}
While it is possible to produce Fig. \ref{fig:PN_Electrical_system} from the electrical circuit in Fig. \ref{fig:illustrative_examples}a by visual inspection; such an approach falls apart in systems with multiple energy domains or capabilities with multiple inputs and outputs.  Instead, the Engineering System Net and its state transition function is constructed according to Defn. \ref{Defn:ESN} and \ref{Defn:ESN-STF}.  Note that the hetero-functional graph theory toolbox \cite{Thompson:2023:ISC-JR02} can automatically calculate the positive and negative hetero-functional incidence matrices from an XML input file that instantiates the information from Defn. \ref{Defn:D1} - \ref{defn:capabilityCh7}.   The Engineering System Net in Fig. \ref{fig:PN_Electrical_system} shows the system buffers as places, the capabilities as transitions, and the incidence between them.  

In the third step, the hetero-functional network minimum cost flow problem is set up and solved.  More specifically, Eq. \ref{eq:dummy_objective_function}-\ref{eq:device_model_initial} are written out explicitly.  

\begin{align}
\text{minimize } Z = 0 \label{eq:HFGT_electrical_dummy_target_function}
\end{align}
\begin{align}
\text{s.t. }
\label{eq:HFGT_electrical_continuity}
\begin{bmatrix}
+1 & -1 & 0 & 0 & 0 & 0 & 0\\
0 & +1 & -1 & 0 & 0 & 0 & 0 \\
0 & 0 & +1 & -1 & -1 & 0 & -1\\
0 & 0 & 0 & 0 & +1 & -1 & 0\\
\end{bmatrix}\begin{bmatrix}
 i_{V_S} \\
 i_{R_1} \\
 i_{L_1} \\
 i_{C_1} \\
 i_{R_2} \\
 i_{L_2} \\
 i_{R_3} \\
\end{bmatrix}[k]\Delta T = 0   \quad \forall k \in \{1, \dots, K\}
\end{align}
\begin{align} \label{eq:HFGT_electrical_priVars_initial}
\begin{bmatrix}
0 & 0 & 1 & 0 & 0 & 0 & 0\\
0 & 0 & 0 & 0 & 0 & 1 & 0
\end{bmatrix}
\begin{bmatrix}
i_{V_S} \\
i_{R_1} \\
i_{L_1} \\
i_{C_1} \\
i_{R_2} \\
i_{L_2} \\
i_{R_3}
\end{bmatrix}[1]
=
\begin{bmatrix}
0\\
0
\end{bmatrix}
\end{align}
\begin{align}
\label{eq:HFGT_electrical_device_model_R}\nonumber
\begin{bmatrix}
0 & 1 & 0 & 0 & 0 & 0 & 0 \\
0 & 0 & 0 & 0 & 1 & 0 & 0 \\
0 & 0 & 0 & 0 & 0 & 0 & 1 
\end{bmatrix}
\begin{bmatrix}
i_{V_S} \\
i_{R_1} \\
i_{L_1} \\
i_{C_1} \\
i_{R_2} \\
i_{L_2} \\
i_{R_3}
\end{bmatrix}[k]
=
\begin{bmatrix}
\frac{1}{R_1} & 0 & 0 \\
0 & \frac{1}{R_2} & 0 \\
0 & 0 & \frac{1}{R_3}
\end{bmatrix}
\begin{bmatrix}
0 & 1 & 0 & 0 & 0 & 0 & 0 \\
0 & 0 & 0 & 0 & 1 & 0 & 0 \\
0 & 0 & 0 & 0 & 0 & 0 & 1 
\end{bmatrix}
\begin{bmatrix}
-1 & 0 & 0 & 0 \\
+1 & -1 & 0 & 0 \\
0 & +1 & -1 & 0 \\
0 & 0 & +1 & 0 \\
0 & 0 & +1 & -1 \\
0 & 0 & 0 & -1 \\
0 & 0 & +1 & 0 
\end{bmatrix}
\begin{bmatrix}
V_S \\
V_{RL} \\
V_{C_1} \\
V_{L_2}
\end{bmatrix}[k] & \\
\forall k \in \{1, \dots, K\} & 
\end{align}
\begin{align}
\label{eq:HFGT_electrical_device_model_L}\nonumber
\begin{bmatrix}
0 & 0 & 1 & 0 & 0 & 0 & 0 & 0 \\
0 & 0 & 0 & 0 & 0 & 1 & 0 & 0
\end{bmatrix}
\left(
\begin{bmatrix}
i_{V_S} \\
i_{R_1} \\
i_{L_1} \\
i_{C_1} \\
i_{R_2} \\
i_{L_2} \\
i_{R_3}
\end{bmatrix}[k+1]
-
\begin{bmatrix}
i_{V_S}\\
i_{R_1} \\
i_{L_1} \\
i_{C_1} \\
i_{R_2} \\
i_{L_2} \\
i_{R_3}
\end{bmatrix}[k]
\right)
=
\begin{bmatrix}
\frac{1}{L_1} & 0 \\
0 & \frac{1}{L_2}
\end{bmatrix}
\begin{bmatrix}
0 & 0 & 1 & 0 & 0 & 0 & 0 \\
0 & 0 & 0 & 0 & 0 & 1 & 0
\end{bmatrix}
\begin{bmatrix}
-1 & 0 & 0 & 0 \\
+1 & -1 & 0 & 0 \\
0 & +1 & -1 & 0 \\
0 & 0 & +1 & 0 \\
0 & 0 & +1 & -1 \\
0 & 0 & 0 & -1 \\
0 & 0 & +1 & 0 
\end{bmatrix}
\begin{bmatrix}
V_S \\
V_{RL} \\
V_{C_1} \\
V_{L_2}
\end{bmatrix}[k]\Delta T &\\
\forall k \in \{1, \dots, K-1\} & 
\end{align}
\begin{align}
\label{eq:HFGT_electrical_device_model_C}\nonumber
\begin{bmatrix}
0 & 0 & 0 & 1 & 0 & 0
\end{bmatrix}
\begin{bmatrix}
i_{V_S} \\
i_{R_1} \\
i_{L_1} \\
i_{C_1} \\
i_{R_2} \\
i_{L_2} \\
i_{R_3}
\end{bmatrix}[k] \Delta T
= 
\begin{bmatrix}
C_1
\end{bmatrix}
\begin{bmatrix}
0 & 0 & 0 & 1 & 0 & 0 & 0
\end{bmatrix}
\begin{bmatrix}
-1 & 0 & 0 & 0 \\
+1 & -1 & 0 & 0 \\
0 & +1 & -1 & 0 \\
0 & 0 & +1 & 0 \\
0 & 0 & +1 & -1 \\
0 & 0 & 0 & -1 \\
0 & 0 & +1 & 0 
\end{bmatrix}
\left(
\begin{bmatrix}
V_S \\
V_{RL} \\
V_{C_1} \\
V_{L_2}
\end{bmatrix}[k+1]
-
\begin{bmatrix}
V_S \\
V_{RL} \\
V_{C_1} \\
V_{L_2}
\end{bmatrix}[k]
\right) &\\
\forall k \in \{1, \dots, K-1\} &
\end{align}
\begin{align}
\label{eq:HFGT_electrical_aux_Var_source}
\begin{bmatrix}
1 & 0 & 0 & 0 
\end{bmatrix}
\begin{bmatrix}
V_S \\
V_{RL} \\
V_{C_1} \\
V_{L_2}
\end{bmatrix}[k]
= 1   \quad \forall k \in \{1, \dots, K\}
\end{align}
\begin{align}
\label{eq:HFGT_electrical_auxVar_initial}
\begin{bmatrix}
0 & 0 & 1 & 0 
\end{bmatrix}
\begin{bmatrix}
V_S \\
V_{RL} \\
V_{C_1} \\
V_{L_2}
\end{bmatrix}[k=1]
= 0 
\end{align}

This explicit statement of the HFNMCF problem in the context of the electrical system shown in Fig. \ref{fig:illustrative_examples}a provides the following insights:
\begin{itemize}
\item Eq. \ref{eq:HFGT_electrical_dummy_target_function} shows that the electrical system does not have an objective function and is simply a set of simultaneous equations.  
\item Eq. \ref{eq:HFGT_electrical_continuity} is a matrix restatement of the continuity laws in Eq. \ref{eq:electrical_cont_1} - \ref{eq:electrical_cont_3}.  Note that Eq. \ref{eq:HFGT_electrical_continuity} introduces an additional matrix row to account for the current provided by the voltage source $i_{V_S}$.  While this variable is not required in the linear graph and bond graph methodologies, the HFGT derivation requires across and through variables for all capabilities.  Also note that Eq. \ref{eq:HFGT_electrical_continuity} does not include the current balance associated with the voltage ground $V_0$.  This is because incidence matrices of closed systems (i.e. circuits) have a rank of N-1\cite{Gould:2012:00} and so the last redundant equation must be eliminated to make Eq. \ref{eq:HFGT_electrical_continuity} full rank.  Finally, the hetero-functional incidence matrix $M$ can be automatically produced from the HFGT toolbox \cite{Thompson:2023:ISC-JR02} for systems of arbitrary size.
\item There are no equations that impose exogenous values on the currents because there are no current sources.  
\item Eq. \ref{eq:HFGT_electrical_priVars_initial} is the initial condition on the inductors' current as the state variables.  
\item Eq. \ref{eq:HFGT_electrical_device_model_R} is a matrix restatement of the constitutive law for resistors (i.e. Ohm's Law) in Eq. \ref{eq:electrical_const_2}, \ref{eq:electrical_const_4}, and \ref{eq:electrical_const_6}.
\item Eq. \ref{eq:HFGT_electrical_device_model_L} is a matrix restatement of the constitutive law for inductors in Eq. \ref{eq:electrical_const_3} and \ref{eq:electrical_const_5}.
\item Eq. \ref{eq:HFGT_electrical_device_model_C} is a matrix restatement of the constitutive law for capacitors in Eq. \ref{eq:electrical_const_1}.
\item The selector matrices in Eq. \ref{eq:HFGT_electrical_priVars_initial}-\ref{eq:HFGT_electrical_aux_Var_source} can be automatically produced from the HFGT toolbox \cite{Thompson:2023:ISC-JR02}for systems of arbitary size.  
\item Eq. \ref{eq:HFGT_electrical_aux_Var_source} imposes exogenous values on the voltages due to the presence of voltage sources.  
\item Eq. \ref{eq:HFGT_electrical_auxVar_initial} is the initial condition on the capacitor voltage as a state variable. 
\item The compatibility laws stated in Eq. \ref{eq:electrical_comp_1} - \ref{eq:electrical_comp_3} are superfluous because all of the voltages have been stated in absolute terms relative to the ground rather than as voltage differences between electrical points.  
\end{itemize}
Ultimately, the HFNMCF problem restates, in discrete time, and in matrix form, the simultaneous continuity and constitutive equations from the linear graph derivation and eliminates entirely the need for compatibility equations.  Furthermore, the HFNMCF problem reveals that only the continuity laws create relationships \emph{between} capabilities (Defn. \ref{defn:capabilityCh7}) in the engineering system.  The remaining constraints, those tied to initial conditions, exogenous values, and constitutive laws, address capabilities individually.  Consequently, the HFGT toolbox\cite{Thompson:2023:ISC-JR02} makes setting up the HFNMCF problem relatively straightforward because it can automatically generate the hetero-functional graph structure and keep track of the indices and types of each capability.  

Once the HFNMCF problem for the electrical system has been set up, it can be simulated straightforwardly and compared against the state space ODE model derived by linear graph and/or bond graph.  The following parameter values are chosen:
$R_1 = 200 \ \Omega$, $R_2 = 200 \ \Omega$, $R_3 = 220 \ \Omega$, $L_1 = 100 \ \text{mH}$, $L_2 = 150 \ \text{mH}$, $C_1 = 10 \ \mu\text{F}$, and $V_s = 1 \ \text{V}$ (step input voltage). The simulation time is $t = 0.01$ seconds, and the time step $\Delta T = 1e-4$ seconds.  The HFNMCF results for the primary (current) and auxiliary (voltage) decision variables are shown as solid lines in Fig. \ref{fig:HFGT_Electrical_system_results}.  The associated state space ODE results are shown in dashed lines with embedded triangles.  

\begin{figure}[H]
\centering
\begin{subfigure}[b]{0.45\textwidth}
\centering
\includegraphics[width=1\linewidth]{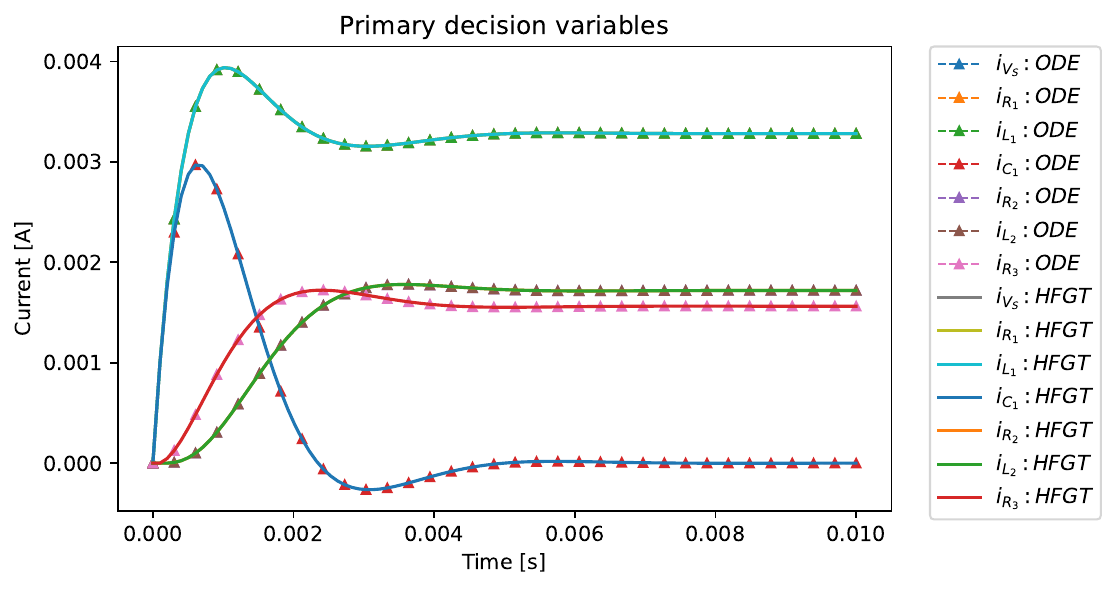} 
\caption{}
\label{fig:HFGT_Electrical_X}
\end{subfigure}
\hfill
\begin{subfigure}[b]{0.45\textwidth}
\centering
\includegraphics[width=1\linewidth]{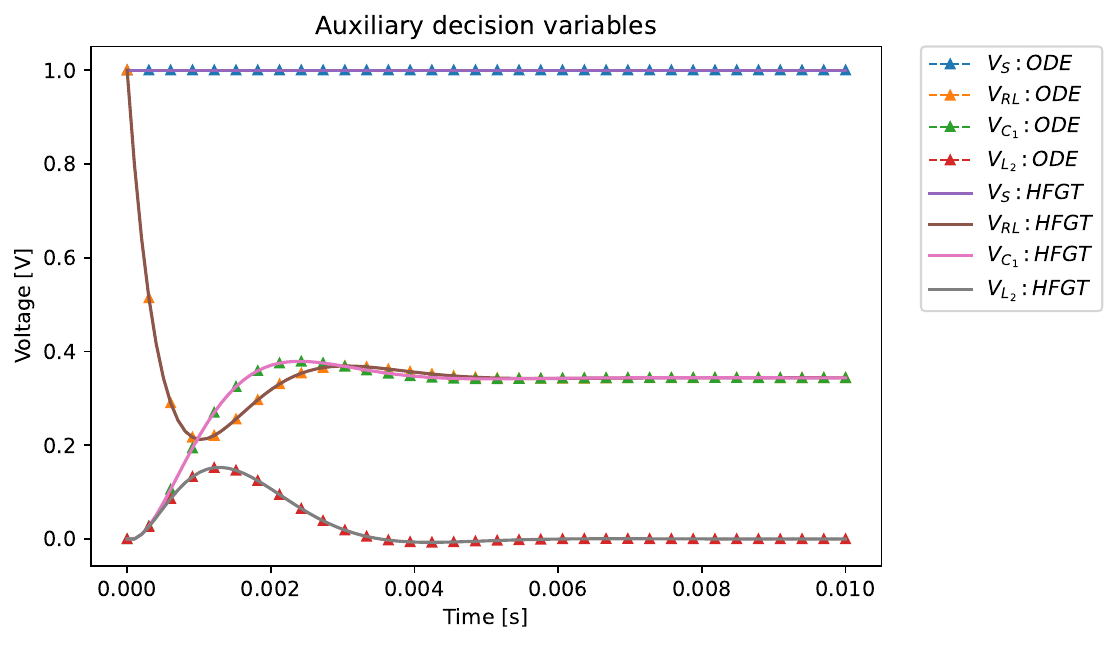}
\caption{}
\label{fig:HFGT_Electrical_Y}
\end{subfigure}
\caption{(a) Time series of primary decision variables of the electrical system. (b) Time series of auxiliary decision variables of the electrical system.  Results from the HFNMCF problem are shown in solid lines. Results from the State Space ODE model are shown in dashed lines with embedded triangles.}
\label{fig:HFGT_Electrical_system_results}
\end{figure}

\subsection{Translational Mechanical System}

The first step is to recognize that the translational mechanical system shown in Fig. \ref{fig:illustrative_examples}b is first, a specialization into the mechanical domain, followed by an instantiation of the engineering system meta-architecture in Fig. \ref{fig:LFESMetaArchitecture}.  Consequently, Defn. \ref{Defn:D1}-\ref{defn:capabilityCh7} are understood as follows.  There are four points that have distinct absolute values of across-variables that serve as independent buffers: $V_{m_1}$, $V_{k_2b_2}$, $V_{m_2}$ and $V_g$.  Additionally, there is one through-variable source $F_S$ that serves as a transformation resource.  Additionally, the transportation resources include generalized resistors $b_1$, $b_2$, generalized inductors $k_1$, $k_2$, and generalized capacitors $m_1$, $m_2$.  Fig. \ref{fig:LFESMetaArchitecture} shows that each of these transformation and transportation resources has exactly one system process; inject power with imposed through variable, dissipate power, store potential energy, and store kinetic energy.  The result is that each of these transformation and transportation resources introduces exactly one system capability with a primary through variable and an auxiliary across variable as attributes.  

In the next step, the translational mechanical system shown in Fig. \ref{fig:illustrative_examples}b is transformed into the engineering system net shown in Fig. \ref{fig:PN_Translational_mechanical_system}. 
\begin{figure}[H]
\centering
\includegraphics[width=0.7\textwidth]{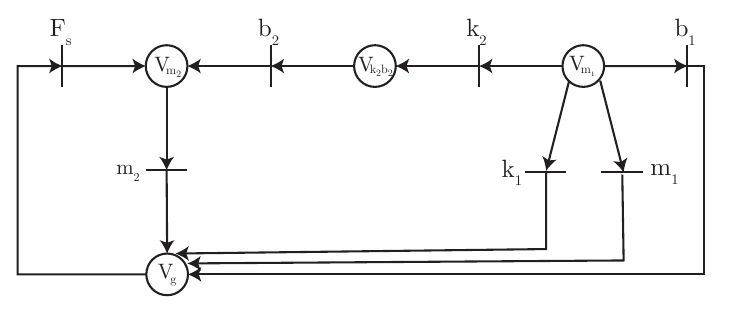} 
\caption{Engineering system net for the translational mechanical system shown in Fig. \ref{fig:illustrative_examples}b}
\label{fig:PN_Translational_mechanical_system}
\end{figure}
While it is possible to produce Fig. \ref{fig:PN_Translational_mechanical_system} from either the translational mechanical system in Fig. \ref{fig:illustrative_examples}b or the associated linear graph in Fig. \ref{fig:linear_graph_translational_mechanical_system} by visual inspection; such an approach falls apart in systems with multiple energy domains or capabilities with multiple inputs and outputs.  Instead, the Engineering System Net and its state transition function is constructed according to Defn. \ref{Defn:ESN} and \ref{Defn:ESN-STF}.  Note that the hetero-functional graph theory toolbox \cite{Thompson:2023:ISC-JR02} can automatically calculate the positive and negative hetero-functional incidence matrices from an XML input file that instantiates the information from Defn. \ref{Defn:D1} - \ref{defn:capabilityCh7}.   The Engineering System Net in Fig. \ref{fig:PN_Translational_mechanical_system} shows the system buffers as places, the capabilities as transitions, and the incidence between them.  

In the third step, the hetero-functional network minimum cost flow problem is set up and solved.  More specifically, Eq. \ref{eq:dummy_objective_function}-\ref{eq:device_model_initial} are written out explicitly.  

\begin{align}
\text{minimize } Z = 0 \label{eq:HFGT_translational_dummy_target_function}
\end{align}
\begin{align}
\text{s.t. }
\label{eq:HFGT_translational_continuity}
\begin{bmatrix}
+1 & -1 & +1 & 0 & 0 & 0 & 0\\
0 & 0 & -1 & +1 & 0 & 0 & 0 \\
0 & 0 & 0 & -1 & -1 & -1 & -1
\end{bmatrix}\begin{bmatrix}
F_{S}\\
F_{m_2} \\
F_{b_2} \\
F_{k_2} \\
F_{m_1} \\
F_{b_1} \\
F_{k_1}
\end{bmatrix}[k]\Delta T = 0   \quad \forall k \in \{1, \dots, K\}
\end{align}

\begin{align}
\label{eq:HFGT_translational_priVars_source}
\begin{bmatrix}
1 & 0 & 0 & 0 & 0 & 0 & 0
\end{bmatrix}
\begin{bmatrix}
F_{S}\\
F_{m_2} \\
F_{b_2} \\
F_{k_2} \\
F_{m_1} \\
F_{b_1} \\
F_{k_1}
\end{bmatrix}[k]
= 
1 \quad \forall k \in \{1, \dots, K\}
\end{align}

\begin{align} \label{eq:HFGT_translational_priVars_initial}
\begin{bmatrix}
0 & 0 & 0 & 1 & 0 & 0 & 0 \\
0 & 0 & 0 & 0 & 0 & 0 & 1 
\end{bmatrix}
\begin{bmatrix}
F_{S}\\
F_{m_2} \\
F_{b_2} \\
F_{k_2} \\
F_{m_1} \\
F_{b_1} \\
F_{k_1}
\end{bmatrix}[k=1]
= 
\begin{bmatrix}
0 \\
0
\end{bmatrix}
\end{align}
\begin{align}
\label{eq:HFGT_translational_device_model_R}\nonumber
\begin{bmatrix}
0 & 0 & 0 & 0 & 0 & 1 & 0 \\
0 & 0 & 1 & 0 & 0 & 0 & 0 \\
\end{bmatrix}
\begin{bmatrix}
F_{S}\\
F_{m_2} \\
F_{b_2} \\
F_{k_2} \\
F_{m_1} \\
F_{b_1} \\
F_{k_1}
\end{bmatrix}[k]
=
\begin{bmatrix}
b_1 & 0\\
0 & b_2
\end{bmatrix}
\begin{bmatrix}
0 & 0 & 0 & 0 & 0 & 1 & 0 \\
0 & 0 & 1 & 0 & 0 & 0 & 0 \\
\end{bmatrix}
\begin{bmatrix}
-1 & 0 & 0 \\
+1 & 0 & 0 \\
-1 & +1 & 0 \\
0 & -1 & +1 \\
0 & 0 & +1 \\
0 & 0 & +1 \\
0 & 0 & +1 \\
\end{bmatrix}
\begin{bmatrix}
V_{m_2} \\
V_{k_2b_2} \\
V_{m_1}
\end{bmatrix}[k] & \\
\forall k \in \{1, \dots, K\} & 
\end{align}
\begin{align}
\label{eq:HFGT_translational_device_model_L}\nonumber
\begin{bmatrix}
0 & 0 & 0 & 0 & 0 & 0 & 1 \\
0 & 0 & 0 & 1 & 0 & 0 & 0 \\
\end{bmatrix}
\left(
\begin{bmatrix}
F_{S}\\
F_{m_2} \\
F_{b_2} \\
F_{k_2} \\
F_{m_1} \\
F_{b_1} \\
F_{k_1} \\
\end{bmatrix}[k+1]
-
\begin{bmatrix}
F_{S}\\
F_{m_2} \\
F_{b_2} \\
F_{k_2} \\
F_{m_1} \\
F_{b_1} \\
F_{k_1} \\
\end{bmatrix}[k]
\right)
=
\begin{bmatrix}
k_1 & 0 \\
0 & k_2
\end{bmatrix}
\begin{bmatrix}
0 & 0 & 0 & 0 & 0 & 0 & 1 \\
0 & 0 & 0 & 1 & 0 & 0 & 0 \\
\end{bmatrix}
\begin{bmatrix}
-1 & 0 & 0 \\
+1 & 0 & 0 \\
-1 & +1 & 0 \\
0 & -1 & +1 \\
0 & 0 & +1 \\
0 & 0 & +1 \\
0 & 0 & +1 \\
\end{bmatrix}
\begin{bmatrix}
V_{m_2} \\
V_{k_2b_2} \\
V_{m_1}
\end{bmatrix}[k]\Delta T &\\
\forall k \in \{1, \dots, K-1\} & 
\end{align}
\begin{align}
\label{eq:HFGT_translational_device_model_C}\nonumber
\begin{bmatrix}
0 & 0 & 0 & 0 & 1 & 0 & 0 \\
0 & 1 & 0 & 0 & 0 & 0 & 0 \\
\end{bmatrix}
\begin{bmatrix}
F_{S}\\
F_{m_2} \\
F_{b_2} \\
F_{k_2} \\
F_{m_1} \\
F_{b_1} \\
F_{k_1} \\
\end{bmatrix}[k] \Delta T
= 
\begin{bmatrix}
m_1 & 0 \\
0 & m_2
\end{bmatrix}
\begin{bmatrix}
0 & 0 & 0 & 0 & 1 & 0 & 0 \\
0 & 1 & 0 & 0 & 0 & 0 & 0 \\
\end{bmatrix}
\begin{bmatrix}
-1 & 0 & 0 \\
+1 & 0 & 0 \\
-1 & +1 & 0 \\
0 & -1 & +1 \\
0 & 0 & +1 \\
0 & 0 & +1 \\
0 & 0 & +1 \\
\end{bmatrix}
\left(
\begin{bmatrix}
V_S \\
V_{R_1} \\
V_{C_1} \\
V_{L_2}
\end{bmatrix}[k+1]
-
\begin{bmatrix}
V_{m_2} \\
V_{k_2b_2} \\
V_{m_1}
\end{bmatrix}[k]
\right) &\\
\forall k \in \{1, \dots, K-1\} &
\end{align}

\begin{align}
\label{eq:HFGT_translational_auxVar_initial}
\begin{bmatrix}
1 & 0 & 0 \\
0 & 0 & 1 
\end{bmatrix}
\begin{bmatrix}
V_{m_2} \\
V_{k_2b_2} \\
V_{m_1}
\end{bmatrix}[k=1]
= 
\begin{bmatrix}
0 \\
0
\end{bmatrix}
\end{align}

This explicit statement of the HFNMCF problem in the context of the translational mechanical system shown in Fig. \ref{fig:illustrative_examples}b provides the following insights:
\begin{itemize}
\item Eq. \ref{eq:HFGT_translational_dummy_target_function} shows that the mechanical system does not have an objective function and is simply a set of simultaneous equations.  
\item Eq. \ref{eq:HFGT_translational_continuity} is a matrix restatement of the continuity laws (i.e. Newton’s First Law) in Eq. \ref{eq:trans_cont_1} - \ref{eq:trans_cont_3}.  Again, the force balance on the ground place is redundant and therefore eliminated.
\item Eq. \ref{eq:HFGT_translational_priVars_source} imposes exogenous values on the force due to the presence of force sources.  
\item Eq. \ref{eq:HFGT_translational_priVars_initial} is the initial condition on the spring force as a state variable.  
\item Eq. \ref{eq:HFGT_translational_device_model_R} is a matrix restatement of the constitutive law for mechanical dampers in Eq. \ref{eq:trans_const_5}, and \ref{eq:trans_const_6}.
\item Eq. \ref{eq:HFGT_translational_device_model_L} is a matrix restatement of the constitutive law for springs in Eq. \ref{eq:trans_const_3} and \ref{eq:trans_const_4}.
\item Eq. \ref{eq:HFGT_translational_device_model_C} is a matrix restatement of the constitutive law for masses in Eq. \ref{eq:trans_const_1} and \ref{eq:trans_const_2}.
\item There are no equations that impose exogenous values on the velocities because there are no velocity sources.
\item Eq. \ref{eq:HFGT_translational_auxVar_initial} is the initial condition on the mass velocities as state variables. 
\item The compatibility laws stated in Eq. \ref{eq:trans_comp_1} - \ref{eq:trans_comp_3} are superfluous because all of the velocities have been stated in absolute terms relative to the ground reference frame rather than as velocity differences between points.  
\end{itemize}

Once the HFNMCF problem for the translational mechanical system has been set up, it can be simulated straightforwardly and compared against the state space ODE model derived by linear graph and/or bond graph.  The following parameter values are chosen: $m_1 = 1 \ \text{kg}$, $m_2 = 2 \ \text{kg}$, $k_1 = 20 \ \text{N/m}$, $k_2 = 10 \ \text{N/m}$, $b_1 = 1 \ \text{Ns/m}$, $b_2 = 10 \ \text{Ns/m}$, and $F_s = 1 \ \text{N}$ (Step input force). The simulation time $K = 20$ seconds, and the time step $\Delta T = 0.02$ seconds.  The HFNMCF results for the primary (force) and auxiliary (velocity) decision variables are shown as solid lines in Fig. \ref{fig:HFGT_Translational_mechanical_system_results}.  The associated state space ODE results are shown in dashed lines with embedded triangles.  

\begin{figure}[H]
\centering
\begin{subfigure}[b]{.45\textwidth}
\centering
\includegraphics[width=1\linewidth]{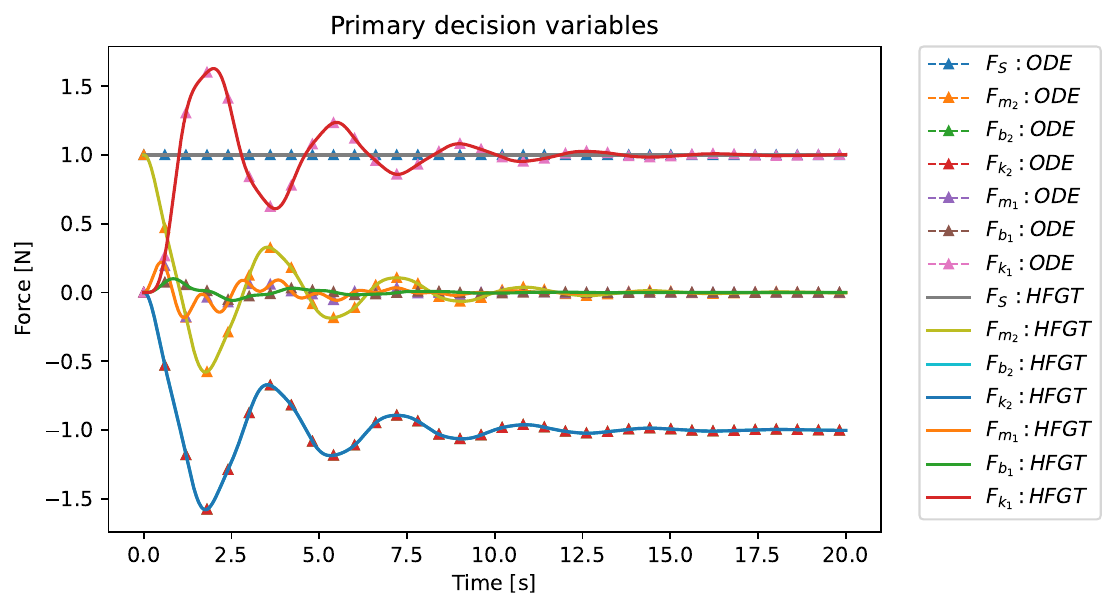} 
\caption{}
\label{fig:HFGT_Translational_mechanical_X}
\end{subfigure}
\hfill
\begin{subfigure}[b]{.45\textwidth}
\centering
\includegraphics[width=1\linewidth]{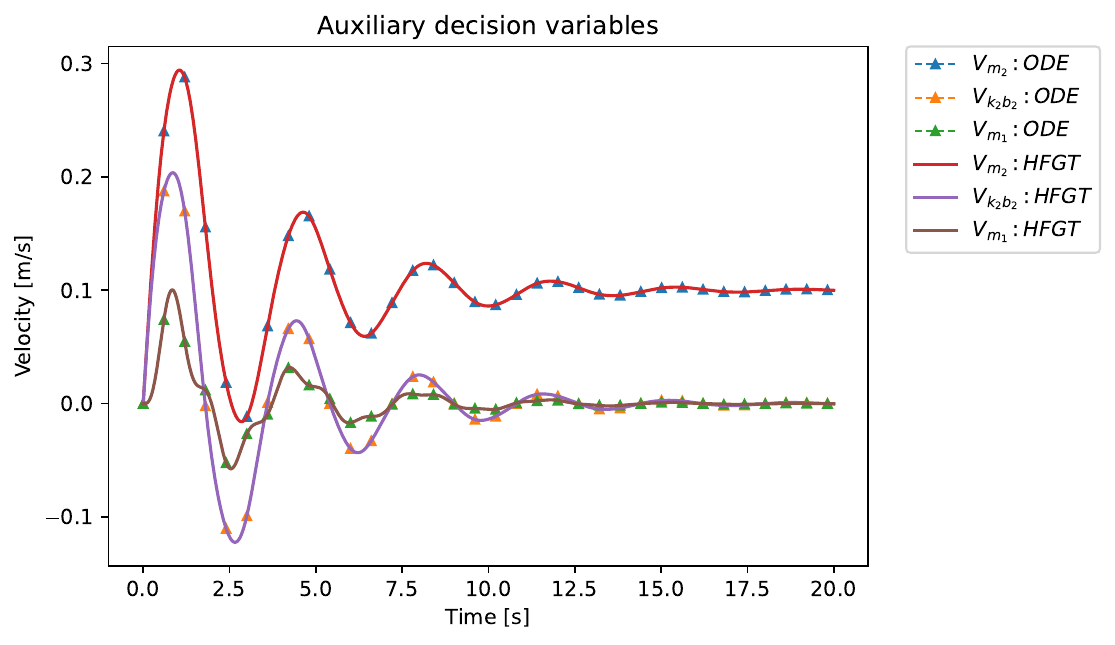}
\caption{}
\label{fig:HFGT_translational_mechanical_Y}
\end{subfigure}
\caption{(a) Time series of primary decision variables of the translational mechanical system. (b) Time series of auxiliary decision variables of the translational mechanical.  Results from the HFNMCF problem are shown in solid lines. Results from the State Space ODE model are shown in dashed lines with embedded triangles.}
\label{fig:HFGT_Translational_mechanical_system_results}
\end{figure}

\subsection{Rotational Mechanical System}
The first step is to recognize that the rotational mechanical system shown in Fig. \ref{fig:illustrative_examples}c is first, a specialization into the rotational mechanical domain, followed by an instantiation of the engineering system meta-architecture in Fig. \ref{fig:LFESMetaArchitecture}.  Consequently, Defn. \ref{Defn:D1}-\ref{defn:capabilityCh7} are understood as follows.  There are two points that have distinct absolute values of across-variables that serve as independent buffers: $\omega_J$, and $\omega_g$.  Additionally, there is one through-variable source $\tau_s$ that serves as a transformation resource.  Additionally, the transportation resources include a generalized resistor $b$, a generalized inductor $K$, and a generalized capacitor $J$.  Fig. \ref{fig:LFESMetaArchitecture} shows that each of these transformation and transportation resources has exactly one system process; inject power with imposed through variable, dissipate power, store potential energy, and store kinetic energy.  The result is that each of these transformation and transportation resources introduces exactly one system capability with a primary through-variable and an auxiliary across-variable as attributes.  

In the next step, the rotational mechanical system shown in Fig. \ref{fig:illustrative_examples}c is transformed into the engineering system net shown in Fig. \ref{fig:PN_Rotational_mechanical_system}. 
\begin{figure}[H]
\centering
\includegraphics[width=0.3\textwidth]{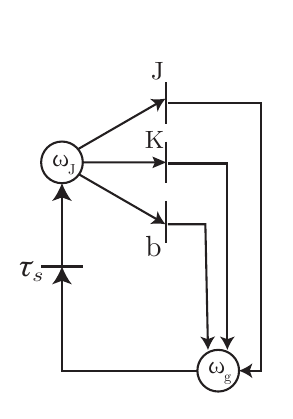} 
\caption{Petri-Net for the rotational mechanical system shown in Figure \ref{fig:illustrative_examples}(c)}
\label{fig:PN_Rotational_mechanical_system}
\end{figure}
\noindent While it is possible to produce Fig. \ref{fig:PN_Rotational_mechanical_system} from the rotational mechanical system in Fig. \ref{fig:illustrative_examples}c or the linear graph in Fig. \ref{X} by visual inspection; such an approach falls apart in systems with multiple energy domains or capabilities with multiple inputs and outputs.  Instead, the Engineering System Net and its state transition function is constructed according to Defn. \ref{Defn:ESN} and \ref{Defn:ESN-STF}.  Again, the hetero-functional graph theory toolbox \cite{Thompson:2023:ISC-JR02} automatically calculates the positive and negative hetero-functional incidence matrices from an XML input file that instantiates the information from Defn. \ref{Defn:D1} - \ref{defn:capabilityCh7}.   The Engineering System Net in Fig. \ref{fig:PN_Rotational_mechanical_system} shows the system buffers as places, the capabilities as transitions, and the incidence between them.  

In the third step, the hetero-functional network minimum cost flow problem is set up and solved.  More specifically, Eq. \ref{eq:dummy_objective_function}-\ref{eq:device_model_initial} are written out explicitly.  

\begin{align}
\text{minimize } Z = 0 \label{eq:HFGT_rotational_dummy_target_function}
\end{align}
\begin{align}
\text{s.t. }
\label{eq:HFGT_rotational_continuity}
\begin{bmatrix}
+1 & -1 & -1 & -1
\end{bmatrix}\begin{bmatrix}
\tau_{s} \\
\tau_{J} \\
\tau_{K} \\
\tau_{b}
\end{bmatrix}[k]\Delta T = 0   \quad \forall k \in \{1, \dots, K\}
\end{align}

\begin{align}
\label{eq:HFGT_rotational_priVars_source}
\begin{bmatrix}
1 & 0 & 0 & 0
\end{bmatrix}
\begin{bmatrix}
\tau_{s} \\
\tau_{J} \\
\tau_{K} \\
\tau_{b}
\end{bmatrix}[k]
= 1 \quad \forall k \in \{1, \dots, K\}
\end{align}

\begin{align} \label{eq:HFGT_rotational_priVars_initial}
\begin{bmatrix}
0 & 0 & 1 & 0
\end{bmatrix}
\begin{bmatrix}
\tau_{s} \\
\tau_{J} \\
\tau_{K} \\
\tau_{b}
\end{bmatrix}[k=1]
= 
0
\end{align}
\begin{align}
\label{eq:HFGT_rotational_device_model_R}\nonumber
\begin{bmatrix}
0 & 0 & 0 & 1
\end{bmatrix}
\begin{bmatrix}
\tau_{s} \\
\tau_{J} \\
\tau_{K} \\
\tau_{b}
\end{bmatrix}[k]
=
\begin{bmatrix}
b_\omega
\end{bmatrix}
\begin{bmatrix}
0 & 0 & 0 & 1
\end{bmatrix}
\begin{bmatrix}
-1 \\ +1 \\ +1 \\ +1
\end{bmatrix}
\begin{bmatrix}
\omega_J
\end{bmatrix}[k] & \\
\forall k \in \{1, \dots, K\} & 
\end{align}
\begin{align}
\label{eq:HFGT_rotational_device_model_L}\nonumber
\begin{bmatrix}
0 & 0 & 1 & 0 
\end{bmatrix}
\left(
\begin{bmatrix}
\tau_{s} \\
\tau_{J} \\
\tau_{K} \\
\tau_{b}
\end{bmatrix}[k+1]
-
\begin{bmatrix}
\tau_{s} \\
\tau_{J} \\
\tau_{K} \\
\tau_{b}
\end{bmatrix}[k]
\right)
=
\begin{bmatrix}
K_\omega
\end{bmatrix}
\begin{bmatrix}
0 & 0 & 1 & 0 
\end{bmatrix}
\begin{bmatrix}
-1 \\ +1 \\ +1 \\ +1
\end{bmatrix}
\begin{bmatrix}
\omega_J
\end{bmatrix}[k]\Delta T &\\
\forall k \in \{1, \dots, K-1\} & 
\end{align}
\begin{align}
\label{eq:HFGT_rotational_device_model_C}\nonumber
\begin{bmatrix}
0 & 1 & 0 & 0
\end{bmatrix}
\begin{bmatrix}
\tau_{s} \\
\tau_{J} \\
\tau_{K} \\
\tau_{b}
\end{bmatrix}[k] \Delta T
= 
\begin{bmatrix}
J
\end{bmatrix}
\begin{bmatrix}
0 & 1 & 0 & 0
\end{bmatrix}
\begin{bmatrix}
-1 \\ +1 \\ +1 \\ +1
\end{bmatrix}
\left(
\begin{bmatrix}
\omega_J
\end{bmatrix}[k+1]
-
\begin{bmatrix}
\omega_J
\end{bmatrix}[k]
\right) &\\
\forall k \in \{1, \dots, K-1\} &
\end{align}

\begin{align}
\label{eq:HFGT_rotational_auxVar_initial}
\begin{bmatrix}
\omega_J
\end{bmatrix}[k=1]
= 
0
\end{align}

This explicit statement of the HFNMCF problem in the context of the rotational mechanical system shown in Fig. \ref{fig:illustrative_examples}c provides the following insights:
\begin{itemize}
\item Eq. \ref{eq:HFGT_rotational_dummy_target_function} shows that mechanical system does not have an objective function and is simply a set of simultaneous equations.  
\item Eq. \ref{eq:HFGT_rotational_continuity} is a matrix restatement of the continuity law in Eq. \ref{eq:rot_cont_1}.  Again, the torque balance on the ground place is redundant and therefore eliminated.  
\item Eq. \ref{eq:HFGT_rotational_priVars_source} imposes exogenous values on the torque due to the presence of torque sources.  
\item Eq. \ref{eq:HFGT_rotational_priVars_initial} is the initial condition on the rotational spring torque as a state variable.  
\item Eq. \ref{eq:HFGT_rotational_device_model_R} is a matrix restatement of the constitutive law for rotational dampers in Eq. \ref{eq:rot_const_3}.
\item Eq. \ref{eq:HFGT_rotational_device_model_L} is a matrix restatement of the constitutive law for rotational springs in Eq. \ref{eq:rot_const_2}.
\item Eq. \ref{eq:HFGT_rotational_device_model_C} is a matrix restatement of the constitutive law for rotational inertias in Eq. \ref{eq:rot_const_1}.
\item The selector matrices in Eq. \ref{eq:HFGT_rotational_priVars_source}-\ref{eq:HFGT_rotational_auxVar_initial} can be automatically produced from the HFGT toolbox \cite{Thompson:2023:ISC-JR02} for systems of arbitrary size.  
\item There are no equations that impose exogenous values on the angular velocity because there are no angular velocity sources.
\item Eq. \ref{eq:HFGT_rotational_auxVar_initial} is the initial condition on the disk angular velocity as a state variable. 
\item The compatibility laws stated in Eq. \ref{eq:rot_comp_1} and \ref{eq:rot_comp_2} are superfluous because all of the angular velocities have been stated in absolute terms relative to the ground reference frame rather than as angular velocity differences between points.  
\end{itemize}

Once the HFNMCF problem for the rotational mechanical system has been set up, it can be simulated straightforwardly and compared against the state space ODE model derived by linear graph and/or bond graph.  The following parameter values are chosen:  $J = 0.5 \ \text{kg} \cdot \text{m}^2$, $k = 2 \ \text{N} \cdot \text{m/rad}$, $b = 0.5 \ \text{N} \cdot \text{m} \cdot \text{s/rad}$, and $\tau_s = 1 \ \text{N} \cdot \text{m}$ (Step input torque). The simulation time $K = 15 $ seconds, and the time step $\Delta T = 0.1 $ seconds.  The HFNMCF results for the primary (torque) and auxiliary (angular velocity) decision variables are shown as solid lines in Fig. \ref{fig:HFGT_Rotational_mechanical_system_results}.  The associated state space ODE results are shown in dashed lines with embedded triangles. 

\begin{figure}[H]
\centering
\begin{subfigure}[b]{.45\textwidth}
\centering
\includegraphics[width=1\linewidth]{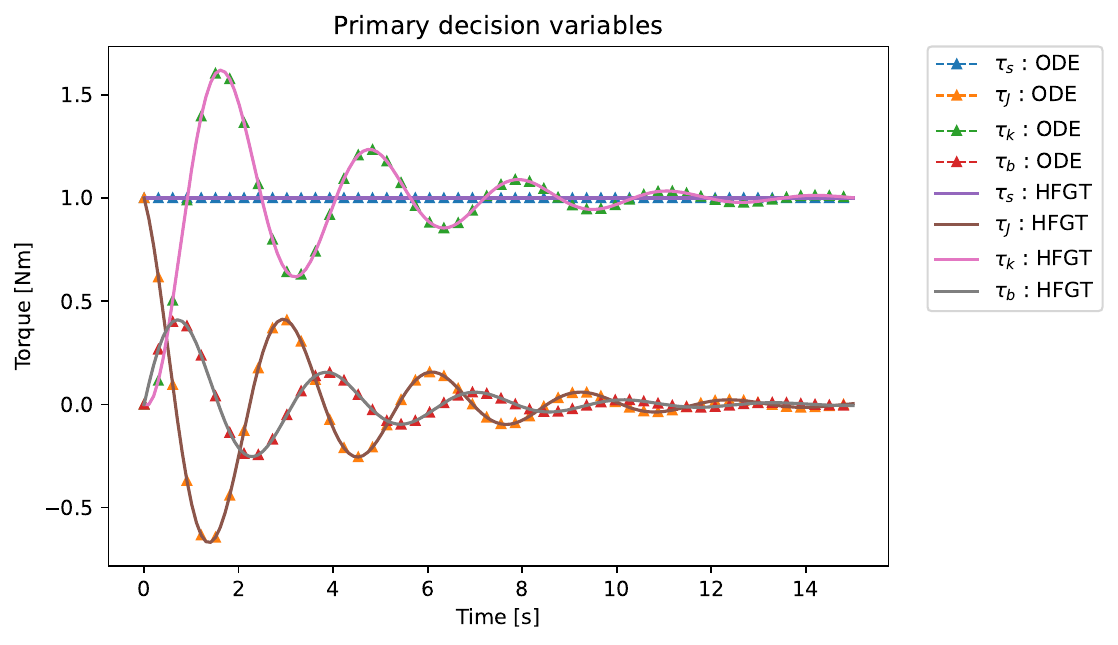} 
\caption{}
\label{fig:HFGT_Rotational_mechanical_X}
\end{subfigure}
\hfill
\begin{subfigure}[b]{.45\textwidth}
\centering
\includegraphics[width=1\linewidth]{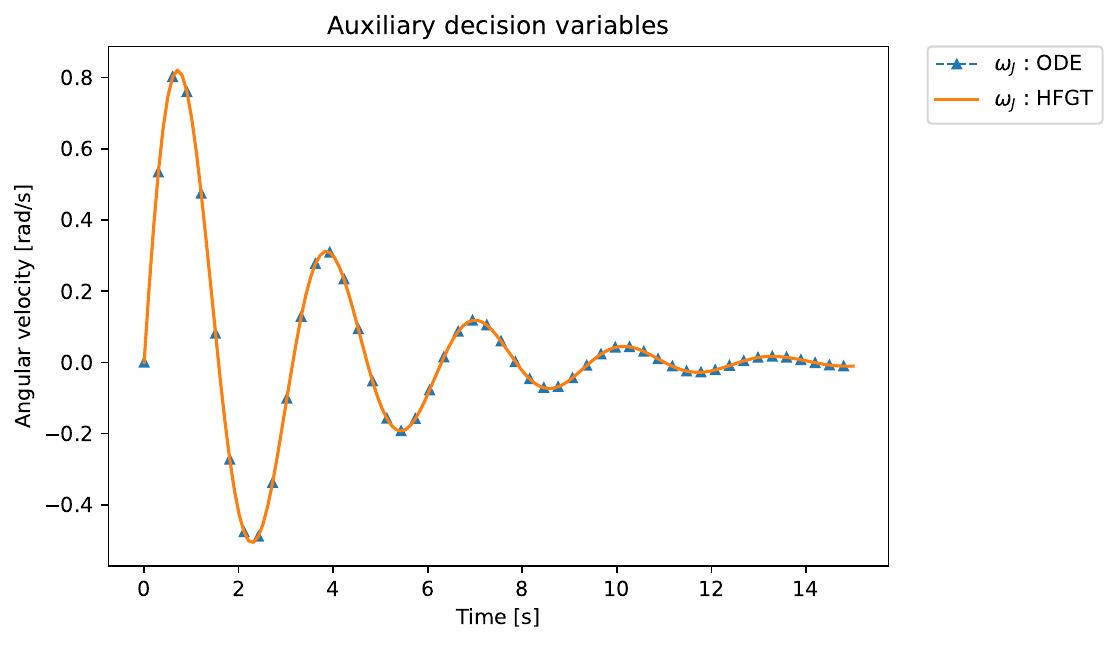}
\caption{}
\label{fig:HFGT_Rotational_mechanical_Y}
\end{subfigure}
\caption{(a) Time series of primary decision variables of the rotational mechanical system. (b) Time series of auxiliary decision variables of the rotational mechanical system.  Results from the HFNMCF problem are shown in solid lines. Results from the State Space ODE model are shown in dashed lines with embedded triangles.}
\label{fig:HFGT_Rotational_mechanical_system_results}
\end{figure}

\subsection{Fluidic System}
The first step is to recognize that the fluidic system shown in Fig. \ref{fig:illustrative_examples}d is first, a specialization into the fluidic domain, followed by an instantiation of the engineering system meta architecture in Fig. \ref{fig:LFESMetaArchitecture}.  Consequently, Defn. \ref{Defn:D1}-\ref{defn:capabilityCh7} are understood as follows. There are four fluidic points that have distinct absolute values of across-variables that serve as independent buffers: $P_1$, $P_{\mathcal{I} R_{1}}$, $P_2$, and $P_0$.  Additionally, there is one through-variable source $\dot{\cal V}_f$ that serves as a transformation resource.  Additionally, the transportation resources include generalized resistors $R_1$ and $R_2$, a generalized inductor $\cal I$, and generalized capacitors $C_1$ and $C_2$.  Fig. \ref{fig:LFESMetaArchitecture} shows that each of these transformation and transportation resources has exactly one system process; inject power with imposed through variable, dissipate power, store potential energy, and store kinetic energy.  The result is that each of these transformation and transportation resources introduces exactly one system capability with a primary through-variable and an auxiliary across-variable as attributes.  

In the next step, the fluidic system shown in Fig. \ref{fig:illustrative_examples}d is transformed into the engineering system net shown in Fig. \ref{fig:PN_Fluidic_system}. 
\begin{figure}[H]
\centering
\includegraphics[width=0.7\textwidth]{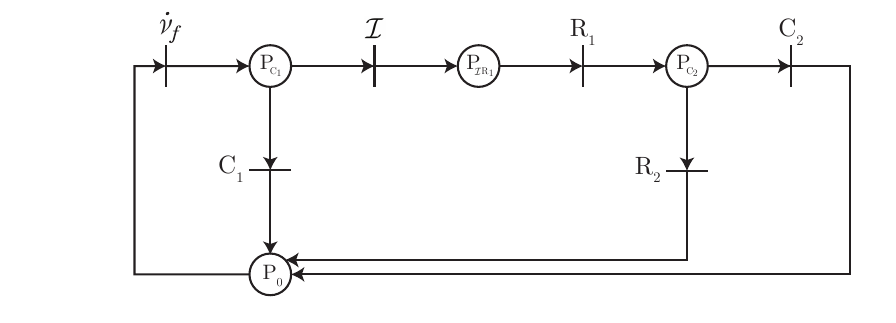} 
\caption{Petri-Net for the fluidic system shown in Figure \ref{fig:illustrative_examples}d}
\label{fig:PN_Fluidic_system}
\end{figure}
\noindent While it is possible to produce Fig. \ref{fig:PN_Fluidic_system} from the fluidic system in Fig. \ref{fig:illustrative_examples}d and the linear graph in Fig. \ref{X} by visual inspection; such an approach falls apart in systems with multiple energy domains or capabilities with multiple inputs and outputs.  Instead, the Engineering System Net and its state transition function is constructed according to Defn. \ref{Defn:ESN} and \ref{Defn:ESN-STF}.  Again, the hetero-functional graph theory toolbox \cite{Thompson:2023:ISC-JR02} automatically calculates the positive and negative hetero-functional incidence matrices from an XML input file that instantiates the information from Defn. \ref{Defn:D1} - \ref{defn:capabilityCh7}.   The Engineering System Net in Fig. \ref{fig:PN_Fluidic_system} shows the system buffers as places, the capabilities as transitions, and the incidence between them.  

In the third step, the hetero-functional network minimum cost flow problem is set up and solved.  More specifically, Eq. \ref{eq:dummy_objective_function}-\ref{eq:device_model_initial} are written out explicitly.  

\begin{align}
\text{minimize } Z = 0 \label{eq:HFGT_fluidic_dummy_target_function}
\end{align}
\begin{align}
\text{s.t. }
\label{eq:HFGT_fluidic_continuity}
\begin{bmatrix}
+1 & -1 & -1 & 0 & 0 & 0 \\
0 & +1 & 0 & -1 & 0 & 0 \\
0 & 0 & 0 & +1 & -1 & -1 
\end{bmatrix}\begin{bmatrix}
\dot{\cal V}_{f} \\
\dot{\cal V}_{\cal I} \\
\dot{\cal V}_{C_1} \\
\dot{\cal V}_{R_1} \\
\dot{\cal V}_{C_2} \\
\dot{\cal V}_{R_2}
\end{bmatrix}[k]\Delta T = 0   \quad \forall k \in \{1, \dots, K\}
\end{align}

\begin{align}
\label{eq:HFGT_fluidic_priVars_source}
\begin{bmatrix}
1 & 0 & 0 & 0 & 0 & 0
\end{bmatrix}
\begin{bmatrix}
\dot{\cal V}_{f} \\
\dot{\cal V}_{\cal I} \\
\dot{\cal V}_{C_1} \\
\dot{\cal V}_{R_1} \\
\dot{\cal V}_{C_2} \\
\dot{\cal V}_{R_2}
\end{bmatrix}[k]
=
1 \quad \forall k \in \{1, \dots, K\}
\end{align}

\begin{align} \label{eq:HFGT_fluidic_priVars_initial}
\begin{bmatrix}
0 & 1 & 0 & 0 & 0 & 0 
\end{bmatrix}
\begin{bmatrix}
\dot{\cal V}_{f} \\
\dot{\cal V}_{\cal I} \\
\dot{\cal V}_{C_1} \\
\dot{\cal V}_{R_1} \\
\dot{\cal V}_{C_2} \\
\dot{\cal V}_{R_2}
\end{bmatrix}[k=1]
= 
0
\end{align}

\begin{align}
\label{eq:HFGT_fluidic_device_model_R}\nonumber
\begin{bmatrix}
0 & 0 & 0 & 1 & 0 & 0 \\
0 & 0 & 0 & 0 & 0 & 1 \\
\end{bmatrix}
\begin{bmatrix}
\dot{\cal V}_{f} \\
\dot{\cal V}_{\cal I} \\
\dot{\cal V}_{C_1} \\
\dot{\cal V}_{R_1} \\
\dot{\cal V}_{C_2} \\
\dot{\cal V}_{R_2}
\end{bmatrix}[k]
=
\begin{bmatrix}
\frac{1}{R_1} & 0\\
0 & \frac{1}{R_2}
\end{bmatrix}
\begin{bmatrix}
0 & 0 & 0 & 1 & 0 & 0 \\
0 & 0 & 0 & 0 & 0 & 1 \\
\end{bmatrix}
\begin{bmatrix}
-1 & 0 & 0 \\
+1 & -1 & 0 \\
+1 & 0 & 0 \\
0 & +1 & -1 \\
0 & 0 & +1 \\
0 & 0 & +1 \\
\end{bmatrix}
\begin{bmatrix}
P_1 \\
P_{\mathcal{I} R_{1}} \\
P_2
\end{bmatrix}[k] & \\
\forall k \in \{1, \dots, K\} & 
\end{align}
\begin{align}
\label{eq:HFGT_fluidic_device_model_L}\nonumber
\begin{bmatrix}
0 & 1 & 0 & 0 & 0 & 0
\end{bmatrix}
\left(
\begin{bmatrix}
\dot{\cal V}_{f} \\
\dot{\cal V}_{\cal I} \\
\dot{\cal V}_{C_1} \\
\dot{\cal V}_{R_1} \\
\dot{\cal V}_{C_2} \\
\dot{\cal V}_{R_2}
\end{bmatrix}[k+1]
-
\begin{bmatrix}
\dot{\cal V}_{f} \\
\dot{\cal V}_{\cal I} \\
\dot{\cal V}_{C_1} \\
\dot{\cal V}_{R_1} \\
\dot{\cal V}_{C_2} \\
\dot{\cal V}_{R_2}
\end{bmatrix}[k]
\right)
=
\begin{bmatrix}
\frac{1}{\cal I}
\end{bmatrix}
\begin{bmatrix}
0 & 1 & 0 & 0 & 0 & 0
\end{bmatrix}
\begin{bmatrix}
-1 & 0 & 0 \\
+1 & -1 & 0 \\
+1 & 0 & 0 \\
0 & +1 & -1 \\
0 & 0 & +1 \\
0 & 0 & +1 \\
\end{bmatrix}
\begin{bmatrix}
P_1 \\
P_{\mathcal{I} R_{1}} \\
P_2
\end{bmatrix}[k]\Delta T &\\
\forall k \in \{1, \dots, K-1\} & 
\end{align}
\begin{align}
\label{eq:HFGT_fluidic_device_model_C}\nonumber
\begin{bmatrix}
0 & 0 & 1 & 0 & 0 & 0 \\
0 & 0 & 0 & 0 & 1 & 0 \\
\end{bmatrix}
\begin{bmatrix}
\dot{\cal V}_{f} \\
\dot{\cal V}_{\cal I} \\
\dot{\cal V}_{C_1} \\
\dot{\cal V}_{R_1} \\
\dot{\cal V}_{C_2} \\
\dot{\cal V}_{R_2}
\end{bmatrix}[k] \Delta T
= 
\begin{bmatrix}
C_1 & 0 \\
0 & C_2
\end{bmatrix}
\begin{bmatrix}
0 & 0 & 1 & 0 & 0 & 0 \\
0 & 0 & 0 & 0 & 1 & 0 \\
\end{bmatrix}
\begin{bmatrix}
-1 & 0 & 0 \\
+1 & -1 & 0 \\
+1 & 0 & 0 \\
0 & +1 & -1 \\
0 & 0 & +1 \\
0 & 0 & +1 \\
\end{bmatrix}
\left(
\begin{bmatrix}
P_1 \\
P_{\mathcal{I} R_{1}} \\
P_2
\end{bmatrix}[k+1]
-
\begin{bmatrix}
P_1 \\
P_{\mathcal{I} R_{1}} \\
P_2
\end{bmatrix}[k]
\right) &\\
\forall k \in \{1, \dots, K-1\} &
\end{align}
\begin{align}
\label{eq:HFGT_fluidic_auxVar_initial}
\begin{bmatrix}
1 & 0 & 0 \\
0 & 0 & 1 
\end{bmatrix}
\begin{bmatrix}
P_1 \\
P_{\mathcal{I} R_{1}} \\
P_2
\end{bmatrix}[k=1]
= 
\begin{bmatrix}
0 \\
0
\end{bmatrix}
\end{align}

This explicit statement of the HFNMCF problem in the context of the fluidic system shown in Fig. \ref{fig:illustrative_examples}d provides the following insights:
\begin{itemize}
\item Eq. \ref{eq:HFGT_fluidic_dummy_target_function} shows that the fluidic system does not have an objective function and is simply a set of simultaneous equations.  
\item Eq. \ref{eq:HFGT_fluidic_continuity} is a matrix restatement of the continuity laws in Eq. \ref{eq:fluid_cont_1} - \ref{eq:fluid_cont_3}.  Again, the volumetric flow rate balance on the ground place is redundant and therefore eliminated.  
\item Eq. \ref{eq:HFGT_fluidic_priVars_source} imposes exogenous values on the volumetric flow rate due to the presence of a volumetric flow rate source.
\item Eq. \ref{eq:HFGT_fluidic_priVars_initial} is the initial condition on the pipe inductance volumetric flow rate as a state variable.  
\item Eq. \ref{eq:HFGT_fluidic_device_model_R} is a matrix restatement of the constitutive law for the elements' fluidic resistances in Eq. \ref{eq:fluid_const_4}, and \ref{eq:fluid_const_5}.
\item Eq. \ref{eq:HFGT_fluidic_device_model_L} is a matrix restatement of the constitutive law for the elements' fluidic inductances in Eq. \ref{eq:fluid_const_3}.
\item Eq. \ref{eq:HFGT_fluidic_device_model_C} is a matrix restatement of the constitutive law for the tanks in Eq. \ref{eq:fluid_const_1} and \ref{eq:fluid_const_2}.
\item The selector matrices in Eq. \ref{eq:HFGT_fluidic_priVars_initial}-\ref{eq:HFGT_fluidic_device_model_C} can be automatically produced from the HFGT toolbox \cite{Thompson:2023:ISC-JR02}for systems of arbitary size.  
\item There are no equations that impose exogenous values on the pressure because there are no pressure sources.
\item Eq. \ref{eq:HFGT_fluidic_auxVar_initial} is the initial condition on the tank pressures as state variables. 
\item The compatibility laws stated in Eq. \ref{eq:fluid_comp_1} and \ref{eq:fluid_comp_2} are superfluous because all of the pressures have been stated in absolute terms relative to the reference ambient pressure rather than as pressure differences between fluidic points.  
\end{itemize}

Once the HFNMCF problem for the fluidic system has been set up, it can be simulated straightforwardly and compared against the state space ODE model derived by linear graph and/or bond graph.  The following parameter values are chosen:  $R_1 = 2 \ \text{s/m}^2$, $R_2 = 1 \ \text{s/m}^2$, $C_1 = 0.02 \ \text{m}^3$, $C_2 = 0.05 \ \text{m}^3$, $\mathcal{I} = 2 \ \text{N} \cdot \text{s}^2/\text{m}^5$, and $\dot{v}_f = 1 \ \text{m}^3/\text{s}$ (Step input volumetric flow rate). The simulation time $K = 10 $ seconds, and the time step $\Delta T = 0.02 $ seconds.  The HFNMCF results for the primary (volumetric flow rate) and auxiliary (pressure) decision variables are shown as solid lines in Fig. \ref{fig:HFGT_Fluidic_system_results}.  The associated state space ODE results are shown in dashed lines with embedded triangles.


\begin{figure}[H]
\centering
\begin{subfigure}[b]{.45\textwidth}
\centering
\includegraphics[width=1\linewidth]{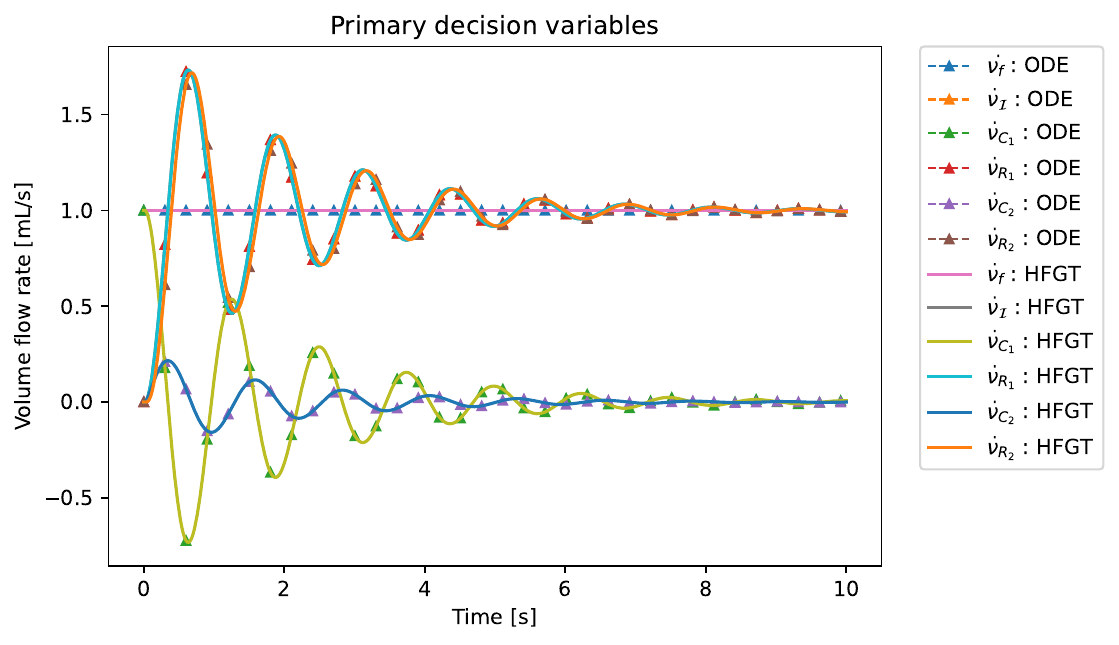} 
\caption{}
\label{fig:HFGT_fluidic_X}
\end{subfigure}
\hfill
\begin{subfigure}[b]{.45\textwidth}
\centering
\includegraphics[width=1\linewidth]{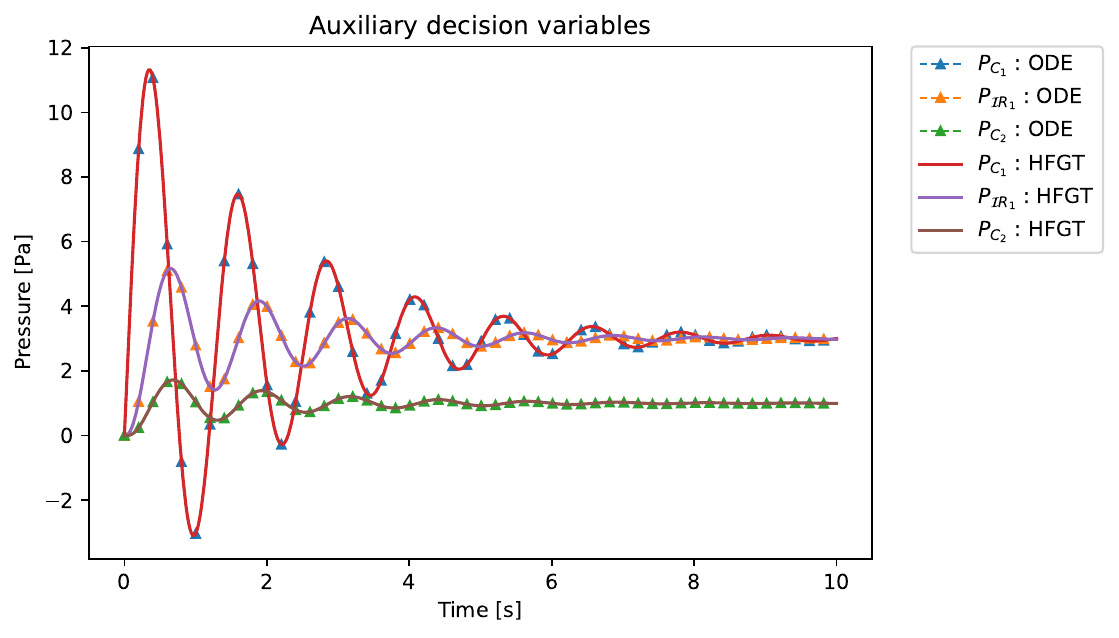}
\caption{}
\label{fig:HFGT_Fluidic_Y}
\end{subfigure}
\caption{(a) Time series of primary decision variables of the fluidic system. (b) Time series of auxiliary decision variables of the fluidic system.  Results from the HFNMCF problem are shown in solid linears. Results from the State Space ODE model are shown in dashed lines with embedded triangles.}
\label{fig:HFGT_Fluidic_system_results}
\end{figure}

\subsection{Thermal System}
The first step is to recognize that the thermal system shown in Fig. \ref{fig:illustrative_examples}e is first, a specialization into the thermal domain, followed by an instantiation of the engineering system meta architecture in Fig. \ref{fig:LFESMetaArchitecture}.  Consequently, Defn. \ref{Defn:D1}-\ref{defn:capabilityCh7} are understood as follows.  There are three thermal points that have distinct absolute values of across-variables that serve as independent buffers: $T_{C_h}$, $T_{C_i}$, and $T_0$.  Additionally, there is one through-variable source $\dot{Q_s}$ that serves as a transformation resource.  Additionally, the transportation resources include generalized resistors $R_i$ and $R_h$, and generalized capacitors $C_i$ and $C_h$.  Fig. \ref{fig:LFESMetaArchitecture} shows that each of these transformation and transportation resources has exactly one system process; inject power with imposed through variable, dissipate power, and store (thermal) potential energy.  The result is that each of these transformation and transportation resources introduces exactly one system capability with a primary through-variable and an auxiliary across-variable as attributes.  

In the next step, the thermal system shown in Fig. \ref{fig:illustrative_examples}e is transformed into the engineering system net shown in Fig. \ref{fig:PN_Thermal_system}. 
\begin{figure}
\centering
\includegraphics[width=0.45\textwidth]{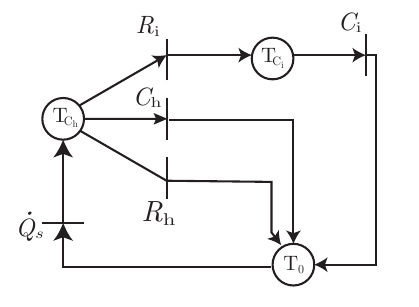} 
\caption{Petri-Net for the thermal system shown in Fig. \ref{fig:illustrative_examples}e}
\label{fig:PN_Thermal_system}
\end{figure}
\noindent While it is possible to produce Fig. \ref{fig:PN_Thermal_system} from the thermal system in Fig. \ref{fig:illustrative_examples}e, it's analogous electrical circuit in Fig. \ref{fig:thermal_system_electrical_circuit} and/or its linear graph in Fig. \ref{fig:linear_graph_thermal_system} by visual inspection; such an approach falls apart in systems with multiple energy domains or capabilities with multiple inputs and outputs.  Instead, the Engineering System Net and its state transition function is constructed according to Defn. \ref{Defn:ESN} and \ref{Defn:ESN-STF}.  Note that the hetero-functional graph theory toolbox \cite{Thompson:2023:ISC-JR02} can automatically calculate the positive and negative hetero-functional incidence matrices from an XML input file that instantiates the information from Defn. \ref{Defn:D1} - \ref{defn:capabilityCh7}.   The Engineering System Net in Fig. \ref{fig:PN_Thermal_system} shows the system buffers as places, the capabilities as transitions, and the incidence between them.  

In the third step, the hetero-functional network minimum cost flow problem is set up and solved.  More specifically, Eq. \ref{eq:dummy_objective_function}-\ref{eq:device_model_initial} are written out explicitly.  

\begin{align}
\text{minimize } Z = 0 \label{eq:HFGT_thermal_dummy_target_function}
\end{align}
\begin{align}
\text{s.t. }
\label{eq:HFGT_thermal_continuity}
\begin{bmatrix}
+1 & -1 & 0 & -1 & -1\\
0 & +1 & -1 & 0 & 0\\
\end{bmatrix}\begin{bmatrix}
\dot{Q}_{s} \\
\dot{Q}_{R_i} \\
\dot{Q}_{C_i} \\
\dot{Q}_{C_h} \\
\dot{Q}_{R_h}
\end{bmatrix}[k]\Delta T = 0   \quad \forall k \in \{1, \dots, K\}
\end{align}

\begin{align}
\label{eq:HFGT_thermal_priVars_source}
\begin{bmatrix}
1 & 0 & 0 & 0 & 0
\end{bmatrix}
\begin{bmatrix}
\dot{Q}_{s} \\
\dot{Q}_{R_i} \\
\dot{Q}_{C_i} \\
\dot{Q}_{C_h} \\
\dot{Q}_{R_h}
\end{bmatrix}[k]
= \begin{bmatrix}
1
\end{bmatrix}  \quad \forall k \in \{1, \dots, K\}
\end{align}

\begin{align}
\label{eq:HFGT_thermal_device_model_R}\nonumber
\begin{bmatrix}
0 & 1 & 0 & 0 & 0 \\
0 & 0 & 0 & 0 & 1 \\
\end{bmatrix}
\begin{bmatrix}
\dot{Q}_{s} \\
\dot{Q}_{R_i} \\
\dot{Q}_{C_i} \\
\dot{Q}_{C_h} \\
\dot{Q}_{R_h}
\end{bmatrix}[k]
=
\begin{bmatrix}
\frac{1}{R_i} & 0\\
0 & \frac{1}{R_h}
\end{bmatrix}
\begin{bmatrix}
0 & 1 & 0 & 0 & 0 \\
0 & 0 & 0 & 0 & 1 \\
\end{bmatrix}
\begin{bmatrix}
-1 & 0 \\
+1 & -1 \\
0 & +1 \\
+1 & 0 \\
+1 & 0 \\
\end{bmatrix}
\begin{bmatrix}
T_{C_h} \\
T_{C_i}
\end{bmatrix}[k] & \\
\forall k \in \{1, \dots, K\} & 
\end{align}

\begin{align}
\label{eq:HFGT_thermal_device_model_C}\nonumber
\begin{bmatrix}
0 & 0 & 1 & 0 & 0 \\
0 & 0 & 0 & 1 & 0
\end{bmatrix}
\begin{bmatrix}
\dot{Q}_{s} \\
\dot{Q}_{R_i} \\
\dot{Q}_{C_i} \\
\dot{Q}_{C_h} \\
\dot{Q}_{R_h}
\end{bmatrix}[k] \Delta T
= 
\begin{bmatrix}
C_i & 0 \\
0 & C_h
\end{bmatrix}
\begin{bmatrix}
0 & 0 & 1 & 0 & 0 \\
0 & 0 & 0 & 1 & 0
\end{bmatrix}
\begin{bmatrix}
-1 & 0 \\
+1 & -1 \\
0 & +1 \\
+1 & 0 \\
+1 & 0
\end{bmatrix}
\left(
\begin{bmatrix}
T_{C_h} \\
T_{C_i}
\end{bmatrix}[k+1]
-
\begin{bmatrix}
T_{C_h} \\
T_{C_i}
\end{bmatrix}[k]
\right) &\\
\forall k \in \{1, \dots, K-1\} &
\end{align}

\begin{align}
\label{eq:HFGT_thermal_auxVar_initial}
\begin{bmatrix}
1 & 0 \\
0 & 1 
\end{bmatrix}
\begin{bmatrix}
T_{C_h} \\
T_{C_i}
\end{bmatrix}[1]
= 
\begin{bmatrix}
0 \\
0
\end{bmatrix}
\end{align}

This explicit statement of the HFNMCF problem in the context of the thermal system shown in Fig. \ref{fig:illustrative_examples}e provides the following insights:
\begin{itemize}
\item Eq. \ref{eq:HFGT_thermal_dummy_target_function} shows that the thermal system does not have an objective function and is simply a set of simultaneous equations.  
\item Eq. \ref{eq:HFGT_thermal_continuity} is a matrix restatement of the continuity laws in Eq. \ref{eq:thermal_cont_1} and \ref{eq:thermal_cont_2}.  Again, the heat balance on the ground place is redundant and therefore eliminated.  
\item Eq. \ref{eq:HFGT_thermal_priVars_source} imposes exogenous values on the heat flow rate due to the presence of heat flow sources.  
\item Due to the absence of generalized inductors in thermal systems, there are no equations for their initial condition.
\item Eq. \ref{eq:HFGT_thermal_device_model_R} is a matrix restatement of the constitutive law for element's thermal resistances in Eq. \ref{eq:thermal_const_3} and \ref{eq:thermal_const_4}.
\item There are no equations that restate the constitutive laws for thermal inductors because there are no generalized inductors in thermal systems.
\item Eq. \ref{eq:HFGT_thermal_device_model_C} is a matrix restatement of the constitutive law for the thermal capacitance of system elements in Eq. \ref{eq:thermal_const_1} and \ref{eq:thermal_const_2}.
\item The selector matrices in Eq. \ref{eq:HFGT_thermal_priVars_source}-\ref{eq:HFGT_thermal_auxVar_initial} can be automatically produced from the HFGT toolbox \cite{Thompson:2023:ISC-JR02} for systems of arbitrary size.  
\item There are no equations that impose exogenous values on the temperature because there are no temperature sources.
\item Eq. \ref{eq:HFGT_thermal_auxVar_initial} is the initial condition on the thermal capacitors' temperature as the state variables. 
\item The compatibility laws stated in Eq. \ref{eq:thermal_comp_1} and \ref{eq:thermal_comp_2} are superfluous because all of the temperatures have been stated in absolute terms relative to $0 \deg C$ as a reference temperature rather than as temperature differences between points.  
\end{itemize}

Once the HFNMCF problem for the fluidic system has been set up, it can be simulated straightforwardly and compared against the state space ODE model derived by linear graph and/or bond graph.  The following parameter values are chosen:
$R_i = 0.5 \ \text{K/W}$, $R_h = 0.2 \ \text{K/W}$, $C_i = 1 \ \text{J/K}$, $C_h = 2 \ \text{J/K}$, $\dot{Q}_S = 1 \ \text{W}$ \ \text{(Step input heat flow rate)}.The simulation time $K = 5 $ seconds, and the time step $\Delta T = 0.1 $ seconds. The HFNMCF results for the primary (heat flow rate) and auxiliary (temperature) decision variables are shown as solid lines in Fig. \ref{fig:HFGT_Thermal_system_results}.  The associated state space ODE results are shown in dashed lines with embedded triangles.


\begin{figure}[H]
\centering
\begin{subfigure}[b]{.45\textwidth}
\centering
\includegraphics[width=1\linewidth]{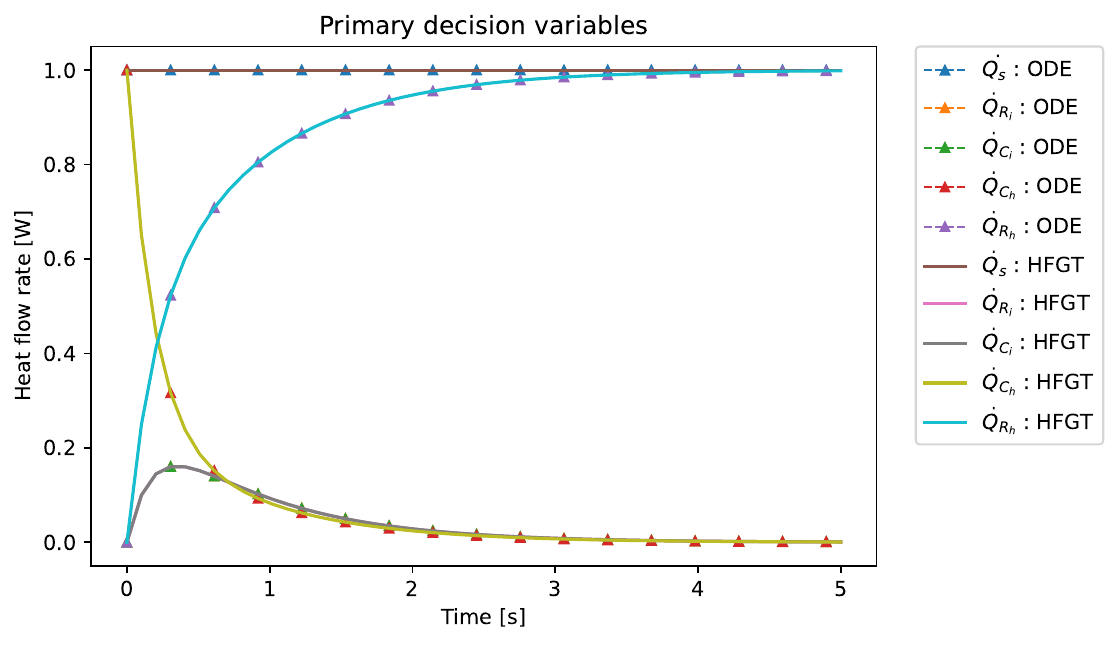} 
\caption{}
\label{fig:HFGT_Thermal_X}
\end{subfigure}
\hfill
\begin{subfigure}[b]{.45\textwidth}
\centering
\includegraphics[width=1\linewidth]{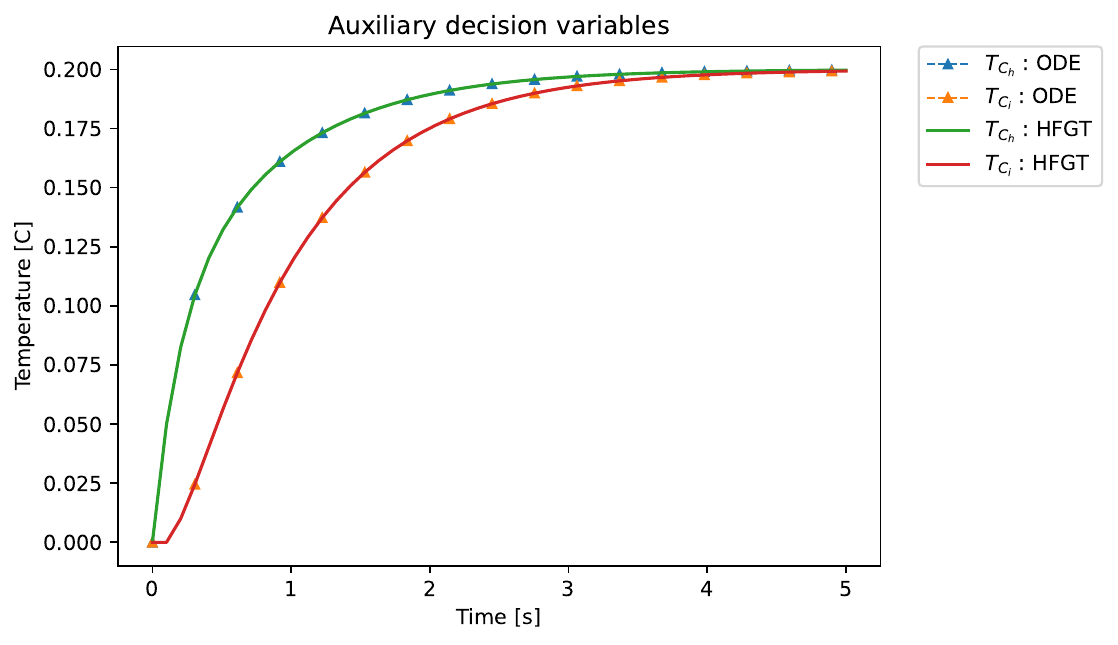}
\caption{}
\label{fig:HFGT_Thermal_Y}
\end{subfigure}
\caption{(a) Time series of primary decision variables of the thermal system. (b) Time series of auxiliary decision variables of the thermal system.  Results from the HFNMCF problem are shown in solid lines. Results from the State Space ODE model are shown in dashed lines with embedded triangles.}
\label{fig:HFGT_Thermal_system_results}
\end{figure}

\subsection{Multi-Energy System}

The first step is to recognize that the electro-mechanical system shown in Fig. \ref{fig:illustrative_examples}f is first, a specialization into the multi-system domain, followed by an instantiation of the engineering system meta-architecture in Fig. \ref{fig:LFESMetaArchitecture}.  Consequently, Defn. \ref{Defn:D1}-\ref{defn:capabilityCh7} are understood as follows.  There are four electrical points that have distinct absolute values of across-variables that serve as independent buffers: $V_S$, $V_{RL}$, $V_{LM}$, $V_0$. Also, there are two mechanical points that have distinct absolute values of across-variables that serve as independent buffers: $\omega_J$ and $\omega_0$. Additionally, there is one across variable source $V_S$ that serves as a transformation resource.  Additionally, the transportation resources include generalized resistors $R$ and $B$, a generalized inductor $L$, a generalized capacitor $J$, and a generalized transformer (i.e. motor) with motor constant $1/k_a$.  Fig. \ref{fig:LFESMetaArchitecture} shows that each of these transformation and transportation resources has exactly one system process; inject power with imposed through variable, dissipate power, store potential energy, and store kinetic energy.  The result is that each of these transformation and transportation resources introduces exactly one system capability with a primary through-variable and an auxiliary across-variable as attributes.  

In the next step, the electro-mechanical system shown in Fig. \ref{fig:illustrative_examples}a is transformed into the engineering system net shown in Fig. \ref{fig:PN_Electrical_system}. 
\begin{figure}[H]
\centering
\includegraphics[width=0.7\textwidth]{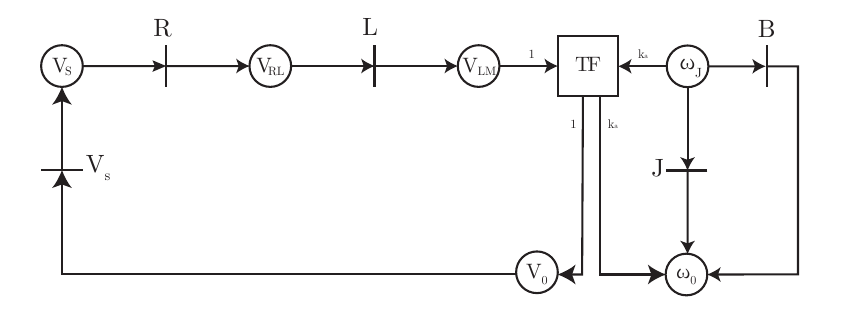} 
\caption{Petri-Net for the electro-mechanical system shown in Figure \ref{fig:illustrative_examples}f}
\label{fig:PN_Multi_system}
\end{figure}
\noindent Notice that the transformer (i.e. motor) acts as a transition with two inputs and two outputs: the motor current in and out, and the motor torque in and out.  Also note that the through variable associated with the motor torque no longer appears explicitly, and instead appears implicitly in the form of the arc weights label with the motor constant $k_a$.  While Fig. \ref{fig:PN_Multi_system} strongly resembles the electro-mechanical diagram in Fig. \ref{fig:illustrative_examples}f and the linear graph in Fig. \ref{fig:linear_graph_multi_system}, it can also be produced directly from HFGT.  The Engineering System Net and its state transition function is constructed according to Defn. \ref{Defn:ESN} and \ref{Defn:ESN-STF}.  Note that the hetero-functional graph theory toolbox \cite{Thompson:2023:ISC-JR02} can automatically calculate the positive and negative hetero-functional incidence matrices from an XML input file that instantiates the information from Defn. \ref{Defn:D1} - \ref{defn:capabilityCh7}.   The Engineering System Net in Fig. \ref{fig:PN_Multi_system} shows the system buffers as places, the capabilities as transitions, and the incidence between them.

In the third step, the hetero-functional network minimum cost flow problem is set up and solved.  More specifically, Eq. \ref{eq:dummy_objective_function}-\ref{eq:device_model_initial} are written out explicitly.  

\begin{align}
\text{minimize } Z = 0 \label{eq:HFGT_multi_dummy_target_function}
\end{align}
\begin{align}
\text{s.t. }
\label{eq:HFGT_multi_continuity}
\begin{bmatrix}
+1 & -1 & 0 & 0 & 0 & 0 \\
0 & +1 & -1 & 0 & 0 & 0 \\
0 & 0 & +1 & -1 & 0 & 0\\
0 & 0 & 0 & +\frac{1}{k_a} & -1 & -1\\
\end{bmatrix}\begin{bmatrix}
i_{V_S} \\
i_{R} \\
i_{L} \\
i_{m}\\
\tau_{B}\\
\tau_{J} \\ 
\end{bmatrix}[k]\Delta T = 0   \quad \forall k \in \{1, \dots, K\}
\end{align}

\begin{align} \label{eq:HFGT_multi_priVars_initial}
\begin{bmatrix}
0 & 0 & 1 & 0 & 0 & 0
\end{bmatrix}
\begin{bmatrix}
i_{V_S} \\
i_{R} \\
i_{L} \\
i_{m}\\
\tau_{B}\\
\tau_{J} \\ 
\end{bmatrix}[k=1]
= 
\begin{bmatrix}
0
\end{bmatrix}
\end{align}
\begin{align}
\label{eq:HFGT_multi_device_model_R}\nonumber
\begin{bmatrix}
0 & 1 & 0 & 0 & 0 & 0\\
0 & 0 & 0 & 0 & 1 & 0\\
\end{bmatrix}
\begin{bmatrix}
i_{V_S} \\
i_{R} \\
i_{L} \\
i_{m}\\
\tau_{B}\\
\tau_{J} \\  
\end{bmatrix}[k]
=
\begin{bmatrix}
\frac{1}{R} & 0\\
0 & B
\end{bmatrix}
\begin{bmatrix}
0 & 1 & 0 & 0 & 0 & 0\\
0 & 0 & 0 & 0 & 1 & 0\\
\end{bmatrix}
\begin{bmatrix}
-1 & 0 & 0 & 0 \\
+1 & -1 & 0 & 0 \\
0 & +1 & -1 & 0 \\
0 & 0 & +1 & -\frac{1}{k_a} \\
0 & 0 & 0 & +1 \\
0 & 0 & 0 & +1 \\
\end{bmatrix}
\begin{bmatrix}
V_S \\
V_{RL} \\
V_{LM} \\
\omega_J \\
\end{bmatrix}[k] & \\
\forall k \in \{1, \dots, K\} & 
\end{align}
\begin{align}
\label{eq:HFGT_multi_device_model_L}\nonumber
\begin{bmatrix}
0 & 0 & 1 & 0 & 0 & 0
\end{bmatrix}
\left(
\begin{bmatrix}
i_{V_S} \\
i_{R} \\
i_{L} \\
i_{m}\\
\tau_{B}\\
\tau_{J} \\ 
\end{bmatrix}[k+1]
-
\begin{bmatrix}
i_{V_S} \\
i_{R} \\
i_{L} \\
i_{m}\\
\tau_{B}\\
\tau_{J} \\ 
\end{bmatrix}[k]
\right)
=
\begin{bmatrix}
\frac{1}{L}
\end{bmatrix}
\begin{bmatrix}
0 & 0 & 1 & 0 & 0 & 0
\end{bmatrix}
\begin{bmatrix}
-1 & 0 & 0 & 0 \\
+1 & -1 & 0 & 0\\
0 & +1 & -1 & 0 \\
0 & 0 & +1 & +k_a\\
0 & 0 & 0 & +1 \\
0 & 0 & 0 & +1\\
\end{bmatrix}
\begin{bmatrix}
V_S \\
V_{RL} \\
V_{LM} \\
\omega_J\\
\end{bmatrix}[k]\Delta T &\\
\forall k \in \{1, \dots, K-1\} & 
\end{align}
\begin{align}
\label{eq:HFGT_multi_device_model_C}\nonumber
\begin{bmatrix}
0 & 0 & 0 & 0 & 0 & 1 
\end{bmatrix}
\begin{bmatrix}
i_{V_S} \\
i_{R} \\
i_{L} \\
i_{m}\\
\tau_{B}\\
\tau_{J} \\ 
\end{bmatrix}[k] \Delta T
= 
\begin{bmatrix}
J
\end{bmatrix}
\begin{bmatrix}
0 & 0 & 0 & 0 & 0 & 1
\end{bmatrix}
\begin{bmatrix}
-1 & 0 & 0 & 0\\
+1 & -1 & 0 & 0\\
0 & +1 & -1 & 0\\
0 & 0 & +1 & -\frac{1}{k_a}\\
0 & 0 & 0 & +1\\
0 & 0 & 0 & +1 
\end{bmatrix}
\left(
\begin{bmatrix}
V_S \\
V_{RL} \\
V_{LM} \\
\omega_J \\
\end{bmatrix}[k+1]
-
\begin{bmatrix}
V_S \\
V_{RL} \\
V_{LM} \\
\omega_J \\
\end{bmatrix}[k]
\right) &\\
\forall k \in \{1, \dots, K-1\} &
\end{align}


\begin{align}
\label{eq:HFGT_multi_device_model_transformer2}
\begin{bmatrix}
0 & 0 & 1 & -\frac{1}{k_a}
\end{bmatrix}
\begin{bmatrix}
V_S \\
V_{RL} \\
V_{LM} \\
\omega_J\\
\end{bmatrix}[k]=
\begin{bmatrix}
0
\end{bmatrix} \quad \forall k \in \{1, \dots, K\}
\end{align}

\begin{align}
\label{eq:HFGT_multi_aux_Var_source}
\begin{bmatrix}
1 & 0 & 0 & 0
\end{bmatrix}
\begin{bmatrix}
V_S \\
V_{RL} \\
V_{LM} \\
\omega_J\\
\end{bmatrix}[k]
= \begin{bmatrix}
1
\end{bmatrix}  \quad \forall k \in \{1, \dots, K\}
\end{align}
\begin{align}
\label{eq:HFGT_multi_auxVar_initial}
\begin{bmatrix}
0 & 0 & 0 & 1 
\end{bmatrix}
\begin{bmatrix}
V_S \\
V_{RL} \\
V_{LM} \\
\omega_J\\
\end{bmatrix}[k=1]
= 
\begin{bmatrix}
0
\end{bmatrix}
\end{align}

This explicit statement of the HFNMCF problem in the context of the electro-mechanical system shown in Fig. \ref{fig:illustrative_examples}f provides the following insights:
\begin{itemize}
\item Eq. \ref{eq:HFGT_multi_dummy_target_function} shows that an electro-mechanical system does not have an objective function and is simply a set of simultaneous equations.  
\item Eq. \ref{eq:HFGT_multi_continuity} is a matrix restatement of the continuity laws in Eq. \ref{eq:elecmech_cont_1} - \ref{eq:elecmech_cont_3}.  
The first three rows apply a current balance on each electrical place, while the last applies a torque balance on each mechanical place.  Note that the motor constant $k_a$ serves to transform the motor torque $i_m$ into the motor torque.  This serves to combine the transformer's constitutive law in Eq. \ref{eq:elecmech_const_6} with the torque balance in Eq. \ref{eq:elecmech_cont_1}.  Note that Eq. \ref{eq:HFGT_multi_continuity} introduces an additional matrix row to account for the current provided by the voltage source $i_{V_S}$.  While this variable is not required in the linear graph and bond graph methodologies, the HFGT derivation requires across and through variables for all capabilities.  Again, the current balance on the electrical ground and the torque balance on the mechanical ground are redundant and therefore eliminated.  
\item There are no equations that impose exogenous values on the currents and angular velocities because there are no associated sources. 
\item Eq. \ref{eq:HFGT_multi_priVars_initial} is the initial condition on the inductor current as a state variable.  
\item Eq. \ref{eq:HFGT_multi_device_model_R} is a matrix restatement of the constitutive law for the resistor and the physical damper in Eq. \ref{eq:elecmech_const_3} and \ref{eq:elecmech_const_4}.
\item Eq. \ref{eq:HFGT_multi_device_model_L} is a matrix restatement of the constitutive law for the inductor in Eq. \ref{eq:elecmech_const_2}.
\item Eq. \ref{eq:HFGT_multi_device_model_C} is a matrix restatement of the constitutive law for rotating disk in Eq. \ref{eq:elecmech_const_1}.
\item Eq. \ref{eq:HFGT_multi_device_model_transformer2} is a matrix restatement of the constitutive law for the auxiliary (i.e. across) decision variable for transformer in Eq. \ref{eq:elecmech_const_5}.  Again, the transformer's second constitutive law has already been incorporated in the context of \ref{eq:HFGT_multi_continuity}.
\item The selector matrices in Eq. \ref{eq:HFGT_multi_priVars_initial}-\ref{eq:HFGT_multi_aux_Var_source} can be automatically produced from the HFGT toolbox \cite{Thompson:2023:ISC-JR02} for systems of arbitrary size.  
\item Eq. \ref{eq:HFGT_multi_aux_Var_source} imposes exogenous values on the voltages due to the presence of voltage sources.  
\item Eq. \ref{eq:HFGT_multi_auxVar_initial} is the initial condition on the disk angular velocity as a state variable. 
\item The compatibility laws stated in Eq. \ref{eq:elecmech_comp_1} - \ref{eq:elecmech_comp_3} are superfluous because all of the voltages and angular velocities have been stated in absolute terms relative to the ground rather than as voltage/angular velocity differences between points.  
\end{itemize}

Once the HFNMCF problem for the fluidic system has been set up, it can be simulated straightforwardly and compared against the state space ODE model derived by linear graph and/or bond graph.  The following parameter values are chosen:
$R = 1 \ \Omega$, $L = 0.01 \ \text{mH}$, $J = 5 \ \text{kg} \cdot \text{m}^2$, $B = 0.1 \ \text{N} \cdot \text{m} \cdot \text{s/rad}$, and $k_a = 0.1$, and $V_s = 1 \ \text{V}$ (Step input voltage).  The simulation time $K = 0.3 $ seconds, and the time step $\Delta T = 0.001 $ seconds.  The HFNMCF results for the primary (current and torque) and auxiliary (voltage and angular velocity) decision variables are shown as solid lines in Fig. \ref{fig:HFGT_Multi_system_results}.  The associated state space ODE results are shown in dashed lines with embedded triangles.

\begin{figure}[H]
\centering
\begin{subfigure}[b]{.45\textwidth}
\centering
\includegraphics[width=1\linewidth]{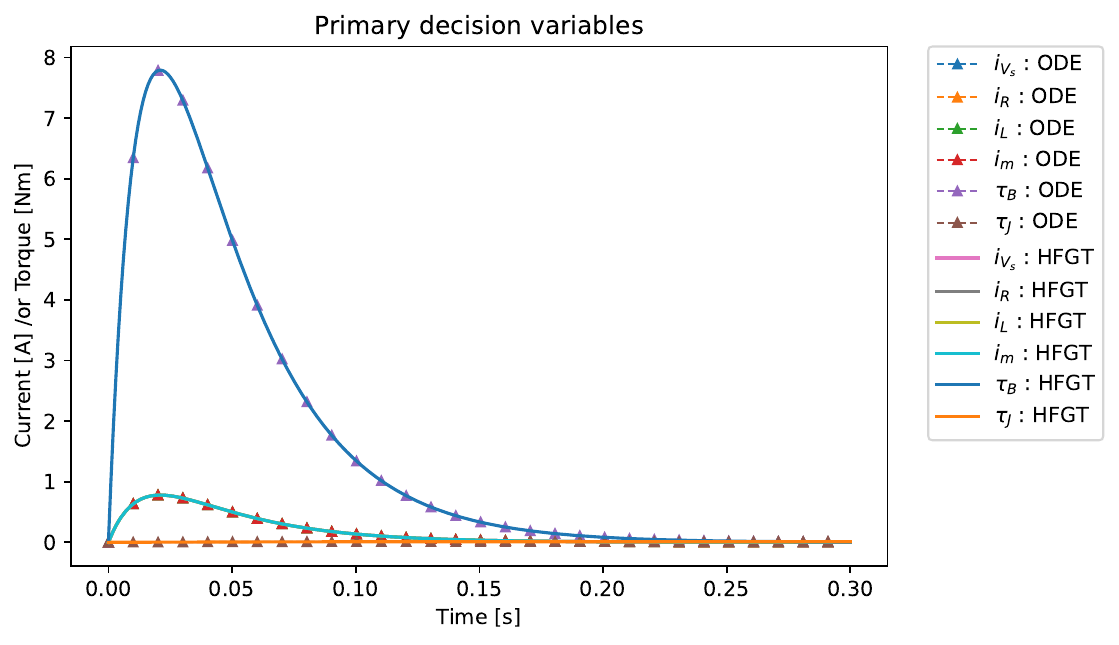} 
\caption{}
\label{fig:HFGT_Multi_X}
\end{subfigure}
\hfill
\begin{subfigure}[b]{.45\textwidth}
\centering
\includegraphics[width=1\linewidth]{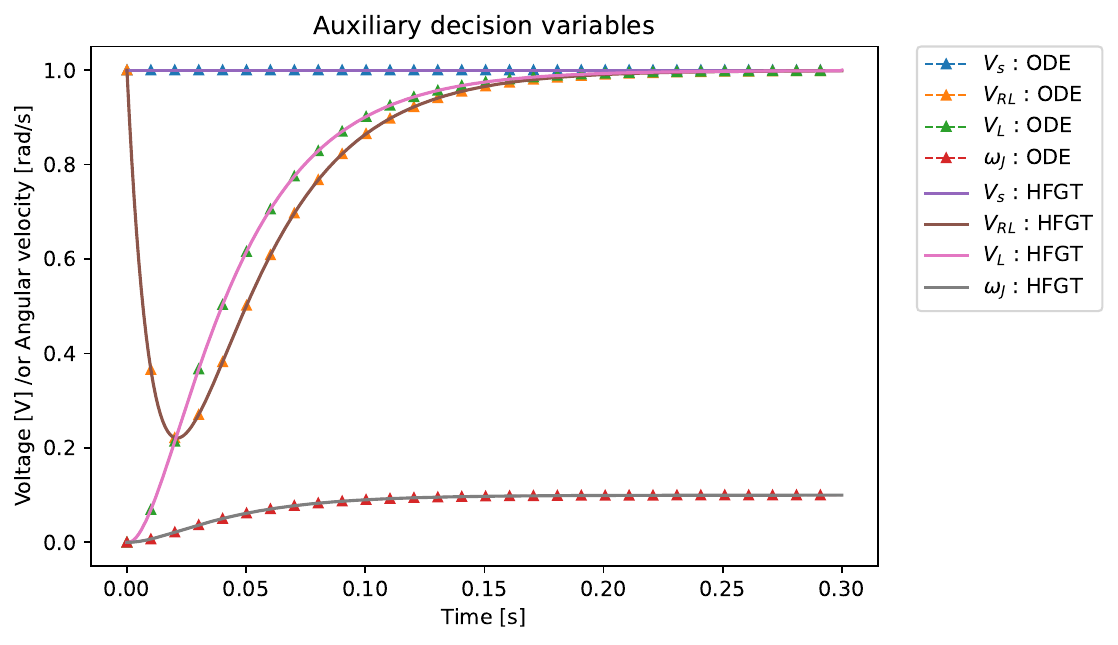}
\caption{}
\label{fig:HFGT_Multi_Y}
\end{subfigure}
\caption{(a) Time series of primary decision variables of the electro-mechanical system. (b) Time series of auxiliary decision variables of the electro-mechanical system.  Results from the HFNMCF problem are shown in solid lines. Results from the State Space ODE model are shown in dashed lines with embedded triangles.}
\label{fig:HFGT_Multi_system_results}
\end{figure}

\section{On the Generality of Hetero-functional Graphs}
\label{sec:On_generality_of_HFGT}
The previous section concretely demonstrated the relationship between linear graphs, bond graphs, and hetero-functional graphs on the six illustrative examples depicted in Fig \ref{fig:illustrative_examples}.  Each time, the result of the HFNMCF problem was numerically equivalent to the simulation of the associated state space ODE model.  In effect, the linear graphs and bond graph methodologies derive continuity laws (i.e. 0-Junction in Eulerian view systems/ 1-Junction in Lagrangian view systems), constitutive laws, and compatibility laws (i.e. 1-Junction in Eulerian view system/ 0-Junction in Lagrangian view systems) and then use algebraic manipulations to simplify them into a state space ODE model.  Hetero-functional graph theory, quite similarly, includes the continuity laws in the engineering system net state transition function, includes the constitutive laws in the device models, and eliminates the need for compatibility models by virtue of its choice of reference frame and then states these laws as equations that are solved simultaneously and numerically (without manual algebraic manipulations).  Consequently, while the numerical evidence in the previous section is compelling from a pedagogical perspective, in reality, the pattern of results points to two more general results.  

\begin{thm}\label{Thm:LG_HFG}
Given an arbitrary linear graph composed of 
\begin{enumerate*}
\item across variable sources, 
\item through variable sources, 
\item D-Type elements, 
\item A-Type elements, 
\item T-Type elements, 
\item generalized transformers, and 
\item generalized gyrators, 
\end{enumerate*}
organized in an arbitrary topology, and a minimal set of initial conditions on the associated state variables, the solution of its associated state space ODE model is equivalent to the solution of a specialized instance of the HFNMCF problem.  
\end{thm}

\begin{proof}
The linear graph's state space ODE model is equivalent to a set of simultaneous differential algebraic equations composed of continuity laws, constitutive laws, and compatibility laws.  Furthermore, the compatibility laws can be eliminated entirely with an algebraic change of variable that measures all across variables relative to a ground reference frame.  Furthermore, the differential algebraic equation form of the continuity and constitutive laws can be stated as algebraic equations with a sufficiently small choice of the discrete-time step $\Delta T$.  Next, the continuity laws can be algebraically recast into the form stated in Eq. \ref{eq:HFGT_continuity}.  Additionally, the constitutive laws can be algebraically recast into the form stated in Eq. \ref{eq:device_model} where 
\begin{itemize}
\item generalized resistors laws take the form in Eq. \ref{eq:R_element_aux_eq},
\item generalized T-Type element laws take the form in Eq. \ref{eq:L_element_aux_eq},
\item generalized A-Type element laws take the form in Eq. \ref{eq:C_element_aux_eq}, 
\item generalized transformers laws take the form in Eq. \ref{eq:transformer_element_aux_eq},
\item generalized gyrator laws take the form in Eq. \ref{eq:gyrator_element_aux_eq}.
\end{itemize}
Additionally, the across-variable sources are described by Eq. \ref{eq:HFGT_priVar_source}, and the through-variable sources are described by Eq. \ref{eq:device_model_source}.  Next, the initial conditions on through-type state variables are described by Eq. \ref{eq:HFGT_priVar_initial} and the initial conditions on across-type state variables are described by Eq. \ref{eq:device_model_initial}.  Next, the solution of the simultaneous equations in Eq. \ref{eq:HFGT_continuity}-\ref{eq:device_model_initial} can be recast as an optimization program with the null objective function in Eq. \ref{eq:dummy_objective_function} and these same equations as constraints.  Finally, the optimization problem stated in Eq. \ref{eq:dummy_objective_function}-\ref{eq:device_model_initial} is a special case of the HFNMCF problem in Eq. \ref{Eq:ObjFunc1}-\ref{Eq:DevicModels2} under the conditions described in the beginning of Sec. \ref{subsec:HFGT_by_example}.  
\end{proof}

\begin{thm}\label{Thm:BG_HFG}
Given an arbitrary bond graph composed of 
\begin{enumerate*}
\item effort sources, 
\item flow sources, 
\item generalized resistors, 
\item generalized capacitors, 
\item generalized inductors, 
\item generalized transformers, and 
\item generalized gyrators, 
\end{enumerate*}
organized in an arbitrary topology, and a minimal set of initial conditions on the associated state variables, the solution of its associated state space ODE model is equivalent to the solution of a specialized instance of the HFNMCF problem.  
\end{thm}

\begin{proof}
The bond graph's state space ODE model is equivalent to a set of simultaneous differential algebraic equations composed of 0-junction laws, 1-junction laws, and constitutive laws. Furthermore, 0-junction laws, 1-junction laws, and bond graph system elements can be transformed 1-to-1 to form the continuity laws, compatibility laws, and system elements of the corresponding linear graph.  As elaborated in Sec. \ref{subsec:bond_graphs}) and more specifically, 
\begin{itemize}
\item In systems with an Eulerian view (e.g. electrical systems, fluidic systems, thermal systems)
\begin{itemize}
    \item effort sources are equivalent to across variable sources,
    \item flow sources are equivalent to through variable sources,
    \item generalized resistors are equivalent to D-Type variables,
    \item generalized capacitors are equivalent to A-type variables,
    \item generalized inductors are equivalent to T-Type variables,
    \item 0-Junction laws are equivalent to constitutive laws,
    \item 1-junction laws are equivalent to compatibility laws.
    \end{itemize} 
\item In systems with the Lagrangian view (e.g. Mechanical system)
\begin{itemize}
    \item effort sources are equivalent to through-variable sources,
    \item flow sources are equivalent to across-variable sources,
    \item generalized resistors are equivalent to D-Type variables,
    \item generalized capacitors are equivalent to T-Type variables,
    \item generalized inductors are equivalent to A-type variables,
    \item 0-Junction laws are equivalent to compatibility laws,
    \item 1-junction laws are equivalent to continuity laws.
\end{itemize}
\item A generalized gyrator in a bond graph is equivalent to either a generalized transformer or generalized gyrator in the linear graph methodology depending on the choice of system on either side of the element.  
\item Similarly, a generalized transformer in a bond graph is equivalent to either a generalized transformer or generalized gyrate in the linear graph methodology depending on the choice of system on either side of the element.  
\end{itemize}
Because an arbitrary bond graph model can be transformed to an equivalent linear graph model, then by Theorem \ref{Thm:LG_HFG}, the ODE state space model derived by the bond graph is a special case of the HFNMCF problem.
\end{proof}

\section{Conclusion and Future Work}\label{sec:Conclusion}
This paper relates hetero-functional graphs to linear graphs and bond graphs.  Despite having completely different theoretical origins, it demonstrates the former is a generalization of the latter two.  To facilitate the comparison, each of the three modeling techniques is described and then compared conceptually.   These three descriptions reveal that hetero-functional graphs, linear graphs, and bond graphs have completely different ontologies so finding a direct relationship through a purely abstract treatment is difficult.  Instead, the paper focuses the discussion concretely on six example systems:  (a) an electrical system, (b) a translational mechanical system, (c) a rotational mechanical system, (d) a fluidic system, (e) a thermal system, and (f) a multi-energy (electro-mechanical) system.  Each of these systems was modeled with hetero-functional graphs, linear graphs, and bond graphs to reveal that dynamic simulation models produced by each of these modeling techniques result in numerically equivalent results.  Finally, this concrete numerical evidence provides significant intuition and insight that overcomes the ontological differences between these three types of graph approaches.  The paper proves mathematically that hetero-functional graphs are a formal generalization of both linear graphs and bond graphs.  

This abstract and highly general result is significant for several reasons.  First, linear graphs and bond graphs have a much longer history in the literature and have produced extensive theoretical and practical results.   Until now, these contributions have been theoretically divorced from the hetero-functional graph theory literature.  A direct relationship between hetero-functional, linear, and bond graphs facilitates the cross-pollination of theoretical and practical results between these approaches.  For example, bond graphs are often used to study system causality\cite{Karnopp:1990:00}; where such an analysis has been elusive in hetero-functional graphs.  More practically, the well-known modeling and simulation tool, Modelica\cite{Fritzson:2011:00,Fritzson:2014:00}, was originally developed on a bond-graph foundation but its connection to hetero-functional graphs has not been made.  Looking further ahead, and as exposited in the introduction, the 21st century is creating challenges that demand a deep understanding of the structure and behavior of systems-of-systems.  This requires modeling approaches with an open, rather than closed, set of modeling primitives that spans systems of fundamentally different functions.  These systems must address operands of energy, matter, information, money, and living organisms, and not just energy.  It also requires modeling approaches that can handle continuous time, discrete time, and discrete event dynamics.  This paper reveals that HFGT can provide such analytical flexibility and extensibility without losing the rich tradition of graph-based modeling that originated in the previous century.

\bibliographystyle{IEEEtran}
\bibliography{LIINESLibrary,LIINESPublications,LG_BG_HFGT_References}

\end{document}